\documentclass{lmcs}
\pdfoutput=1

\usepackage{lastpage}
\lmcsdoi{17}{3}{19}
\lmcsheading{}{\pageref{LastPage}}{}{}%
{Jul.~05,~2018}{Aug.~13,~2021}{}

\usepackage[utf8]{inputenc}

\usepackage{silence}
\usepackage{mathpartir}
\usepackage{stmaryrd}
\usepackage{amssymb}
\usepackage{mathtools}

\usepackage{thm-restate}
\usepackage{thmtools}

\usepackage{subcaption}

\usepackage{tikz}
\usetikzlibrary{positioning}
\usetikzlibrary{automata}
\usetikzlibrary{angles}

\usepackage{algorithm2e}
\usepackage{setspace}

\usepackage{hyperref}
\usepackage[nameinlink]{cleveref}

\newcommand{\naturals}{\mathbb{N}}

\newcommand{\setcompr}[2]{\left\{#1\;:\;#2\right\}}
\newcommand{\set}[1]{\left\{#1\right\}}

\newcommand{\angl}[1]{\left\langle#1\right\rangle}

\newcommand*{\rmbrace}{|\mskip-4mu\}}
\newcommand*{\lmbrace}{\{\mskip-4mu|}
\newcommand*{\mset}[1]{\lmbrace#1\rmbrace}
\newcommand{\M}{\mathbb{M}}
\newcommand{\mempty}{\boxslash}

\newcommand{\ltr}[1]{\mathtt{#1}}

\newcommand{\lp}[1]{\mathbf{#1}}
\newcommand{\SP}{\ensuremath{\mathsf{SP}}}

\newcommand{\pipe}{\;\;|\;\;}
\newcommand{\terms}{\ensuremath{\mathcal{T}}}
\newcommand{\sacc}{\ensuremath{\mathcal{F}}}
\newcommand{\pdepth}[1]{\ensuremath{\mathsf{depth}(#1)}}

\newcommand{\sem}[1]{%
  \mathchoice%
  {\left\llbracket#1\right\rrbracket} 
  {\llbracket#1\rrbracket} 
  {\llbracket#1\rrbracket} 
  {\llbracket#1\rrbracket} 
}

\newcommand{\supp}{\pi}
\newcommand{\runrel}{\rightarrow}
\newcommand{\arun}[2][]{\mathrel{\raisebox{-3pt}{$\xrightarrow[#1]{#2}$}}}
\newcommand{\epstrans}{\mathrel{\arun{1}_A}}

\newcommand{\sderiv}{\delta_\Sigma}
\newcommand{\pderiv}{\gamma_\Sigma}
\newcommand{\sarun}[1]{\arun{#1}_\Sigma}
\newcommand{\ssupp}{\supp_\Sigma}

\newcommand{\At}{\mathsf{At}}

\newcommand{\sst}[1]{{#1}^\mathbf{s}}
\newcommand{\pst}[1]{{#1}^\mathbf{p}}

\makeatletter
\let\c@thm\undefined%
\makeatother

\theoremstyle{plain}
\declaretheorem[name=Theorem,numberwithin=section]{thm}
\newtheorem{clm}[thm]{Claim}
\newtheorem{lem}[thm]{Lemma}
\newtheorem{cor}[thm]{Corollary}
\newtheorem{prop}[thm]{Proposition}
\newtheorem{fact}[thm]{Fact}

\theoremstyle{thmC}
\newtheorem{lemC}[thm]{Lemma}

\theoremstyle{definition}
\newtheorem{exa}[thm]{Example}
\newtheorem{defi}[thm]{Definition}

\crefname{prop}{Proposition}{Propositions}
\crefname{thm}{Theorem}{Theorems}
\crefname{clm}{Claim}{Claims}
\crefname{lem}{Lemma}{Lemmas}
\crefname{lemC}{Lemma}{Lemmas}
\crefname{cor}{Corollary}{Corollaries}
\crefname{exa}{Example}{Examples}
\crefname{defi}{Definition}{Definitions}
\crefname{fact}{Fact}{Facts}

\crefname{section}{Section}{Sections}
\crefname{figure}{Figure}{Figures}

\begin{document}

\title[Equivalence checking for weak bi-Kleene algebra]%
{Equivalence checking for weak bi-Kleene algebra\rsuper{*}}
\titlecomment{\lsuper{*}An earlier version of this paper was published at CONCUR'17~\cite{kappe-brunet-luttik-silva-zanasi-2017}}

\author[T.~Kapp\'{e}]{Tobias Kapp\'{e}\rsuper{a}}
\author[P.~Brunet]{Paul Brunet\rsuper{b}}
\author[B.~Luttik]{Bas Luttik\rsuper{c}}
\author[A.~Silva]{Alexandra Silva\rsuper{b}}
\author[F.~Zanasi]{Fabio Zanasi\rsuper{b}\texorpdfstring{\vspace{-7mm}}{}}

\address{\lsuper{a}Cornell University, Ithaca, New York, USA}
\email{tkappe@cornell.edu}
\address{\lsuper{b}University College London, United Kingdom}
\address{\lsuper{c}Eindhoven University of Technology, The Netherlands}

\thanks{%
  T.~Kapp\'{e} was partially supported by the ERC Starting Grant 679127 (ProFoundNet) and DARPA grant HR001120C0107 (Pronto).
  A.~Silva was partially supported by the ERC Starting Grant 679127 (ProFoundNet) and a Leverhulme Prize (PLP--2016--129).
  P.~Brunet acknowledges support from EPSRC grant n.\ EP/R006865/1.
  F.~Zanasi acknowledges support from EPSRC grant n.\ EP/R020604/1.
}


\begin{abstract}
  Pomset automata are an operational model of \emph{weak bi-Kleene algebra}, which describes programs that can \emph{fork} an execution into parallel threads, upon completion of which execution can \emph{join} to resume as a single thread.
  We characterize a fragment of pomset automata that admits a decision procedure for language equivalence.
  Furthermore, we prove that this fragment corresponds precisely to \emph{series-rational expressions}, i.e., rational expressions with an additional operator for bounded parallelism.
  As a consequence, we obtain a new proof that equivalence of series-rational expressions is decidable.
\end{abstract}

\maketitle

\section{Introduction}
Kleene's theorem states the correspondence between the operational world of automata and the denotational world of expressions, on formal languages~\cite{kleene-1956}.
This famous discovery has proven pivotal to establish later results, such as Kozen's axiomatisation of equivalence of rational expressions~\cite{kozen-1994}, as well as to transpose the application of algorithms from a denotational to an operational setting --- for instance, one can leverage Hopcroft and Karp's algorithm for finite automata~\cite{hopcroft-karp-1971} to decide equivalence of rational expressions.

In spite of their simplicity, finite automata and rational expressions provide valuable tools in analyzing the behaviour of sequential programs~\cite{kozen-1996}.
The behavioural patterns of present-day programs, however, are not limited to sequential scenarios, where each event either strictly precedes or succeeds all others.
Indeed, reasoning about programs that run on multi-core processors requires us to adapt our descriptions such that two events need not be strictly ordered, but instead may occur \emph{in parallel}.
The study of \emph{concurrent Kleene algebra}~\cite{hoare-moeller-struth-wehrman-2009}, in the broadest sense, is concerned with extending techniques from rational expressions and finite automata to reason about systems that include parallelism.

We propose \emph{pomset automata} as an operational model for a fragment of concurrent Kleene algebra known as \emph{weak bi-Kleene algebra}, which describes programs where an execution may \emph{fork} into parallel computations, to \emph{join} the results of those computations later on, resuming execution.
The language semantics of these automata is given by sets of \emph{partially ordered multisets}, or \emph{pomsets}.
The first main contribution is a proof that language equivalence of states is decidable for the class of \emph{fork-acyclic} finite pomset automata.
The second main contribution is a Kleene theorem, which shows that this same fragment corresponds precisely to the denotational model of bi-Kleene algebra, known as \emph{series-rational expressions}, or \emph{sr-expressions} for short~\cite{lodaya-weil-2000,laurence-struth-2014} --- that is, rational expressions extended with parallel composition.
This correspondence then yields a decision procedure for deciding equivalence of series-rational expressions, via pomset automata.

\medskip
In \cref{section:related-work}, we discuss related work; we go over the necessary background regarding pomsets in \cref{section:preliminaries}.
In \cref{section:pomset-automata}, we introduce pomset automata and their semantics.
Additionally, we introduce a (structural) restriction on pomset automata, defining the class of \emph{fork-acyclic} pomset automata.
In \cref{section:language-equivalence}, we develop an algorithm for checking language equivalence of a subclass of fork-acyclic pomset automata; in \cref{sec:well-structured-automata}, we extend this procedure to fork-acyclic pomset automata in general.
In \cref{section:expressions-to-automata}, we show how to obtain a pomset automaton that recognizes the language of a series-rational expression; conversely, in \cref{section:automata-to-expressions}, we show how to obtain an equivalent series-rational expression from a finite and fork-acyclic pomset automaton.
We list directions for further work in \cref{section:further-work}.

For the sake of self-containment, we include proofs of formal claims that are not cited.
Routine proofs are delegated to the appendices to ensure brevity.

\section{Related work}%
\label{section:related-work}
There exist three pomset-based operational models for sr-expressions in the literature.
\emph{Branching automata} were pioneered by Lodaya and Weil~\cite{lodaya-weil-2000}.
These are non-deterministic finite automata enriched with two additional types of transition to mediate forking and joining of computation.
No decision procedure for language equivalence of branching automata is known.
Branching automata also come equipped with a reverse construction, which shows how to obtain an equivalent sr-expression from a particular class of branching automata.
The difference is that the description of this class involves a \emph{semantic} condition, i.e., makes a statement about the pomsets that can be accepted by a state in the branching automaton, whereas fork-acyclic pomset automata are described in purely structural terms.

Another model, proposed by Jipsen and Moshier~\cite{jipsen-moshier-2016} based on~\cite{lodaya-weil-2000} and also called \emph{branching automata}, is given by non-deterministic finite automata, enriched with a relation that specifies where computation may be joined, provided all threads can be traced back to a given common state.
In some sense, pomset automata are a dual to this model, in that they specify where execution can be forked, after which all threads can be joined at a given state after termination.
As far as we can tell, there is no known decision procedure for language equivalence of branching automata.
Branching automata in the style of Jipsen and Moshier also come with a translation back to sr-expressions; being based on~\cite{lodaya-weil-2000}, this construction inherits the semantic description of automata to which it can be applied.

Petri nets, specifically safe Petri nets, can also be used to describe the behaviours modelled by sr-expressions~\cite{brunet-pous-struth-2017,lodaya-ranganayakulu-rangarajan-2003}.
The advantage of this approach is that it allows one to use results from Petri net theory to study sr-expressions.
Furthermore, particularly in the case of~\cite{brunet-pous-struth-2017}, one can leverage the encoding of sr-expressions into Petri nets to develop a decision procedure for equivalence of sr-expressions, as well as a more general type of equivalence that allows threads to sequentialise (corresponding to concurrent Kleene algebra proper).
However, it should be noted that safe Petri nets can express a form of concurrency that is strictly more general than the type of concurrency that can be described by concurrent Kleene algebra~\cite{grabowski-1981}.
As such, converting a safe Petri net to an equivalent sr-expression necessarily discards some behaviour~\cite{lodaya-ranganayakulu-rangarajan-2003}.

The operational models discussed above associate an automaton or Petri net to an sr-expression by induction on the structure of the expression, using a Thompson-style translation~\cite{thompson-1968}.
In contrast, our expressions-to-automata translation generalizes Antimirov's construction~\cite{antimirov-1996}, and thus allows the operational representation to be constructed lazily.
This is particularly beneficial for algorithms that explore the state space of automata step-by-step, as it prevents them from computing the entire state space.

In~\cite{baeten-luttik-muller-vantilburg-2016}, Baeten et al.\ give an operational semantics to sr-expressions in terms of (non-deterministic) transitions systems, by interpreting the parallel composition as interleaving, which obviates the need for a parallel thread construction.
They show that there exist transition systems that are not bisimilar to any sr-expression, and characterise the fragment of systems for which such an sr-expression does exist.
A full Kleene theorem is recovered when they extend sr-expressions with interaction.

Our algorithm to check language equivalence in pomset automata was inspired by the work of Laurence and Struth~\cite{laurence-struth-2014} on a more general form of sr-expressions; specifically, the idea of checking language equality by computing the \emph{atoms} of languages is due to them.

\medskip

This paper is an extension of a CONCUR'17 paper~\cite{kappe-brunet-luttik-silva-zanasi-2017}; in comparison, the present work contains a more generalised presentation of pomset automata that allows for non-determinism, which enables us to define the class of \emph{well-structured} pomset automata.
We then use this structural restriction to guide the construction of a decision procedure for language equivalence of fork-acyclic pomset automata.
Moreover, the syntactic derivatives are now presented in the style of Antimirov~\cite{antimirov-1996}, rather than Brzozowski~\cite{brzozowski-1964}.

Another closely related paper is~\cite{kappe-brunet-luttik-silva-zanasi-2019}, in which we generalised the Kleene theorem of~\cite{kappe-brunet-luttik-silva-zanasi-2017} to \emph{series-parallel regular expressions}, which include the parallel variant of the Kleene star known as \emph{parallel star}.
For this generalisation to work, one needs to loosen the notion of fork-acyclicity to the strictly more liberal \emph{well-nestedness}; crucially, the definition of well-nested pomset automata relies on the presentation of pomset automata found in~\cite{kappe-brunet-luttik-silva-zanasi-2017}.
Furthermore, in~\cite{kappe-brunet-luttik-silva-zanasi-2019}, we show that language equivalence of pomset automata in general is undecidable --- this justifies the fact that the decision procedure presented in this paper requires fork-acyclicity.

\section{Preliminaries}%
\label{section:preliminaries}

We fix a finite set of symbols $\Sigma$, referred to as the \emph{alphabet}.
When $S$ and $T$ are sets, we write $T^S$ for the set of functions from $S$ to $T$.
We write $2^S$ for the set of subsets of $S$, which can be identified with the functions from $S$ to the two-element set $2 = \set{0, 1}$.

A \emph{multiset} over a set $S$ is a ``subset'' of $S$ where elements may occur more than once; more formally, it is a function $\phi: S \to \naturals$.
Finite multisets are denoted using double braces, e.g., $\phi = \mset{1, 1}$ is the multiset over $\mathbb{N}$ where $\phi(n) = 2$ if $n = 1$, and $\phi(n) = 0$ otherwise.
In the following, we fix multisets $\phi, \psi$ over $S$.
When $s \in S$, we use $s \in \phi$ as a shorthand for $\phi(s) \neq 0$.
We say that $\phi$ is \emph{finite} if there are finitely many $s \in S$ such that $s \in \phi$.
If $\phi$ is finite, then the \emph{size} of $\phi$, denoted $|\phi|$, is given by $\sum_{s \in S} \phi(s)$.
We write $\M(S)$ for the \emph{set of finite multisets} over $S$.
We denote the empty multiset by $\mempty$, and write $\phi \sqcup \psi$ for the \emph{disjoint union} of $\phi$ and $\psi$, where $(\phi \sqcup \psi)(s) = \phi(s) + \psi(s)$.
For instance, $\mset{0} \sqcup \mset{0, 1} = \mset{0, 0, 1}$.

A \emph{non-deterministic finite automaton} (\emph{NFA}) is a tuple $A = \angl{Q, \delta, F}$ where $Q$ is a finite set of \emph{states}, with $F \subseteq Q$ the \emph{accepting states}, and $\delta: Q \times \Sigma \to 2^Q$ is a function.
The \emph{language} of $q \in Q$ in $A$, denoted $L_A(q)$, is the set of words $w = \ltr{a}_1\cdots\ltr{a}_n$ such that there exist $q = q_0, \dots, q_n \in Q$ where $q_{i+1} \in \delta(q_i, \ltr{a}_{i+1})$ for $0 \leq i < n$, and $q_n \in F$.

\subsection{Pomsets}

We commonly represent an execution of a program as a \emph{word} over some finite alphabet $\Sigma$.
In such a word, each \emph{position} corresponds to an \emph{event} in the execution whose \emph{name} is given by the symbol on that position; events are ordered according to their positions.
For instance, if $\Sigma = \set{\ltr{a}, \ltr{b}, \ltr{c}}$, then the word $\ltr{abca}$ represents an execution where an event of type $\ltr{a}$ occurs, followed by events of type $\ltr{b}$, $\ltr{c}$ and $\ltr{a}$, in that order.

To represent an execution of a program with concurrency, we need to relax this model to allow a partial order on events.
For instance, a concurrent program may execute an event of type $\ltr{a}$ before forking into two threads that execute events of type $\ltr{b}$ and $\ltr{c}$ respectively, after which the threads join to perform an event of type $\ltr{a}$.
Note now, in this execution, there is no ordering of the events labelled by $\ltr{b}$ and $\ltr{c}$ --- they are \emph{concurrent}.

The model most commonly found in the literature to account for such executions was proposed independently by Winkowski~\cite{winkowski-1977} and Pratt~\cite{pratt-1982}, and studied extensively by Winkowski~\cite{winkowski-1979} and Grabowski~\cite{grabowski-1981}; we adopt Pratt's terminology.
Defining the model requires some patience, as the indirection between \emph{events} and their \emph{names}, which is given implicitly in words by \emph{positions} and their \emph{symbols} is slightly tricky to generalize.

\begin{defi}%
  \label{definition:labelled-poset}
  A \emph{labelled partially ordered set} (\emph{labelled poset}) is a tuple $\lp{u} = \angl{S_\lp{u}, \leq_\lp{u}, \lambda_\lp{u}}$ consisting of a \emph{carrier} set $S_\lp{u}$, a partial order $\leq_\lp{u}$ on $S_\lp{u}$ and a \emph{labelling} function $\lambda_\lp{u}: S_\lp{u} \to \Sigma$.
\end{defi}

For technical reasons, we adopt the convention that the carrier of a labelled poset is a subset of $\mathbb{N}$; under this convention, the collection of labelled posets is a proper set.

The definition above gets us close to where we need to be.
For instance, the example execution above can be represented by the labelled poset $\angl{S_\lp{u}, \leq_\lp{u}, \lambda_\lp{u}}$, in which
\begin{mathpar}
  S_\lp{u} = \set{1, 2, 3, 4}
  \and
  1 \leq_\lp{u} 2 \leq_\lp{u} 4
  \and
  1 \leq_\lp{u} 3 \leq_\lp{u} 4
  \and
  \lambda_\lp{u} = \set{1 \mapsto \ltr{a}, 2 \mapsto \ltr{b}, 3 \mapsto \ltr{c}, 4 \mapsto \ltr{a}}
\end{mathpar}

However, if the event labelled by $\ltr{b}$ were represented by $5$ instead of $2$ (adjusting $\leq_\lp{u}$ and $\lambda_\lp{u}$ accordingly), then this new labelled poset would still represent the same execution; in some sense, we care only about the labels of the events, and their order.
Hence, we should abstract from the exact contents of the carrier.
This is done as follows.

\begin{defi}[Pomsets]
  Let $\lp{u} = \angl{S_\lp{u}, \leq_\lp{u}, \lambda_\lp{u}}$ and $\lp{v} = \angl{S_\lp{v}, \leq_\lp{v}, \lambda_\lp{v}}$ be labelled posets.
  A \emph{labelled poset isomorphism} from $\lp{u}$ to $\lp{v}$ is a bijection $h: S_\lp{u} \to S_\lp{v}$, such that $\lambda_\lp{v} \circ h = \lambda_\lp{u}$, and for $s,s' \in S_\lp{u}$ we have $s \leq_\lp{u} s'$ if and only if $h(s) \leq_\lp{u} h(s')$.
  We write $\lp{u} \cong \lp{v}$ if such an isomorphism exists between $\lp{u}$ and $\lp{v}$, and note that $\cong$ is an equivalence.

  A \emph{partially ordered multiset}, or \emph{pomset} for short, is an equivalence class of labelled posets; we write $[S_\lp{u}, \leq_\lp{u}, \lambda_\lp{u}]$ for the equivalence class of the labelled poset $\angl{S_\lp{u}, \leq_\lp{u}, \lambda_\lp{u}}$.
\end{defi}

Since the collection of labelled posets forms a proper set, $\cong$ is a proper relation.
Consequently, a pomset is also a proper set, as is the set of pomsets.

We write $1$ for the \emph{empty pomset}, i.e., the isomorphism class consisting of the unique empty labelled poset.
When $\ltr{a} \in \Sigma$ we may write $\ltr{a}$ to denote the unique pomset containing one event, which is labelled with $\ltr{a}$; such a pomset is called \emph{primitive}.

\begin{defi}[Pomset composition]%
  \label{definition:pomset-composition}
  Let $U = [S_\lp{u}, \leq_\lp{u}, \lambda_\lp{u}]$ and $V = [S_\lp{v}, \leq_\lp{v}, \lambda_\lp{v}]$ be pomsets.
  Without loss of generality, we may assume that $S_\lp{u}$ and $S_\lp{v}$ are disjoint.

  We define two types of composition.
  The \emph{sequential composition} of $U$ and $V$, denoted $U \cdot V$, is the pomset
  $[S_\lp{u} \cup S_\lp{v}, {\leq_\lp{u}} \cup {\leq_\lp{v}} \cup {(S_\lp{u} \times S_\lp{v})}, \lambda_\lp{u} \cup \lambda_\lp{v}]$.
  The \emph{parallel composition} of $U$ and $V$, denoted $U \parallel V$, is the pomset
  $[S_\lp{u} \cup S_\lp{v}, {\leq_\lp{u}} \cup {\leq_\lp{v}}, \lambda_\lp{u} \cup \lambda_\lp{v}]$.
\end{defi}

In the above, $\lambda_\lp{u} \cup \lambda_\lp{v}: S_\lp{u} \cup S_\lp{v} \to \Sigma$ is the function that agrees with $\lambda_\lp{u}$ on $S_\lp{u}$, and with $\lambda_\lp{v}$ on $S_\lp{v}$ --- note that $\lambda_\lp{u} \cup \lambda_\lp{v}$ is well-defined, given that $S_\lp{u}$ and $S_\lp{v}$ are disjoint.
The order relations of $U \cdot V$ and $U \parallel V$ are partial orders for the same reason.

Of course, one should check that the operators given above are well-defined, i.e., that $U \cdot V$ and $U \parallel V$ are the same (as in, given by isomorphic labelled posets) regardless of the disjoint labelled posets that are chosen as representatives.
This turns out to be the case.

It is easy to see that both operators are associative, and that $\parallel$ is commutative.
Furthermore, $1$ is the unit of both sequential and parallel composition, i.e., $U \cdot 1 = 1 \cdot U = U \parallel 1 = U$ for all pomsets $U$~\cite{gischer-1988}.
For the remainder of this paper, we adopt the convention that $\cdot$ binds more tightly than $\cdot$, i.e., $U \cdot V \parallel W$ should be read as $(U \cdot V) \parallel W$.

\begin{defi}[Pomset types]
  Let $U$ be a pomset.
  We say that $U$ is \emph{sequential} if there exist non-empty pomsets $U_1, U_2$ such that $U = U_1 \cdot U_2$.
  Also, $U$ is a \emph{sequential prime} if it is non-empty, and for all pomsets $V$ and $W$ such that $U = V \cdot W$, we have $V = 1$ or $W = 1$.

  Similarly, we say that $U$ is \emph{parallel} if there exist non-empty pomsets $U_1, U_2$ such that $U = U_1 \parallel U_2$.
  Also, $U$ is a \emph{parallel prime} if it is non-empty, and for all pomsets $V$ and $W$ such that $U = V \parallel W$, it holds that $V = 1$ or $W = 1$.
\end{defi}

For our type of concurrency, we study a specific type of pomset, as follows.
\begin{defi}
  The set of \emph{series-parallel pomsets}, or \emph{sp-pomsets}, denoted $\SP(\Sigma)$, is the smallest set that contains the empty and primitive pomsets, and is closed under sequential and parallel composition.
  In other words, $\SP(\Sigma)$ is the smallest set satisfying the rules
  \begin{mathpar}
    \inferrule{~}{%
      1 \in \SP(\Sigma)
    }
    \and
    \inferrule{
      \ltr{a} \in \Sigma
    }{%
      \ltr{a} \in \SP(\Sigma)
    }
    \and
    \inferrule{%
      U, V \in \SP(\Sigma)
    }{%
      U \cdot V \in \SP(\Sigma)
    }
    \and
    \inferrule{%
      U, V \in \SP(\Sigma)
    }{%
      U \parallel V \in \SP(\Sigma)
    }
  \end{mathpar}
\end{defi}

It should be clear that the pomsets in $\SP(\Sigma)$ built without parallel composition have a total order on their nodes, and hence correspond exactly to words.
We will make this identification throughout this paper, writing $\Sigma^*$ for the set of words over $\Sigma$.

Series-parallel pomsets have the convenient property that they may be partitioned into empty, primitive, sequential and parallel pomsets, in the following way.

\begin{lemC}[{\cite[Theorem 3.1]{gischer-1988}}]%
  \label{lemma:sp-unique-kind}
  Let $U \in \SP(\Sigma)$.
  Exactly one of the following is true:
  \emph{(i)}~$U$ is empty, or
  \emph{(ii)}~$U$ is primitive, or
  \emph{(iii)}~$U$ is sequential, or
  \emph{(iv)}~$U$ is parallel.
\end{lemC}

Another useful tool in dissecting pomsets comes from \emph{factorisation}.
In a sense, factorising a pomset is analogous to writing a word as the sequence of its letters; the difference here is that whereas words can be composed only sequentially, pomsets can also be composed in parallel --- hence, we obtain two types of factorisation.

\begin{defi}[Factorisation]
  Let $U$ be a pomset.
  When $U = U_1 \cdots U_n$ with $U_1, \dots, U_n$ sequential primes, we refer to the sequence $U_1, \dots, U_n$ as a \emph{sequential factorisation} of $U$.

  Similarly, when $U = U_1 \parallel \dots \parallel U_n$ such that $U_1, \dots, U_n$ are parallel primes, we refer to the multiset $\mset{U_1, \dots, U_n}$ as a \emph{parallel factorisation} of $U$.
\end{defi}

Just like a word can be uniquely written as a concatenation of its symbols, so does sequential factorisation give rise to such a unique concatenation.

\begin{lem}[{\cite[Lemma~3.2]{gischer-1988} and~\cite[Proposition~2]{grabowski-1981}; see also~\cite{graham-knuth-motzkin-1972}}]%
  \label{lemma:sequential-factorisation-unique}
  Sequential factorisations exist uniquely for sp-pomsets.
\end{lem}

We conclude with a similar statement about parallel factorisation.

\begin{restatable}{lem}{restateparallelfactorisationunique}%
  \label{lemma:parallel-factorisation-unique}
  Parallel factorisations exist uniquely for sp-pomsets.
\end{restatable}

\subsection{Pomset languages}

All possible executions of a sequential program (described as words) can be collected in a set to form a \emph{language} describing the behaviour of a program.
Analogously, we can collect pomsets in a \emph{pomset language} to describe the behaviour of a concurrent program, as follows.

\begin{defi}
  A \emph{pomset language} is a set of pomsets; a pomset language made up of sp-pomsets is referred to as a \emph{series-parallel language}, or \emph{sp-language} for short.

  The composition operators of pomsets lift in a pointwise manner; concretely
  \begin{mathpar}
    L \cdot L' = \setcompr{U \cdot U'}{U \in L, U' \in L'}
    \and
    L \parallel L' = \setcompr{U \parallel U'}{U \in L, U' \in L'}
  \end{mathpar}
  The Kleene closure also applies to pomset languages, as follows
  \[
    L^* = \bigcup_{n \in \naturals} L^n
    \quad \mbox{in which} \quad
    L^0 = \set{1}
    \quad \mbox{and} \quad
    L^{n+1} = L \cdot L^n
  \]
\end{defi}

To relate pomset languages over different alphabets, the notion of \emph{substitution} is useful; a substitution allows us to translate a pomset language over one alphabet into a pomset language over another alphabet, just by substituting the letters.
\begin{defi}
  Let $\Delta$ be an alphabet.
  A \emph{substitution} is a function $\zeta: \Sigma \to 2^{\SP(\Delta)}$.
  We can lift the domain $\zeta$ to $\SP(\Sigma)$ inductively, as follows:
  \begin{mathpar}
    \zeta(1) = \set{1}
    \and
    \zeta(U \cdot V) = \zeta(U) \cdot \zeta(V)
    \and
    \zeta(U \parallel V) = \zeta(U) \parallel \zeta(V)
  \end{mathpar}
  When $L$ is a pomset language, we write $\zeta(L)$ for the set $\bigcup_{U \in L} \zeta(U)$.

  Finally, we call $\zeta$ \emph{atomic} if all of the following hold:
  \begin{enumerate}[(i)]
  \item
    for $\ltr{a} \in \Sigma$ we have that $\zeta(\ltr{a})$ consists of sequential primes exclusively, and

  \item
    for $\ltr{a}, \ltr{b} \in \Sigma$, we have that $\zeta(\ltr{a}) \cap \zeta(\ltr{b}) \neq \emptyset$ if and only if $\ltr{a} = \ltr{b}$.
  \end{enumerate}
\end{defi}

\noindent
Atomic substitutions have the following useful properties.

\begin{restatable}{lem}{restateatomicproperties}%
\label{lemma:atomic-properties}
Let $L, L' \subseteq \Sigma^*$, and let $\zeta$ be a substitution.
If $\zeta$ is atomic, then
\begin{mathpar}
\zeta(L \cap L') = \zeta(L) \cap \zeta(L')
\and
\zeta(L \setminus L') = \zeta(L) \setminus \zeta(L')
\and
\zeta(L) = \emptyset \iff L = \emptyset
\end{mathpar}
\end{restatable}

\subsection{Series-rational expressions}
Rational expressions can denote languages, which in turn can describe the behaviour of a program.
Continuing our analogy, \emph{series-rational expressions}~\cite{lodaya-weil-2000} can denote the behaviour of a program with fork/join-style concurrency.
Essentially, these are rational expressions extended with parallel composition.

\begin{defi}
The set of \emph{series-rational expressions}, or \emph{sr-expressions} for short, denoted $\terms$, is the smallest set generated by the grammar
\[
    e, f ::= 0 \pipe 1 \pipe \ltr{a} \in \Sigma \pipe e + f \pipe e \cdot f \pipe e \parallel f \pipe e^*
\]
\end{defi}

Like rational expressions, series-rational expressions have a straightforward semantics in terms of sp-languages, where each operator is interpreted as an operator on sp-languages.

\begin{defi}[Semantics]
We define $\sem{-}: \terms \to 2^{\SP(\Sigma)}$ inductively, as follows.
\begin{mathpar}
\sem{0} = \emptyset
\and
\sem{1} = \set{1}
\and
\sem{\ltr{a}} = \set{\ltr{a}}
\and
\sem{e + f} = \sem{e} \cup \sem{f}
\\
\sem{e \cdot f} = \sem{e} \cdot \sem{f}
\and
\sem{e \parallel f} = \sem{e} \parallel \sem{f}
\and
\sem{e^*} = \sem{e}^*
\end{mathpar}
When $L$ is a pomset language and there exists an $e \in \terms$ such that $L = \sem{e}$, we say that $L$ is a \emph{series-rational language}, or \emph{sr-language} for short.
\end{defi}

An important property of sr-expressions is whether their semantics contains the empty pomset.
We syntactically characterise such sr-expressions.
\begin{defi}%
\label{definition:accepting-terms}
  $\sacc$ is the smallest subset of $\terms$ that satisfies the following for all $e, f \in \terms$:
  \begin{mathpar}
    \inferrule{~}{%
      1 \in \sacc
    }
    \and
    \inferrule{%
      e \in \sacc
    }{%
      e + f \in \sacc
    }
    \and
    \inferrule{%
      f \in \sacc
    }{%
      e + f \in \sacc
    }
    \and
    \inferrule{%
      e, f \in \sacc
    }{%
      e \cdot f \in \sacc
    }
    \and
    \inferrule{%
      e, f \in \sacc
    }{%
      e \parallel f \in \sacc
    }
    \and
    \inferrule{~}{%
      e^* \in \sacc
    }
  \end{mathpar}
\end{defi}

To see that $\sacc$ indeed characterizes the empty pomset property, we have

\begin{restatable}{lem}{restatesemvssacc}%
  \label{lemma:sem-vs-sacc}
  Let $e \in \terms$.
  Now $e \in \sacc$ if and only if $1 \in \sem{e}$.
\end{restatable}

\section{Pomset automata}%
\label{section:pomset-automata}

We now turn our attention to pomset automata, which are intended as an operational model for bi-Kleene algebra.
Intuitively, a pomset automaton is a \emph{non-deterministic} automaton enriched with an additional type of transition.
Instances of this type of transition tell us where execution may fork into several ``threads'', as well as where execution resumes when each of these threads has reached an accepting state.
More formally, we have the following.

\begin{defi}
  A \emph{pomset automaton} (PA) is a tuple $\angl{Q, F, \delta, \gamma}$ where
  \begin{itemize}
  \item $Q$ is a set of \emph{states}, with $F \subseteq Q$ the \emph{accepting states}, and
  \item $\delta: Q \times \Sigma \to 2^Q$ is the \emph{sequential transition function}, and
  \item $\gamma: Q \times \M(Q) \to 2^Q$ is the \emph{parallel transition function}.
  \end{itemize}
  Lastly, for all $q \in Q$, there are only finitely many $\phi \in \M(Q)$ such that $\gamma(q, \phi) \neq \emptyset$.
\end{defi}

In the above, $q' \in \gamma(q, \mset{r_1, \dots, r_n})$ should be interpreted to mean that, in state $q$, the automaton may fork into states $r_1, \dots, r_n$, and when all of these have reached an accepting state, may resume computation from $q'$.
The final requirement ensures that only finitely many fork transitions of a state can be meaningful, i.e., lead to an accepting state.

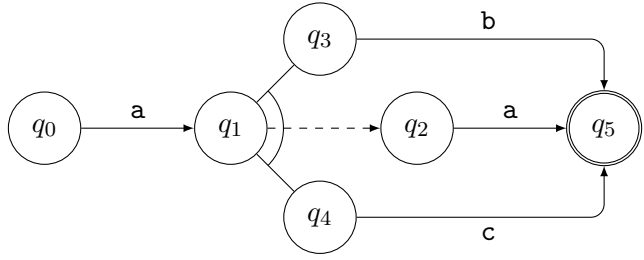
\begin{figure}
  \centering
  \begin{tikzpicture}[every node/.style={transform shape}]
    \node[state] (q0) {$q_0$};
    \node[state,right=15mm of q0] (q1) {$q_1$};
    \node[state,above right=7mm of q1] (q3) {$q_3$};
    \node[state,below right=7mm of q1] (q4) {$q_4$};
    \node[state,right=15mm of q1] (q2) {$q_2$};
    \node[state,accepting,right=15mm of q2] (q5) {$q_5$};

    \draw (q0) edge[-latex] node[above] {$\ltr{a}$} (q1);
    \draw[-] (q1) edge (q3);
    \draw[-] (q1) edge (q4);
    \draw[dashed] (q1) edge[-latex] (q2);
    \draw[rounded corners=5pt,-latex] (q3) -| node[above,xshift=-4em] {$\ltr{b}$} (q5);
    \draw[rounded corners=5pt,-latex] (q4) -| node[below,xshift=-4em] {$\ltr{c}$} (q5);
    \draw (q1) + (-45:7mm) arc (-45:45:7mm);
    \draw (q2) edge[-latex] node[above] {$\ltr{a}$} (q5);
  \end{tikzpicture}
  \caption{PA accepting $\ltr{a} \cdot (\ltr{b} \parallel \ltr{c}) \cdot \ltr{a}$.}\label{figure:example-pa}
\end{figure}

Visually, we can represent pomset automata in a style similar to finite automata, as in \cref{figure:example-pa}.
There, a state is represented by a vertex, which is doubly circled when the state is accepting.
Transitions are represented by edges: the edge between $q_0$ and $q_1$ encodes that $q_1 \in \delta(q_0, \ltr{a})$, and the multi-ended edge that connects $q_1$ to $q_2$ and $q_3$ as well as $q_4$ signifies that $q_2 \in \gamma(q_1, \mset{q_3, q_4})$.
To avoid visual clutter, we draw neither the sequential transitions where $\delta(q, \ltr{a}) = \emptyset$, nor the parallel transitions where $\gamma(q, \mset{r_1, \dots, r_n}) = \emptyset$.
Our convention that for $q \in Q$ there are only finitely many $\phi \in \M(Q)$ with $\gamma(q, \phi) \neq \emptyset$ ensures that a PA with finitely many states also has only finitely many transitions, and can therefore be drawn.

\medskip

For the remainder of this section, we fix a PA $A = \angl{Q, F, \delta, \gamma}$.
We proceed to define how one can ``read'' an sp-pomset $U$ while transitioning from a state $q$ to a state $q'$, in a way that matches the intuition given to the parallel transition function.
This information is encoded in a ternary relation between states, pomsets, and states, called the \emph{run relation}.
If a state $q$ is related to a pomset $U$ and another state $q'$ by this relation, it means that, starting in state $q$, we can read the pomset $U$ to end up in state $q'$.
The pomsets that we can read to reach an accepting state form the language of a state.

\begin{defi}[Runs]%
  \label{definition:runs}
  We define ${\rightarrow_A} \subseteq Q \times \SP(\Sigma) \times Q$ as the smallest relation satisfying
  \begin{mathpar}
    \inferrule{~}{%
      q \arun{1}_A q
    }
    \and
    \inferrule{
      q' \in \delta(q, \ltr{a})
    }{%
      q \arun{\ltr{a}}_A q'
    }
    \and
    \inferrule{%
      q \arun{U}_A q'' \\
      q'' \arun{V}_A q'
    }{%
      q \arun{U \cdot V}_A q'
    }
    \and
    \inferrule{%
      \forall 1 \leq i \leq n.\ q_i \arun{U_i}_A q_i' \in F \\\\
      q' \in \gamma(q, \mset{q_1, \dots, q_n})
    }{%
      q \arun{U_1 \parallel \dots \parallel U_n}_A q'
    }
  \end{mathpar}
  The \emph{language} of $q \in Q$, denoted $L_A(q)$, is the set $\setcompr{U \in \SP(\Sigma)}{q \arun{U}_A q' \in F}$.
\end{defi}

Given an element $q \arun{U}_A q'$ of the run relation, a proof tree (using the inference rules above) witnessing that this triple occurs in $\rightarrow_A$ contains all structural information about how the pomset automaton reads the pomset, i.e., the order in which the states were visited, which fork transitions were used, et cetera.
Such a proof tree can therefore be called a \emph{run} of the pomset automaton; we shall often abuse this nomenclature and refer to individual elements of the run relation as runs, with their underlying proof tree implicitly present.

\begin{exa}%
  \label{example:runs}
  Suppose $A$ is the PA in \cref{figure:example-pa}.
  Then, we have that
  \begin{mathpar}
    q_0 \arun{\ltr{a}}_A q_1
    \and
    q_2 \arun{\ltr{a}}_A q_5
    \and
    q_3 \arun{\ltr{b}}_A q_5
    \and
    q_4 \arun{\ltr{c}}_A q_5
  \end{mathpar}
  From the latter two runs and the fact that $q_2 \in \gamma(q_1, \mset{q_3, q_4})$, it follows that $q_1 \arun{\ltr{b} \parallel \ltr{c}}_A q_2$ by the last rule above.
  By applying the third rule of run composition to this run and the first two runs above, we find that $q_0 \arun{\ltr{a} \cdot (\ltr{b} \parallel \ltr{c}) \cdot \ltr{a}}_A q_5$.
  Since $q_5 \in F$, we have that $\ltr{a} \cdot (\ltr{b} \parallel \ltr{c}) \cdot \ltr{a} \in L_A(q_0)$, that is to say, $q_0$ accepts the pomset $\ltr{a} \cdot (\ltr{b} \parallel \ltr{c}) \cdot \ltr{a}$.
\end{exa}

It is useful to distinguish runs based on the rules that induce them.
To this end, we establish the following terminology for $q, q' \in Q$ and $U \in \SP(\Sigma)$.
If $q \arun{U}_A q'$ follows by an application of the first rule, we speak of a \emph{trivial run}.
Also, if $q \arun{U}_A q'$ has a derivation in which the second rule is applied last, this run is a \emph{sequential unit run}.
Furthermore, if $q \arun{U}_A q'$ is a consequence of the last rule, this run is a \emph{parallel unit run}.
The sequential and parallel unit runs are collectively referred to as \emph{unit runs}.
Lastly, if $q \arun{U}_A q'$ is a result of applying the third rule, i.e., there exist $U_1, U_2 \in \SP(\Sigma)$ and $q'' \in Q$ such that $q \arun{U_1}_A q''$ as well as $q'' \arun{U_2}_A q'$, and neither of these is trivial, then $q \arun{U}_A q'$ is known as a \emph{composite run}.
By definition of $\rightarrow_A$, each run falls into at least one of these categories.

\begin{exa}%
  \label{example:run-sorts}
  Returning to \cref{example:runs} above, we find that $q_0 \arun{\ltr{a}}_A q_1$ is a sequential unit run, and $q_1 \arun{\ltr{b} \parallel \ltr{c}}_A q_2$ is a parallel unit run.
  An example of a trivial run is $q_5 \arun{1}_A q_5$.
  Lastly, the run $q_0 \arun{\ltr{a} \cdot (\ltr{b} \parallel \ltr{c}) \cdot \ltr{a}}_A q_5$ is a composite run.
\end{exa}

\begin{rem}
The type of a run is not uniquely determined by the kind of pomset that labels it.
For example, even though most parallel unit runs are labelled by a parallel pomset, this is not true in general: if $A$ is the PA in \cref{figure:run-confusion}, then we can construct the parallel unit run $q_1 \arun{\ltr{a}}_A q_2$, even though $\ltr{a}$ is not parallel.
Similarly, not every run labelled by the empty pomset is trivial.
For instance, if $q, q' \in Q$ such that $q' \in \gamma(q, \mempty)$, then $q \arun{1}_A q'$.
We deal with this kind of confusion between run types in \cref{section:language-equivalence}.
\end{rem}

We reap the benefits of our new vocabulary with the following useful lemma.

\begin{restatable}{lem}{restaterundeconstruct}%
  \label{lemma:run-deconstruct}
  Let $q \arun{U}_A q'$.
  There exist $q = q_0, \dots, q_\ell = q' \in Q$ and $U_1, \dots, U_\ell \in \SP(\Sigma)$, such that $U = U_1 \cdots U_\ell$, and for all $1 \leq i \leq \ell$ we have that $q_{i-1} \arun{U_i} q_i$ is a unit run.
\end{restatable}

The minimal $\ell$ for a given run as obtained above is known as the \emph{length} of the run.

\begin{figure}
  \centering
  \begin{subfigure}[b]{0.45\textwidth}
    \centering
    \begin{tikzpicture}
      \node[state] (q1) {$q_1$};
      \node[state,above right=7mm of q1] (q3) {$q_3$};
      \node[state,accepting,below right=7mm of q1] (q4) {$q_4$};
      \node[state,accepting,right=15mm of q3] (q5) {$q_5$};
      \node[state,accepting,right=15mm of q1] (q2) {$q_2$};

      \draw[-] (q1) edge (q3);
      \draw[-] (q1) edge (q4);
      \draw[dashed] (q1) edge[-latex] (q2);
      \draw (q3) edge[-latex] node[above] {$\ltr{a}$} (q5);
      \draw (q1) + (-45:7mm) arc (-45:45:7mm);
    \end{tikzpicture}
    \caption{A PA with a parallel unit run, from $q_1$ to $q_2$, which is labelled by a primitive pomset.}%
    \label{figure:run-confusion}
  \end{subfigure}
  \hspace{5mm}
  \begin{subfigure}[b]{0.45\textwidth}
    \centering
    \begin{tikzpicture}
      \node[state,accepting] (q1) {$q_1$};
      \node[state,above right=7mm of q1] (q3) {$q_3$};
      \node[state,accepting,right=15mm of q1] (q2) {$q_2$};
      \node[state,accepting,below right=7mm of q1] (q4) {$q_4$};
      \node[state,accepting,right=15mm of q4] (q5) {$q_5$};

      \draw[-] (q1) edge (q3);
      \draw[-] (q1) edge (q4);
      \draw[dashed] (q1) edge[-latex] (q2);
      \draw (q3) edge[bend right,-latex] node[above left] {$\ltr{a}$} (q1);
      \draw (q1) + (-45:7mm) arc (-45:45:7mm);
      \draw (q4) edge[-latex] node[above] {$\ltr{b}$} (q5);
    \end{tikzpicture}
    \caption{A PA with unbounded parallelism as a result of a fork-cycle.}%
    \label{figure:fork-cycle}
  \end{subfigure}
  \caption{Some problematic pomset automata.}%
  \label{figure:problematic-pas}
\end{figure}
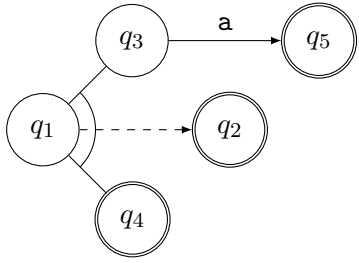
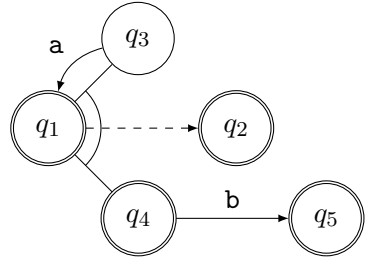

\subsection{Fork-acyclicity}
As it turns out, pomset automata are a rather powerful model of computation.
For instance, they can be used to recognize context-free languages; consequently, language equivalence of states is undecidable in general~\cite{kappe-brunet-luttik-silva-zanasi-2019}.
The heart of the argument from op.\ cit.\ is that we can use the forks of a pomset automaton to simulate something that resembles the call stack of a recursive program.
To prevent this excessive amount of expressive power, we need to put a structural restriction on PAs; specifically, what we will do is ensure that the level of (nested) forks in constructing any run is bounded from above.%
\footnote{%
  Earlier automata models for sr-expressions had to apply similar restrictions~\cite{lodaya-weil-2000,jipsen-moshier-2016}.
}

\begin{exa}
Suppose $A$ is the PA in \cref{figure:fork-cycle}.
We then find that $q_4 \arun{\ltr{b}}_A q_5$ and $q_3 \arun{\ltr{a}}_A q_1$.
Since $q_2 \in \gamma(q_1, \mset{q_3, q_4})$, we know that $q_1 \arun{\ltr{a} \parallel \ltr{b}}_A q_2$.
However, because $q_3 \arun{\ltr{a}}_A q_1$, it follows that $q_3 \arun{\ltr{a} \cdot (\ltr{a} \parallel \ltr{b})}_A q_2$, and hence we also find that $q_1 \arun{\ltr{a} \cdot (\ltr{a} \parallel \ltr{b}) \parallel \ltr{b}}_A q_2$.
This pattern can be repeated indefinitely, leading to an unbounded number of forks in the construction of runs originating from $q_1$.
In addition, we note that $L_A(q_1)$ is not series-rational~\cite{laurence-struth-2014}.
\end{exa}

The problem in the example above is that $q_1$ can fork into a state $q_3$, which can itself reach $q_1$ and fork again, giving rise to an unbounded number of nested forks.
More generally, this behaviour can occur if $q_3$ can construct a run that somehow involves $q_1$, through a series of transitions and forks.
Hence, the first step to inhibit this kind of recursion is to get a handle on the states that can be involved in constructing a run that originates in a given state, by transitioning to that state, by forking into it, or some combination of the two.

\begin{defi}[Support]
  We define $\preceq_A$ as the smallest preorder on $Q$ satisfying for $q \in Q$
  \begin{mathpar}
    \inferrule{%
      \ltr{a} \in \Sigma \\ q'\in \delta(q, \ltr{a})
    }{%
      q' \preceq_A q
    }
    \and
    \inferrule{%
      \phi \in \M(Q) \\
      q' \in \gamma(q, \phi)
    }{%
      q' \preceq_A q
    }
    \and
    \inferrule{%
      r \in \phi \in \M(Q) \\
      \gamma(q, \phi) \neq \emptyset
    }{%
      r \preceq_A q
    }
  \end{mathpar}
  We refer to $\preceq_A$ as the \emph{support relation} of $A$.
  This relation in turn gives rise to the \emph{strict support relation} $\prec_A$, which is the strict order in which $q' \prec_A q$ holds if $q' \preceq_A q$ and $q \not\preceq_A q'$.
\end{defi}

\begin{exa}
  Returning to the PA in \cref{figure:fork-cycle}, we see that $q_5 \preceq_A q_4$, since $q_5 \in \delta(q_4, \ltr{b})$ (by the first rule).
  Because $q_2 \in \gamma(q_1, \mset{q_3, q_4})$, it follows that $q_2 \preceq_A q_1$ (by the second rule) as well as $q_3, q_4 \preceq_A q_1$ (by the third rule).
  By transitivity, it then follows that $q_5 \preceq_A q_1$.
\end{exa}

Intuitively, if $q'$ is necessary to establish some run originating from $q$, then $q' \preceq_A q$; hence, we say that $q'$ \emph{supports} $q$.
In particular, if $r \preceq_A q$ because there exists a $\phi \in \M(Q)$ with $r \in \phi$ and $\gamma(q, \phi) \neq \emptyset$, then $r$ serves as the start of one or more threads that $q$ may fork into, and we say that $r$ is a \emph{fork target} of $q$.
Support can be mutual, for instance when $q' \in \delta(q, \ltr{a})$ and $q \in \delta(q', \ltr{a})$; consequently, $\preceq_A$ need not be antisymmetric.

To break fork cycles, we can define the restriction necessary to avoid infinitely nested forks, by stipulating that a state cannot be supported by any of its fork targets.

\begin{defi}[Fork-acyclicity]
  We say that $A$ is \emph{fork-acyclic} if for $q, r \in Q$ such that there exists a $\phi \in \M(Q)$ with $r \in \phi$ and $\gamma(q, \phi) \neq \emptyset$, we have that $r \prec_A q$, i.e., $q \not\preceq_A r$.
\end{defi}

\begin{exa}
  Returning to the PA in \cref{figure:fork-cycle}, we see that $\gamma(q_1, \mset{q_3, q_4}) \neq \emptyset$, while $q_1 \preceq_A q_3$.
  Hence, this PA is not fork-acyclic.
  On the other hand, the PA in \cref{figure:run-confusion} is fork-acyclic, because neither $q_3$ nor $q_4$ is supported by $q_1$.
\end{exa}

\subsection{Boundedness}
Our primary interest is pomset automata with finitely many states, known as \emph{finite pomset automata}.
It is sometimes convenient to relax this property and speak of PAs with infinitely many states, but where each state relies on only finitely many states to establish its runs.
To formalize this, we introduce the following notions.

\begin{defi}[Support and boundedness]
  We say that $Q' \subseteq Q$ is \emph{support-closed} if for all $q \in Q'$ with $q' \preceq_A q$ we have $q' \in Q'$.
  The \emph{support} of $q \in Q$, denoted $\supp_A(q)$, is the smallest support-closed set containing $q$.
  When $\pi_A(q)$ is finite for all $q \in Q$, we say that $A$ is \emph{bounded}.
\end{defi}

\begin{exa}
  In \cref{figure:fork-cycle}, the set $\set{q_2, q_4, q_5}$ is support-closed, while the set $\set{q_3}$ is not, since $q_1 \preceq_A q_3$.
  The support of $q_1$ is given by the set of all states, while the support of $q_4$ is given by $\set{q_4, q_5}$.
  Like all finite PAs, this PA is bounded.
\end{exa}

When working with bounded and fork-acyclic automata, it is also useful to reason based on how many nested forks can occur in a single run.
To this end, the following is useful.

\begin{defi}[Depth]%
  \label{definition:depth}
  If $q$ is a state in $A$, we write $D_A(q)$ for the \emph{depth of $q$ in $A$}, which is the maximum $n$ such that there exist $q_1, \dots, q_n \in Q$ with $q_1 \prec_A q_1 \prec_A \dots \prec_A q_n = q$, if such an $n$ exists; otherwise, $D_A(q)$ is undefined.
  If $Q$ is finite and $D_A(q)$ is defined for every $q \in Q$, we write $D_A$ for the \emph{depth of $A$}, i.e., the maximum of $D_A(q)$ for all $q \in Q$.
\end{defi}

If $A$ is bounded, then one can use the definition of $\prec_A$ to argue that $D_A(q)$ is defined for every state $q$.
Moreover, if $A$ is finite, then $D_A$ is defined.

\begin{exa}
The PA in \cref{figure:run-confusion} has depth $2$, since $q_3 \prec_A q_1$.
The PA in \cref{figure:fork-cycle} has depth $3$, as $q_5 \prec q_4 \prec q_1$.
\end{exa}

\subsection{Implementation}
We will perform a number of transformations on automata to enforce desirable properties.
To ensure correctness, we require these constructions to transform an automaton $A$ into an automaton $A'$ implementing $A$, in the following sense:
\begin{defi}
  Let $A = \angl{Q, F, \delta, \gamma}$ and $A' = \angl{Q', F', \delta', \gamma'}$ be pomset automata.
  We say that $A'$ \emph{implements} $A$ if the following hold:
  \begin{enumerate}[(i)]
  \item
    $Q \subseteq Q'$, and for every $q \in Q$ it holds that $L_A(q) = L_A(q')$, and

  \item
    if $A$ is fork-acyclic, then so is $A'$.
  \end{enumerate}
\end{defi}

\noindent
Any PA $A$ can be restricted to a support-closed set of states, such that $A$ implements the restricted PA\@.
In particular, if $A$ is bounded, this means that the language of any particular state $q$ can also be described by a finite PA, simply by restricting $A$ to the support of $q$; moreover, if $A$ is bounded, then so is this new PA\@.
Formally, we have the following.

\begin{lem}%
  \label{lemma:restrict-bounded}
  If $Q' \subseteq Q$ is support-closed, then we can construct a PA $A[Q']$ with states $Q'$ such that $A$ implements $A[Q']$.
  Moreover, if $A$ is bounded (resp.\ finite), then so is $A[Q']$.
\end{lem}
\begin{proof}
  Since $Q'$ is support-closed, we can regard $\delta$ as a function of type $Q' \times \Sigma \to 2^{Q'}$, and $\gamma$ as a function of type $Q' \times \M(Q') \to 2^{Q'}$.
  We choose $A[Q'] = \angl{Q', Q' \cap F, \delta, \gamma}$.

  We should show that, for all $q \in Q'$, it holds that $L_{A[Q']}(q) = L_A(q)$.
  The inclusion from left to right should be clear: any (accepting) run in $A[Q']$ can be replayed in $A$ by construction.
  For the other inclusion, we prove more generally that if $q \in Q'$ and $q \arun{U}_A q'$ then $q' \in Q'$ and $q \arun{U}_{A[Q']} q'$, by induction on $\rightarrow_A$.
  In the base, there are two cases.
  \begin{itemize}
  \item
    If $U = 1$ and $q = q'$, then $q' \in Q'$ and $q \arun{U}_{A[Q']} q'$ immediately.

  \item
    If $U = \ltr{a}$ for some $\ltr{a} \in \Sigma$, and $q' \in \delta(q, \ltr{a})$, then $q' \preceq_A q$, and hence $q' \in Q'$.
    Furthermore, we find that $q \arun{U}_{A[Q']} q'$.
  \end{itemize}

  \noindent
  For the inductive step, there are two cases to consider.
  \begin{itemize}
  \item
    Suppose that $q \arun{U}_A q'$ because $U = V \cdot W$, and there exists a $q'' \in Q$ such that $q \arun{V}_A q''$ and $q'' \arun{W}_A q'$.
    It follows that $q'' \in Q'$ and $q \arun{V}_{A[Q']} q''$ by induction.
    Similarly, we find $q' \in Q'$ and $q'' \arun{W}_{A[Q']} q'$, again by induction.
    We then conclude that $q \arun{U}_{A[Q']} q'$.

  \item
    Suppose that $q \arun{U}_A q'$ because $U = U_1 \parallel \dots \parallel U_n$, and there exist $r_1, \dots, r_n \in Q$ and $q_1', \dots, q_n' \in F'$ such that for $1 \leq i \leq n$ we have $q_i \arun{U_i}_A q_i'$, as well as $q' \in \gamma(q, \mset{r_1, \dots, r_n})$.
    We have $r_1, \dots, r_n \preceq_A q$, and thus $r_1, \dots, r_n \in Q'$, as well as $q' \in Q'$.
    By induction, we then find for $1 \leq i \leq n$ that $r_i' \in Q'$ and $r_i \arun{U_i}_{A[Q']} r_i'$.
    We conclude that $q \arun{U}_{A[Q']} q'$.
  \end{itemize}

  \noindent
  A variation of the above argument shows that for $q, q' \in Q$, we have $q \preceq_{A[Q']} q'$ if and only if $q \preceq_A q'$ and $q, q' \in Q'$; hence, if $A$ is fork-acyclic, then so is $A[Q']$.
  This also shows that, for $q \in Q'$, we have that $\supp_{A[Q']}(q) = \supp_A(q)$; hence, if $A$ is bounded, then so is $A[Q']$.
  Finally, if $A$ is finite, then clearly $Q'$ is finite, and hence $A[Q']$ has finitely many states.
\end{proof}

\section{Language equivalence}%
\label{section:language-equivalence}

We consider the following decision problem: given a PA $A = \angl{Q, F, \delta, \gamma}$ and states $q, q' \in Q$, do $q$ and $q'$ accept the same language?
As stated previously, this problem is undecidable for finite PAs that need not be fork-acyclic~\cite{kappe-brunet-luttik-silva-zanasi-2019}.
In this section, we investigate whether there exists a general decision procedure for fork-acyclic finite PAs.

To fully appreciate the complexity of this problem, we start by illustrating the intricacies of PAs through a series of examples.
These show how PAs with very different structures accept the same language.
Any procedure to decide language equivalence for fork-acyclic PAs must take such cases into account.

\begin{exa}[Run confusion]%
\label{example:run-confusion}
In \cref{figure:run-confusion}, we find that both $q_1$ and $q_3$ accept the singleton language $\set{\ltr{a}}$.
However, $q_1$ and $q_3$ have a very different transition structure, since $\delta(q_1, \ltr{a}) = \emptyset$ and $\delta(q_3, \ltr{a}) = \set{q_5}$, with $q_5 \in F$.
Indeed, $q_1$ accepts $\ltr{a}$ by means of a parallel unit run (forking to $q_3$ and $q_4$), while $q_3$ accepts $\ltr{a}$ by means of a sequential unit run.
\end{exa}

More generally, a state could accept a more complicated pomset by means of a composite run (i.e., of length greater than one), while another state accepts the same pomset with a parallel unit run (i.e., of length one) forking into one or more accepting states.
A similar phenomenon occurs when a state forks into a multiset of size one.

\begin{exa}[Empty forks]%
\label{example:empty-forks}
  The definition of pomset automata does not prohibit forks into the empty multiset; in a sense, these are analogous to $\epsilon$-transitions in NFAs, since the parallel composition of zero pomsets is $1$.
  This could allow a non-accepting state to accept the empty pomset, or a state without parallel transitions to accept a parallel pomset.
\end{exa}

Hence, different types of runs may be labelled by the same pomset, especially when fork targets can be accepting, or when unary or nullary forks are allowed.
The next example is about how nested forks may encode the same behaviour differently.

\begin{figure}
  \begin{subfigure}[b]{\textwidth}
    \centering
    \begin{tikzpicture}
      \node[state] (q1) {$q_1$};
      \node[state,above right=7mm of q1] (q3) {$q_3$};
      \node[state,below right=7mm of q1] (q4) {$q_4$};
      \node[state,above=7mm of q3] (q5) {$q_5$};
      \node[state,right=7mm of q3] (q6) {$q_6$};
      \node[state,accepting,below right=7mm of q6] (q2) {$q_2$};

      \draw[-] (q1) edge (q3);
      \draw[-] (q1) edge (q4);
      \draw[-] (q3) edge (q5);
      \draw[-] (q3) edge (q6);

      \draw[dashed] (q1) edge[-latex] (q2);
      \draw[dashed,rounded corners=5pt,-latex] (q3) -- +(1,1) -| (q2);

      \draw (q1) + (-45:7mm) arc (-45:45:7mm);
      \draw (q3) + (0:7mm) arc (0:90:7mm);

      \draw (q6) edge[-latex] node[above right] {$\ltr{b}$} (q2);
      \draw[rounded corners=5pt,-latex] (q4) -| node[above,xshift=-4em] {$\ltr{c}$} (q2);
      \draw[rounded corners=5pt,-latex] (q5) -| +(3.8,-1) node[right,yshift=-1.5em] {$\ltr{a}$} |- (q2.east);

      \node[state,right=6cm of q5] (q4p) {$q_4'$};
      \node[state,below left=7mm of q4p] (q1p) {$q_1'$};
      \node[state,below right=7mm of q1p] (q3p) {$q_3'$};
      \node[state,below=7mm of q3p] (q5p) {$q_5'$};
      \node[state,right=7mm of q3p] (q6p) {$q_6'$};
      \node[state,accepting,above right=7mm of q6p] (q2p) {$q_2'$};

      \draw[-] (q1p) edge (q3p);
      \draw[-] (q1p) edge (q4p);
      \draw[-] (q3p) edge (q5p);
      \draw[-] (q3p) edge (q6p);

      \draw[dashed] (q1p) edge[-latex] (q2p);
      \draw[dashed,rounded corners=5pt,-latex] (q3p) -- +(1,-1) -| (q2p);

      \draw (q1p) + (-45:7mm) arc (-45:45:7mm);
      \draw (q3p) + (0:7mm) arc (0:-90:7mm);

      \draw (q6p) edge[-latex] node[below right] {$\ltr{b}$} (q2p);
      \draw[rounded corners=5pt,-latex] (q4p) -| node[above,xshift=-4em] {$\ltr{a}$} (q2p);
      \draw[rounded corners=5pt,-latex] (q5p) -| +(3.8,1) node[right,yshift=1.5em] {$\ltr{c}$} |- (q2p.east);
    \end{tikzpicture}
    \caption{Associativity of parallelism.}\label{figure:fork-associativity}
  \end{subfigure}\vspace{5mm}

  \begin{subfigure}{\textwidth}
    \centering
    \begin{tikzpicture}
      \node[state] (q1) {$q_1$};

      \node[state,above left=7mm of q1] (q2) {$q_2$};
      \node[state,below left=7mm of q1] (q3) {$q_3$};
      \node[state,accepting,left=15mm of q1] (q6) {$q_6$};

      \node[state,above right=7mm of q1] (q4) {$q_4$};
      \node[state,below right=7mm of q1] (q5) {$q_5$};
      \node[state,accepting,right=15mm of q1] (q7) {$q_7$};

      \draw[-] (q1) edge (q2);
      \draw[-] (q1) edge (q3);
      \draw[dashed] (q1) edge[-latex] (q6);
      \draw (q1) + (135:7mm) arc (135:225:7mm);

      \draw (q2) edge[-latex] node[above left] {$\ltr{a}$} (q6);
      \draw (q3) edge[-latex] node[below left] {$\ltr{b}$} (q6);

      \draw[-] (q1) edge (q4);
      \draw[-] (q1) edge (q5);
      \draw[dashed] (q1) edge[-latex] (q7);
      \draw (q1) + (-45:7mm) arc (-45:45:7mm);

      \draw (q4) edge[-latex] node[above right] {$\ltr{a}$} (q7);
      \draw (q5) edge[-latex] node[below right] {$\ltr{c}$} (q7);

      \node[state,right=10mm of q7] (q1p) {$q_1'$};

      \node[state,above right=7mm of q1p] (q2p) {$q_2'$};
      \node[state,below right=7mm of q1p] (q3p) {$q_3'$};
      \node[state,accepting,right=15mm of q1p] (q4p) {$q_4'$};

      \draw[-] (q1p) edge (q2p);
      \draw[-] (q1p) edge (q3p);
      \draw[dashed] (q1p) edge[-latex] (q4p);
      \draw (q1p) + (-45:7mm) arc (-45:45:7mm);

      \draw (q2p) edge[-latex] node[above right] {$\ltr{a}$} (q4p);
      \draw (q3p) edge[-latex] node[below right] {$\ltr{b}, \ltr{c}$} (q4p);
    \end{tikzpicture}
    \caption{Distributivity of parallelism over union.}\label{figure:fork-distributivity}
  \end{subfigure}\vspace{5mm}

  \begin{subfigure}{\textwidth}
    \centering
    \begin{tikzpicture}
      \node[state] (q1) {$q_1$};

      \node[state,below left=7mm of q1] (q2) {$q_2$};
      \node[state,below right=7mm of q1] (q4) {$q_4$};

      \node[state,above left=7mm of q1] (q3) {$q_3$};
      \node[state,above right=7mm of q1] (q5) {$q_5$};

      \node[state,accepting,right=15mm of q1] (q6) {$q_6$};
      \node[state,accepting,left=15mm of q1] (q7) {$q_7$};

      \draw[-] (q1) edge (q2);
      \draw[-] (q1) edge (q4);
      \draw[dashed,rounded corners=5pt,-latex] (q1) -| (q4);
      \draw (q1) + (135:7mm) arc (135:225:7mm);

      \draw (q2) edge[-latex] node[below left] {$\ltr{b}$} (q7);
      \draw (q4) edge[-latex] node[below right] {$\ltr{b}$} (q6);

      \draw[-] (q1) edge (q3);
      \draw[-] (q1) edge (q5);
      \draw[dashed,rounded corners=5pt,-latex] (q1) -| (q3);
      \draw (q1) + (-45:7mm) arc (-45:45:7mm);

      \draw (q3) edge[-latex] node[above left] {$\ltr{a}$} (q7);
      \draw (q5) edge[-latex] node[above right] {$\ltr{a}$} (q6);

      \node[state,right=10mm of q6] (q1p) {$q_1'$};

      \node[state,above right=7mm of q1p] (q2p) {$q_2'$};
      \node[state,below right=7mm of q1p] (q3p) {$q_3'$};
      \node[state,right=15mm of q1p] (q4p) {$q_4'$};
      \node[state,accepting,right=38mm of q1p] (q5p) {$q_5'$};

      \draw[-] (q1p) edge (q2p);
      \draw[-] (q1p) edge (q3p);
      \draw[dashed] (q1p) edge[-latex] (q4p);
      \draw (q1p) + (-45:7mm) arc (-45:45:7mm);

      \draw[rounded corners=5pt,-latex] (q2p) -| node[above right,xshift=-4em,yshift=1mm] {$\ltr{a}$} (q5p);
      \draw[rounded corners=5pt,-latex] (q3p) -| node[below right,xshift=-4em,yshift=-1mm] {$\ltr{b}$} (q5p);

      \draw (q4p) edge[-latex] node[above] {$\ltr{a}, \ltr{b}$} (q5p);
    \end{tikzpicture}
    \caption{Distributivity of sequential composition over union.}\label{figure:fork-distributivity-bis}
  \end{subfigure}

  \caption{Examples of PAs where $q_1$ and $q_1'$ accept the same language, while having different transition structures.}%
  \label{figure:confounding-pas}
\end{figure}
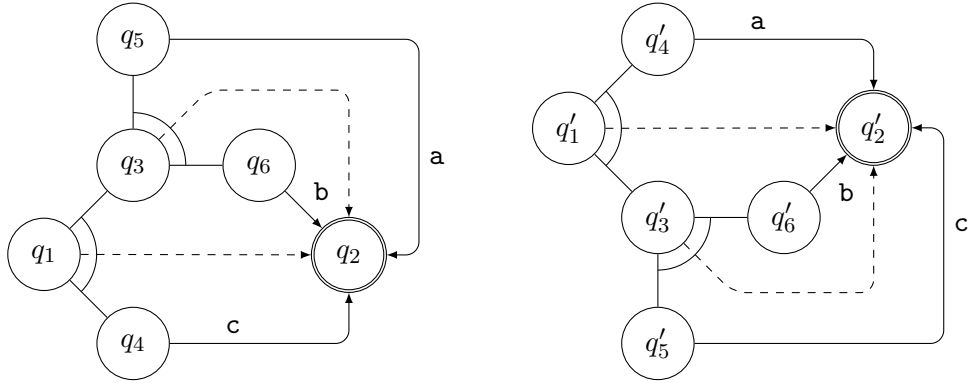
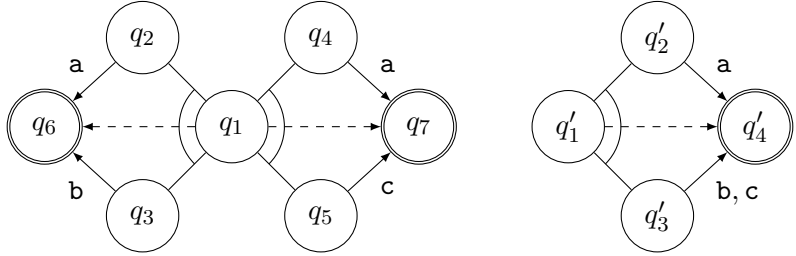
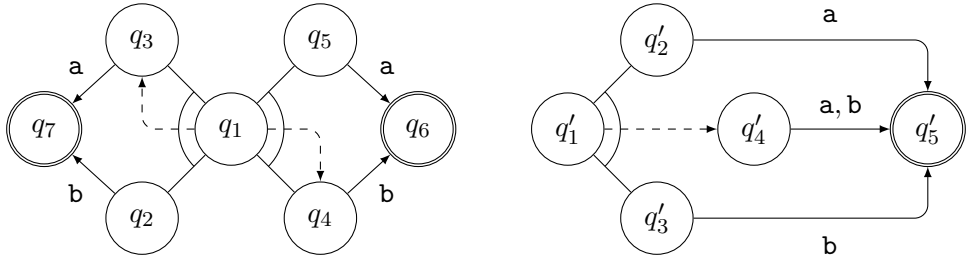

\begin{exa}[Associativity]%
\label{example:associativity}
Consider the PA in \cref{figure:fork-associativity}, where both $q_1$ and $q_1'$ accept the language $\set{\ltr{a} \parallel \ltr{b} \parallel \ltr{c}}$.
In the transition from $q_1$ to $q_2$, the pomset $\ltr{a} \parallel \ltr{b}$ is contributed by $q_3$, and $\ltr{c}$ comes from $q_4$, while in the transition from $q_1'$ to $q_2'$, we obtain $\ltr{a}$ from $q_4'$, and $\ltr{b} \parallel \ltr{c}$ from $q_3'$.
The language of $q_3$ is distinct from, and in fact incomparable with, the languages of $q_3'$ and $q_4'$.
Hence, to compare the language of $q_1$ to that of $q_1'$, we need not only consider the states they may fork into, but also the states where those states may fork into, provided that the second-level forks land in an accepting state.
\end{exa}

We can counteract the phenomena exhibited in \cref{example:run-confusion,example:empty-forks,example:associativity}, by preventing the structures that enable this kind of behaviour.
This leads to the following definition.

\begin{defi}[Well-structured]
A pomset automaton $A = \angl{Q, F, \delta, \gamma}$ is \emph{well-structured} if for $q, q' \in Q$, $\phi \in \M(Q)$ with $q' \in \phi$ and $\gamma(q, \phi) \neq \emptyset$, all of the following hold:
\begin{mathpar}
|\phi| \geq 2\enskip,
\and
q' \not\in F\enskip,\ \text{and}
\and
\forall \psi \in \M(Q).\ \gamma(q', \psi) \cap F = \emptyset\enskip.
\end{mathpar}
\end{defi}
\begin{exa}
The PAs in \cref{figure:example-pa,figure:fork-cycle} are well-structured.
On the other hand, a PA with a state $q$ such that $\gamma(q, \mempty) \neq \emptyset$, or $\gamma(q, \mset{r}) \neq \emptyset$ for some state $r$, is not well-structured (by the first condition).
The PA in \cref{figure:run-confusion} is not well-structured (by the second condition) because $\gamma(q_1, \mset{q_3, q_4}) \neq \emptyset$, while $q_4 \in F$.
Finally, the PA in \cref{figure:fork-associativity} is not well-structured (by the third condition) because $\gamma(q_1, \mset{q_3, q_4}) \neq \emptyset$, while $q_2 \in \gamma(q_3, \mset{q_5, q_6}) \cap F$.
\end{exa}

Well-structuredness does not diminish the expressive power of PAs, as we shall formally prove in \cref{sec:well-structured-automata}.
For instance, the behaviour of the PA discussed in \cref{example:associativity} (see \cref{figure:fork-associativity}) can be expressed using a ternary fork, and the behaviour of the PA discussed in \cref{example:run-confusion} (see \cref{figure:run-confusion}) can be obtained by a sequential transition.
For now, we focus our decision procedure on the fragment of well-structured PAs.

One of the benefits of well-structured PAs is that the type of a run is uniquely determined by the kind of its pomset, as spelled out by the following lemma.

\begin{lem}%
\label{lemma:well-structured-no-confusion}
If $A$ is well-structured and $q \arun{U}_A q'$, then the following hold:
\begin{enumerate}[label={(\roman*)}]
    \item\label{lemma:well-structured-no-confusion:claim:empty}
    $q \arun{U}_A q'$ is trivial if and only if $U$\! is empty.

  \item\label{lemma:well-structured-no-confusion:claim:primitive}
    $q \arun{U}_A q'$ is a sequential unit run if and only if $U$\! is primitive.

  \item\label{lemma:well-structured-no-confusion:claim:sequential}
    $q \arun{U}_A q'$ is a composite run if and only if $U$\! is sequential.

  \item\label{lemma:well-structured-no-confusion:claim:parallel}
    $q \arun{U}_A q'$ is a parallel unit run if and only if $U$\! is parallel.
  \end{enumerate}
\end{lem}
\begin{proof}
  We start by treating~\ref{lemma:well-structured-no-confusion:claim:empty} in detail.
  Here, we note that the implication from left to right holds by definition.
  For the proof from right to left, we proceed by induction on the construction of $q \arun{U}_A q'$.
  In the base, we need only consider the case where $q \arun{U}_A q'$ is already trivial.
  For the inductive step, there are two cases to consider.
  \begin{itemize}
  \item
    Suppose that $U = V \cdot W$ and there exists a $q'' \in Q$ such that $q \arun{V}_A q''$ and $q'' \arun{W}_A q'$, then $V = W = 1$.
    By induction, we then know that $q \arun{V}_A q''$ and $q'' \arun{W}_A q'$ are trivial, and hence $q = q'' = q'$.
    We conclude that $q \arun{U}_A q'$ must also be trivial.

    \item
    Suppose that $U = U_1 \parallel \dots \parallel U_n$ and there exist $r_1, \dots, r_n \in Q$ as well as $r_1', \dots, r_n' \in F$ such that for $1 \leq i \leq n$ we have $r_i \arun{U_i}_A r_i'$, and furthermore $q' \in \gamma(q, \mset{r_1, \dots, r_n})$.
    Then necessarily $U_1, \dots, U_n = 1$.
    Since $A$ is well-structured, we also know that $r_1, \dots, r_n \not\in F$, and furthermore that $n \geq 2$.
    However, by induction, we know that for $1 \leq i \leq n$ it holds that $r_i \arun{U_i}_A r_i'$ is trivial, and hence it would follow that $r_i = r_i'$.
    We have now reached a contradiction, for $r_1 \in F$ while also $r_1 \not\in F$.
    We can therefore disregard this case.
  \end{itemize}

  \noindent
  We treat the implications from left to right for the remaining claims as follows.
  \begin{itemize}
  \item
    For~\ref{lemma:well-structured-no-confusion:claim:primitive}, we find that $U$ is primitive by definition of sequential unit runs.

  \item
    For~\ref{lemma:well-structured-no-confusion:claim:sequential}, suppose that $q \arun{U}_A q'$ is composite.
    We then know that $U = V \cdot W$ and there exists a $q'' \in Q$ such that $q \arun{V}_A q''$ and $q'' \arun{W}_A q'$ are nontrivial.
    By~\ref{lemma:well-structured-no-confusion:claim:empty}, we then know that $V$ and $W$ must be non-empty, and hence $U$ is sequential.

  \item
    For~\ref{lemma:well-structured-no-confusion:claim:parallel}, suppose that $q \arun{U}_A q'$ is a parallel unit run.
    We then know that $U = U_1 \parallel \dots \parallel U_n$ and there exist $r_1, \dots, r_n \in Q$ and $r_1', \dots, r_n' \in F$, such that for $1 \leq i \leq n$ it holds that $r_i \arun{U_i}_A r_i'$, and furthermore $q' \in \gamma(q, \mset{r_1, \dots, r_n})$.
    By the premise that $A$ is well-structured, we know that $r_1, \dots, r_n \not\in F$ and $n \geq 2$.
    It then follows that for $1 \leq i \leq n$, the run $r_i \arun{U_i}_A r_i'$ is non-trivial, and hence $U_i$ is non-empty by the above.
    From this and the fact that $n \geq 2$, we can conclude that $U$ is parallel.
  \end{itemize}
  The implications from right to left for the latter three claims now follow from \cref{lemma:sp-unique-kind}.
  For instance, if $U$ is primitive, then $q \arun{U}_A q'$ must be a sequential unit run, for if it were trivial then $U$ would be empty, if it were composite then $U$ would be sequential, and if it were a parallel unit run then $U$ would be parallel.
\end{proof}

\noindent
Furthermore, because of the restriction on fork targets in a well-structured PA, pomsets in their languages are not parallel:

\begin{restatable}{lem}{restateassociativity}%
\label{lemma:well-structured-associativity}
If $A$ is well-structured and $q$ is a fork target in $A$, then all pomsets in $L_A(q)$ are parallel primes.
\end{restatable}

\medskip
Even for well-structured PAs, however, there are still structural factors confounding language equivalence with which a decision procedure will need to reckon.

\begin{exa}[Distributivity I]%
\label{example:distributivity}
In \cref{figure:fork-distributivity}, we have the state $q_1$ which may fork into $q_2$ and $q_3$, as well as $q_4$ and $q_5$.
Now, $q_1$ accepts the language $\set{\ltr{a} \parallel \ltr{b}, \ltr{a} \parallel \ltr{c}}$, where the former behaviour stems from forking into $q_2$ and $q_3$, but the latter is obtained by forking into $q_4$ and $q_5$.
On the other hand, we have the state $q_1'$, which may fork into $q_2'$ and $q_3'$; here, $q_1'$ accepts the same language as $q_1$, but the two pomsets are due to the \emph{same} fork.
This is a consequence of the distributivity of parallel composition over union of pomset languages:
\[
    L_A(q_1)
    = (\set{\ltr{a}} \parallel \set{\ltr{b}}) \cup (\set{\ltr{a}} \parallel \set{\ltr{c}})
    = \set{\ltr{a}} \parallel (\set{\ltr{b}} \cup \set{\ltr{c}})
    = L_A(q_1')
  \]
\end{exa}
This example illustrates how behaviour of language-equivalent states may be spread out across different parallel transitions, and that this division may differ locally.

The last example stems from an implicit kind of non-determinism that is supported by PAs, as a result of overlap between the languages of fork targets.

\begin{exa}[Distributivity II]%
  \label{example:implicit-non-determinism}
  In \cref{figure:fork-distributivity-bis}, $q_1$ can read $\ltr{a} \parallel \ltr{b}$ to arrive in \emph{either} $q_3$ or $q_4$.
  From that point on, $q_3$ can read $\ltr{a}$ to reach an accepting state, while $q_4$ can read $\ltr{b}$ to do the same.
  In contrast, $q_1'$ can read $\ltr{a} \parallel \ltr{b}$ to arrive in \emph{only one state}, $q_4'$, whence it can read either $\ltr{a}$ or $\ltr{b}$ to arrive in $q_5'$ and accept.
  Nevertheless, $q_1$ and $q_1'$ accept the same language.
  Even though the behaviour implemented by the forks from $q_1$ is the same, the state where they land is not uniquely determined.
  In a sense, this is a consequence of distributivity of sequential composition over union of pomset languages:
  \[
    L_A(q_1)
    = (\set{\ltr{a} \parallel \ltr{b}} \cdot \set{\ltr{a}}) \cup (\set{\ltr{a} \parallel \ltr{b}} \cdot \set{\ltr{b}})
    = \set{\ltr{a} \parallel \ltr{b}} \cdot (\set{\ltr{a}} \cup \set{\ltr{b}})
    = L_A(q_1')
  \]
\end{exa}

\cref{example:distributivity,example:implicit-non-determinism} illustrate that there are many structurally different ways to express parallel behaviour in a pomset automaton.
It is not so clear how to choose for one particular way of representing parallel behaviour.
Instead, we design our algorithm to directly deal with the illustrated equivalences.
To understand how this can be done, it is convenient to first shift perspective from trying to decide language equivalence to trying to find out which states do and do not overlap in terms of their language~\cite{laurence-struth-2014}.

\begin{defi}[Atoms]
  Let $A$ be a PA or an NFA, with states $Q$ and $\alpha \subseteq Q$.
  We write
  \[
    L_A(\alpha) = \Bigl( \bigcap_{q \in \alpha} L_A(q) \Bigr) \setminus \Bigl( \bigcup_{q \not\in \alpha} L_A(q) \Bigr)
  \]
  When $L_A(\alpha) \neq \emptyset$, we say that $\alpha$ is an \emph{atom} of $A$; we write $\At_A$ for the set of atoms of $A$.
\end{defi}
\begin{exa}%
\label{example:atoms}
  In the PA in \cref{figure:fork-distributivity}, $\set{q_3, q_3'}$ is an atom, as is $\set{q_5, q_3'}$; the set $\set{q_2, q_4, q_2'}$ is also an atom.
  In fact, these are all atoms that contain fork targets of \cref{figure:fork-distributivity}.

  The PA in \cref{figure:fork-distributivity-bis} has $\set{q_3, q_5, q_2', q_4'}$ as an atom; similarly, $\set{q_2, q_4, q_3', q_4'}$ is an atom.
  These are again all of the atoms that contain fork targets in \cref{figure:fork-distributivity-bis}.
\end{exa}

The following lemma shows how atoms can be used to decide language equivalence.

\begin{lem}%
\label{lemma:atoms-vs-equality}
Let $A$ be a PA or NFA with states $q_1$ and $q_2$.
Then $L_A(q_1) = L_A(q_2)$ if and only if for all $\alpha \in \At_A$ it holds that $q_1 \in \alpha$ precisely when $q_2 \in \alpha$.
\end{lem}
\begin{proof}
First, suppose that $q_1$ and $q_2$ have the same language, and let $\alpha \in \At$ with $q_1 \in \alpha$.
Then surely $q_2 \in \alpha$, because otherwise we have that $L_A(\alpha) \subseteq L_A(q_1) \setminus L_A(q_2) = \emptyset$, which contradicts that $\alpha$ is an atom.
Similarly, we have that $q_2 \in \alpha$ implies $q_1 \in \alpha$.

For the other implication, let $U \in L_A(q_1)$.
Choose $\alpha = \setcompr{q \in Q}{U \in L_A(q)}$, and note that $\alpha$ is an atom, since $U \in L_A(\alpha)$ by construction.
Because $q_1 \in \alpha$, we know that $q_2 \in \alpha$ by the premise, and thus (by the above), we have that $U \in L_A(q_2)$.
This shows that $L_A(q_1) \subseteq L_A(q_2)$; the other inclusion follows by symmetry.
\end{proof}

Thus, if we can compute the set of atoms of a PA then we can decide language equivalence of states.
As it turns out, this is possible if the PA is finite, fork-acyclic and well-structured.

\begin{lem}\label{lemma:atom-computation}
  If $A$ is finite, fork-acyclic, and well-structured, then $\At_A$ is computable.
\end{lem}
\begin{proof}
  We proceed by induction on $D_A$ (c.f. \cref{definition:depth}).
  In the base, where $D_A = 0$, we have $Q = \emptyset$, since if $q \in Q$, then $D_A(q) \geq 1$; hence, $\emptyset$ is the only atom in this case.\footnote{%
    The empty intersection is assumed to be the set of all sp-pomsets $\SP(\Sigma)$.
  }

  For the inductive step, let $D_A > 0$ and suppose that the claim holds for pomset automata of strictly smaller depth.
  We choose $Q' = \setcompr{q \in Q}{D_A(q) < D_A}$, and note that $Q'$ is support-closed: if $q' \preceq_A q \in Q'$, then $D_A(q') \leq D_A(q)$, and hence $q' \in Q'$.
  By \autoref{lemma:restrict-bounded}, we can then restrict $A$ to obtain $A[Q']$, which, by construction, is of depth strictly less than $D_A$.
  By applying the induction hypothesis, we can compute the atoms of $A[Q']$.

  To compute the atoms of $A$ proper, we shall construct an NFA $A' = \angl{Q, \delta', F}$ whose atoms are precisely those of $A$; the claim then follows because we can compute the atoms of an NFA using a standard algorithm such as the one by Brzozowski and Tamm~\cite{brzozowski-tamm-2014}.
  The idea behind this NFA is that it contains the $\delta$-transitions of $A$, and it encodes the $\gamma$-transitions by transitions labelled with symbols built from the atoms of $A[Q']$.

  The alphabet of our NFA will contain the symbols from $\Sigma$, as well as additional symbols that encode (parts of) the languages of pomsets that can label parallel unit transitions, as divided up by the atoms of $A[Q']$; more precisely, we choose
  \[
    \Delta = \Sigma \cup
        \setcompr{
            \mset{\alpha_1, \dots, \alpha_n} \in \M(\At_{A[Q']})
        }{
            \begin{array}{c}
            \exists q \in Q,\ q_1 \in \alpha_1,\ \dots,\ q_n \in \alpha_n.\\
            \gamma(q, \mset{q_1, \dots, q_n}) \neq \emptyset
            \end{array}
        }
  \]
  Here, we assume without loss of generality that the two sets are disjoint, i.e., that none of the multisets are already symbols in $\Sigma$.
  This alphabet is finite because $\At_{A[Q']}$ is finite, and because by definition of pomset automata, there are finitely many multisets $\phi$ s.t.\ $\gamma(q,\phi)\neq \emptyset$.

  We define the sequential transition function $\delta': Q \times \Delta \to 2^Q$, as follows.
  \begin{mathpar}
    \delta'(q, \ltr{a}) = \delta(q, \ltr{a})
    \and
    \delta'(q, \mset{\alpha_1, \dots, \alpha_n}) = \bigcup \setcompr{\gamma(q, \mset{q_1, \dots, q_n})}{q_1 \in \alpha_1,\ \cdots,\ q_n \in \alpha_n}
  \end{mathpar}

  \begin{figure}
    \begin{subfigure}{\textwidth}
      \centering
      \begin{tikzpicture}
        \node[state] (q1) {$q_1$};

        \node[state,accepting,left=20mm of q1] (q6) {$q_6$};
        \node[state,accepting,right=20mm of q1] (q7) {$q_7$};

        \node[state,above right=7mm of q6] (q2) {$q_2$};
        \node[state,below right=7mm of q6] (q3) {$q_3$};
        \node[state,above left=7mm of q7] (q4) {$q_4$};
        \node[state,below left=7mm of q7] (q5) {$q_5$};

        \draw (q2) edge[-latex] node[above left] {$\ltr{a}$} (q6);
        \draw (q3) edge[-latex] node[below left] {$\ltr{b}$} (q6);

        \draw (q4) edge[-latex] node[above right] {$\ltr{a}$} (q7);
        \draw (q5) edge[-latex] node[below right] {$\ltr{c}$} (q7);

        \draw (q1) edge[-latex] node[above] {\small $\mset{\alpha_3, \alpha_1}$} (q6);
        \draw (q1) edge[-latex] node[above] {\small $\mset{\alpha_3, \alpha_2}$} (q7);

        \node[state,right=10mm of q7] (q1p) {$q_1'$};

        \node[state,accepting,right=20mm of q1p] (q4p) {$q_4'$};
        \node[state,above left=7mm of q4p] (q2p) {$q_2'$};
        \node[state,below left=7mm of q4p] (q3p) {$q_3'$};

        \draw (q2p) edge[-latex] node[above right] {$\ltr{a}$} (q4p);
        \draw (q3p) edge[-latex] node[below right] {$\ltr{b}, \ltr{c}$} (q4p);

        \draw (q1p) edge[-latex] node[above] {\small $\mset{\alpha_3, \alpha_1}$} node[below] {\small $\mset{\alpha_3, \alpha_2}$} (q4p);
      \end{tikzpicture}
      \caption{NFA created in inductive step of atom computation for \cref{figure:fork-distributivity}.}%
      \label{figure:fork-distributivity-atoms}
    \end{subfigure}\vspace{5mm}

    \begin{subfigure}{\textwidth}
      \centering
      \begin{tikzpicture}
        \node[state] (q1) {$q_1$};

        \node[state,left=2cm of q1] (q3) {$q_3$};
        \node[state,right=2cm of q1] (q4) {$q_4$};

        \node[state,accepting,below=of q3] (q7) {$q_7$};
        \node[state,accepting,below=of q4] (q6) {$q_6$};

        \node[state,right=of q7] (q2) {$q_2$};
        \node[state,left=of q6] (q5) {$q_5$};

        \draw (q2) edge[-latex] node[below] {$\ltr{b}$} (q7);
        \draw (q4) edge[-latex] node[right] {$\ltr{b}$} (q6);

        \draw (q3) edge[-latex] node[left] {$\ltr{a}$} (q7);
        \draw (q5) edge[-latex] node[below] {$\ltr{a}$} (q6);

        \draw (q1) edge[-latex] node[above] {\small $\mset{\alpha_1, \alpha_2}$} (q3);
        \draw (q1) edge[-latex] node[above] {\small $\mset{\alpha_1, \alpha_2}$} (q4);

        \node[state,above right=5mm and 13mm of q6] (q1p) {$q_1'$};

        \node[state,above right=7mm of q1p] (q2p) {$q_2'$};
        \node[state,below right=7mm of q1p] (q3p) {$q_3'$};
        \node[state,right=15mm of q1p] (q4p) {$q_4'$};
        \node[state,accepting,right=38mm of q1p] (q5p) {$q_5'$};

        \draw[rounded corners=5pt,-latex] (q2p) -| node[above right,xshift=-4em,yshift=1mm] {$\ltr{a}$} (q5p);
        \draw[rounded corners=5pt,-latex] (q3p) -| node[below right,xshift=-4em,yshift=-1mm] {$\ltr{b}$} (q5p);

        \draw (q4p) edge[-latex] node[above] {$\ltr{a}, \ltr{b}$} (q5p);

        \draw (q1p) edge[-latex] node[above] {\small $\mset{\alpha_1, \alpha_2}$} (q4p);
      \end{tikzpicture}
      \caption{NFA created in inductive step of atom computation for \cref{figure:fork-distributivity-bis}.}%
      \label{figure:fork-distributivity-bis-atoms}
    \end{subfigure}

    \caption{Examples of PAs obtained in the inductive step of atom computation.}%
  \end{figure}
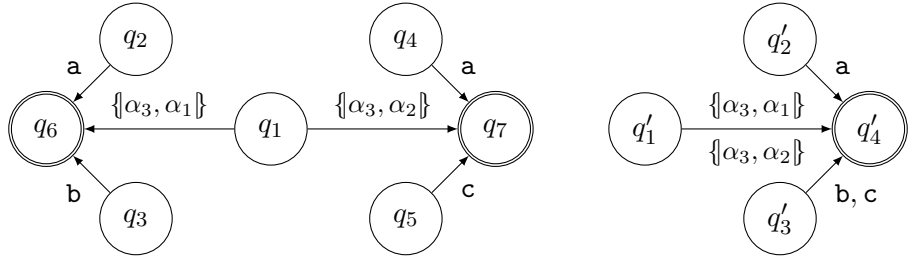
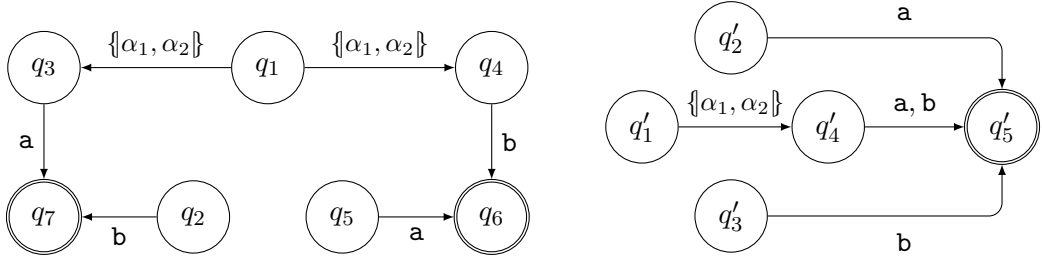

  \begin{exa}
  Let $A$ be the PA in \cref{figure:fork-distributivity}.
  In \cref{example:atoms}, we saw that $\alpha_1 = \set{q_3, q_3'}$, $\alpha_2 = \set{q_5, q_3'}$ and $\alpha_3 = \set{q_2, q_4, q_2'}$ are the atoms that contain fork targets.
  The resulting NFA is drawn in \cref{figure:fork-distributivity-atoms}.
  There, we see that $q_6 \in \delta'(q_2, \ltr{a})$ because $q_6 \in \delta(q_2, \ltr{a})$.
  Furthermore, $q_6 \in \delta'(q_1, \mset{\alpha_3, \alpha_1})$ because $q_6 \in \gamma(q_1, \mset{q_2, q_3})$ and $q_2 \in \alpha_3$ while $q_3 \in \alpha_1$.
  \end{exa}

  \begin{exa}
  Let $A$ be the PA in \cref{figure:fork-distributivity-bis}.
  In \cref{example:atoms}, we found atoms $\alpha_1 = \set{q_3, q_5, q_2', q_4'}$ and $\alpha_2 = \set{q_2, q_4, q_3', q_4'}$.
  The resulting NFA is drawn in \cref{figure:fork-distributivity-bis-atoms}.
  There, we see that $q_7 \in \delta'(q_1, \mset{\alpha_1, \alpha_2})$ because $q_7 \in \gamma(q_1, \mset{q_3, q_2})$ with $q_3 \in \alpha_1$ and $q_2 \in \alpha_2$.
  \end{exa}

  Having defined our target NFA, it remains to show that the atoms of $A$ are the same as those of $A'$.
  To this end, we need to relate the languages of $A'$ to the languages of $A$; we do this by means of the substitution $\zeta: \Delta \to 2^{\SP(\Sigma)}$, given by:
  \begin{mathpar}
    \zeta(\ltr{a}) = \set{\ltr{a}}
    \and
    \zeta(\mset{\alpha_1, \dots, \alpha_n}) = L_A(\alpha_1) \parallel \cdots \parallel L_A(\alpha_n)
  \end{mathpar}

  \noindent
  We need the following two technical properties of $\zeta$.
  The first of these relates the languages of the states of the NFA $A'$ to the languages of the states of the PA $A$ by means of $\zeta$.

\begin{restatable}{fact}{restatesubstitutionvsruns}%
\label{proposition:substitution-vs-runs}
For all $q \in Q$, it holds that $L_A(q) = \zeta(L_{A'}(q))$.
\end{restatable}

  The second property that we need says that $\zeta$ is atomic; this is a consequence of the fact that $A$ is well-structured, and hence if $\mset{\alpha_1, \dots, \alpha_n} \in \Delta$, then $n \geq 2$ and $1 \not\in L_A(\alpha_i)$ for all $1 \leq i \leq n$, meaning that $\zeta(\mset{\alpha_1, \dots, \alpha_n})$ must consist of sequential primes.

\begin{restatable}{fact}{restatesubstitutionatomic}%
\label{proposition:substitution-atomic}
The substitution $\zeta$ is atomic.
\end{restatable}

  \noindent
  Let $\alpha \subseteq Q$; we can then use the above observations to calculate that
  \begin{align*}
    \zeta(L_{A'}(\alpha))
    &= \zeta \Bigl( \Bigl( \bigcap_{q \in \alpha} L_{A'}(q) \Bigr) \setminus \Bigl( \bigcup_{q \not\in \alpha} L_{A'}(q) \Bigr) \Bigr)
      \tag{def. $L_{A'}$ on sets of states} \\
    &= \Bigl( \bigcap_{q \in \alpha} \zeta(L_{A'}(q)) \Bigr) \setminus \Bigl( \bigcup_{q \not\in \alpha} \zeta(L_{A'}(q)) \Bigr)
      \tag{\cref{proposition:substitution-atomic,lemma:atomic-properties}} \\
    &= \Bigl( \bigcap_{q \in \alpha} L_{A}(q) \Bigr) \setminus \Bigl( \bigcup_{q \not\in \alpha} L_{A}(q) \Bigr)
      \tag{\cref{proposition:substitution-vs-runs}} \\
    &= L_{A}(\alpha)
      \tag{def. $L_{A}$ on sets of states}
  \end{align*}
  To wrap up, suppose $\alpha$ is an atom of $A$; then $\zeta(L_{A'}(\alpha)) = L_A(\alpha) \neq \emptyset$, and hence $L_{A'}(\alpha) \neq \emptyset$ by \cref{proposition:substitution-atomic,lemma:atomic-properties}, making $\alpha$ an atom of $A'$.
  Conversely, if $\alpha$ is an atom of $A'$ then $L_A(\alpha) = \zeta(L_{A'}(\alpha)) \neq \emptyset$ by \cref{proposition:substitution-atomic,lemma:atomic-properties}, meaning $\alpha$ is an atom of $A$.
\end{proof}

\cref{lemma:atom-computation,lemma:atoms-vs-equality} now give us the decidability result for well-structured PAs.

\begin{thm}%
  \label{theorem:decision-procedure-well-structured}
  Let $A$ be a finite, fork-acyclic and well-structured PA\@.
  Given states $q_1$ and $q_2$, it is decidable whether $L_A(q_1) = L_A(q_2)$.
\end{thm}

The complexity of the above procedure is dominated by the complexity of calculating the atoms of the NFA constructed in each inductive step.
Brzozowski and Tamm's algorithm is doubly exponential in the number of states~\cite{brzozowski-tamm-2014}, which implies a rough upper bound of $O(D_A \cdot 2^{2^{|Q|}})$ on our algorithm.
Despite this, the algorithm in~\cite{brzozowski-tamm-2014} is known to have good performance in practice~\cite{almeida-moreira-reis-2007}.

\section{Well-structured automata}%
\label{sec:well-structured-automata}

In this section we show that for every finite and fork-acyclic automaton, we can construct a finite, fork-acyclic, and \emph{well-structured} automaton that implements it.
More precisely, we show that this transformation can take a \emph{bounded} and fork-acyclic automaton, and produce a bounded, fork-acyclic and well-structured automaton; the former claim then follows, since every finite automaton is bounded, and every bounded automaton can be made finite while preserving fork-acyclicity, by \cref{lemma:restrict-bounded}.\footnote{It is clear that the construction from this lemma also preserves well-structuredness.}
We first decompose the definition of well-structured automata into three simpler properties.

\begin{defi}[$n$-forking]
Let $n \in \naturals$.
A pomset automaton $A$ is \emph{$n$-forking} if for every state $q\in Q$ and every multiset of states $\phi\in\M(Q)$ such that $\gamma(q,\phi)\neq\emptyset$ we have $|\phi| \geq n$.
\end{defi}

\begin{defi}[Parsimony]
A pomset automaton $A$ is said to be \emph{parsimonious} if, whenever $p \in Q$ and $q \in \phi \in \M(Q)$ such that $\gamma(p, \phi) \neq \emptyset$, we have that $1 \not\in L_A(q)$.
\end{defi}

\begin{defi}[Flat-branching]
A pomset automaton $A$ is \emph{flat-branching} if for all states $q, q'\in Q$ and every multiset $\phi \in \M(Q)$, if $\gamma(q,\phi) \neq \emptyset$ and $q' \in \phi$ then
\begin{mathpar}
\forall \psi \in \M(Q).\; \gamma(q', \psi) \cap F =\emptyset.
\end{mathpar}
\end{defi}

\begin{exa}
All automata displayed in \cref{figure:example-pa,figure:problematic-pas,figure:confounding-pas} are $2$-forking.
The automata in \cref{figure:confounding-pas} are parsimonious, but the ones in \cref{figure:problematic-pas} are not; for instance, in \cref{figure:run-confusion}, we have that $\gamma(q_1, \mset{q_3, q_4}) \neq \emptyset$, while $1 \in L_A(q_4)$.
The automata in \cref{figure:problematic-pas,figure:fork-distributivity,figure:fork-distributivity-bis} are flat-branching, but the one in \cref{figure:fork-associativity} is not; in particular, in that automaton we have that $\gamma(q_1, \mset{q_3, q_4}) \neq \emptyset$, while $q_2 \in \gamma(q_3, \mset{q_5, q_6}) \cap F$.
\end{exa}
One can prove that the above properties are indeed equivalent to well-structuredness.
\vspace{-0.3em}
\begin{restatable}{lem}{restatewellstructuredifftriple}%
\label{lemma:well-structured-iff-triple}
A pomset automaton $A$ is well-structured if and only if it is $2$-forking, parsimonious, and flat-branching.
\end{restatable}
Our task can therefore be reduced to converting a bounded and fork-acyclic PA into an equivalent bounded and fork-acyclic PA that is also $2$-forking, parsimonious and flat-branching.
Note that these properties can be at cross purposes: for instance, ensuring parsimony may introduce forks of smaller sizes, which could make it so that the automaton is no longer $2$-forking.
Similarly, eliminating forks into the empty multiset may introduce new accepting states, which can invalidate flat-branching.

Thus, establishing all properties simultaneously requires some care.
The remainder of this section describes a series of transformations that establish one property while maintaining the ones already established.
More precisely, our construction to convert a bounded PA $A$ into a bounded, $2$-forking, parsimonious and flat-branching PA goes as follows:
\begin{enumerate}
    \item
    we first show how to implement $A$ with a parsimonious automaton $A_0$;

    \item
    then we discuss how to implement $A$ with a $1$-forking automaton $A_1$;

    \item
    we proceed to show how to implement $A$ with a $2$-forking automaton $A_2$;

    \item
    finally, we show that $A$ can be implemented by a flat-branching PA $A_3$.
\end{enumerate}

\noindent
Since each of these transformations preserves the established properties (e.g., $A_2$ is still bounded, $1$-forking and parsimonious), we end up with a bounded PA that implements $A$ and is $2$-forking, parsimonious, and flat-branching, and hence well-structured by \cref{lemma:well-structured-iff-triple}.

\medskip
Before we get into the weeds, we discuss some technical properties that will help simplify the constructions.
First, note that we shall need to rule out a form of $\epsilon$-transitions.
Indeed, although our model does not include such transitions explicitly, they may be achieved by branching to states accepting the empty pomset (if the automaton is not parsimonious), or by using a $\gamma(p, \mempty)$-transition (if the automaton is not $1$-forking).
The result is a run $p\arun 1_A q$ that is non-trivial.
To reason about such transitions, we observe the following:

\begin{restatable}{lem}{restateepsilondecidable}%
\label{lemma:epsilon-decidable}
For any bounded PA $A$, the predicate $p\arun 1_{A}q$ is decidable.
\end{restatable}

Another useful notion for this section will be that of \emph{weak implementation}.
Essentially, a weak implementation of a PA is another PA where the behaviour of each state of the first PA is spread out across a set of states, rather than just one (as is the case for implementation).

\begin{defi}
  A PA $A'$ \emph{weakly implements} a PA $A$ if the following hold:
  \begin{enumerate}[(i)]
  \item for every state $q$ in $A$ there is a finite set of states $Q_q$ in $A'$ s.t.\ $L_A(q)=\bigcup_{x\in Q_q}L_{A'}(x)$.
  \item if $A$ is fork-acyclic, then so is $A'$.
  \end{enumerate}
\end{defi}

\noindent
To prove that there exists an automaton implementing $A$ that satisfies some of the properties above, it suffices to find one that weakly implements $A$.

\begin{restatable}{lem}{restateweakimplementationiffimplementation}%
\label{lem:weak-implem-iff-implem}
If a PA $A'$ weakly implements $A$, then we can construct a PA $A''$ implementing $A$.
If $A'$ is bounded (resp.\ $n$-forking, flat-branching, parsimonious), then so is $A''$.
\end{restatable}

\subsection{Parsimony}%
\label{sec:parsimony}

Let $A = \angl{Q, F, \delta, \gamma}$ be a bounded and fork-acyclic PA;\@ we want to implement $A$ with a fork-acyclic and parsimonious PA $A_0$.
There are two key ideas to this translation:
\begin{itemize}
    \item
    We introduce a new state $\top$, which will be the only accepting state of the automaton; in fact, it will be the only state accepting the empty pomset.
    This state will not have any outgoing transitions, so its language is exactly $\set{1}$.
    We will modify the transition functions, such that if a transition in $A$ can lead to a state that accepts $1$, it can also non-deterministically lead to $\top$ in $A_0$.

  \item
    To ensure our construction preserves successful runs, we need to add $\gamma$ transitions to mitigate the previous modification.
    More precisely, if $p \in Q$ can fork into $\phi \sqcup \psi \in \M(Q)$ to reach $q \in Q$, i.e., $q \in \gamma(p, \phi \sqcup \psi)$, and all states in $\psi$ accept the empty pomset in $A$, then we make sure that, in $A_0$, $q$ can also fork into $\phi$ to reach $p$, simulating the acceptance of the empty pomset from states in $\psi$.
\end{itemize}
Doing so, we obtain an automaton weakly implementing $A$ --- indeed, if $1 \not\in L_A(q)$, then $L_A(q) = L_{A_0}(q)$, and otherwise $L_A(q) = L_{A_0}(q) \cup L_{A_0}(\top)$.
Since $\top$ cannot be a fork target, and any other state cannot accept the empty pomset, this new automaton is parsimonious.
Finiteness and fork-acyclicity are also maintained by this construction.

\begin{defi}[$A_0$]%
\label{definition:parsimonification}
The PA $A_0$ is given by the tuple $\angl{Q_0, F_0, \delta_0, \gamma_0}$ where $Q_0 = Q  \cup \set{\top}$ (with $\top \not\in Q$), and $F_0 = \set{\top}$.
Furthermore, $\delta_0$ is generated by the rules
\begin{mathpar}
\inferrule{%
    q' \in \delta(q, \ltr{a})
}{%
    q' \in \delta_0(q, \ltr{a})
}
\and
\inferrule{%
    q' \in \delta(q, \ltr{a}) \\
    1 \in L_A(q')
}{%
    \top \in \delta_0(q, \ltr{a})
}
\end{mathpar}
Also, $\gamma_0$ is generated by the following rules for all $q \in Q$ and $\phi \in \M(Q)$:
\begin{mathpar}
\inferrule{%
    q' \in \gamma(q, \phi \sqcup \psi) \\\\
    \forall r \in \psi.\ 1 \in L_A(r)
}{%
    q' \in \gamma_0(q, \phi)
}
\and
\inferrule{%
    q' \in \gamma(q, \phi \sqcup \psi) \\
    1 \in L_A(q') \\\\
    \phi \neq \mempty \\
    \forall r \in \psi.\ 1 \in L_A(r)
}{%
    \top \in \gamma_0(q, \phi)
}
\end{mathpar}
Lastly, $\delta_0(\top, \ltr{a}) = \emptyset$ for all $\ltr{a} \in \Sigma$ and $\gamma_0(\top, \phi) = \emptyset$ for all $\phi \in \M(Q_0)$.
\end{defi}

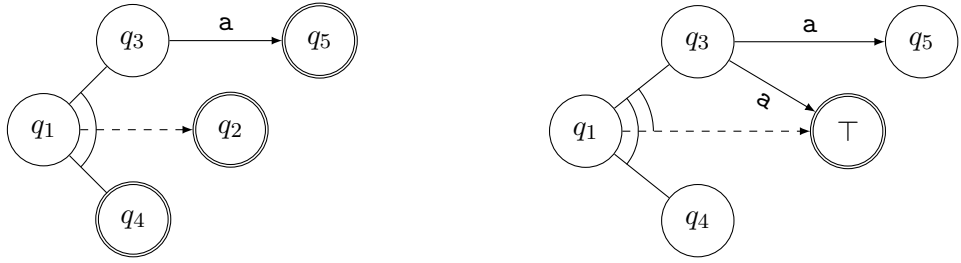
\begin{figure}
    \begin{subfigure}[b]{0.49\textwidth}
        \centering
        \begin{tikzpicture}
            \node[state] (q1) {$q_1$};
            \node[state,above right=7mm of q1] (q3) {$q_3$};
            \node[state,accepting,below right=7mm of q1] (q4) {$q_4$};
            \node[state,accepting,right=15mm of q3] (q5) {$q_5$};
            \node[state,accepting,right=15mm of q1] (q2) {$q_2$};

            \draw[-] (q1) edge (q3);
            \draw[-] (q1) edge (q4);
            \draw[dashed] (q1) edge[-latex] (q2);
            \draw (q3) edge[-latex] node[above] {$\ltr{a}$} (q5);
            \draw (q1) + (-45:7mm) arc (-45:45:7mm);
        \end{tikzpicture}
        \caption{A PA $A$ that is not parsimonious.}%
        \label{figure:parsimonification-before}
    \end{subfigure}
    \begin{subfigure}[b]{0.49\textwidth}
        \centering
        \begin{tikzpicture}
          \node[state] (q1) {$q_1$};

          \node[state,above right=5mm and 8mm of q1] (q3) {$q_3$};
          \node[state,below right=5mm and 8mm of q1] (q4) {$q_4$};
          \node[state,accepting,right=25mm of q1] (top) {$\top$};
          \node[state,right=2cm of q3] (q5) {$q_5$};

          \draw[-] (q1) edge (q3);
          \draw[-] (q1) edge (q4);
          \draw[dashed] (q1) edge[-latex] (top);
          \draw (q3) edge[-latex] node[below,sloped] {$\ltr{a}$} (top);
          \draw (q3) edge[-latex] node[above,sloped] {$\ltr{a}$} (q5);
          \draw (q1) + (-39:7mm) arc (-39:39:7mm);
          \draw (q1) + (0:9mm) arc (0:39:9mm);
        \end{tikzpicture}
        \caption{Part of the PA $A_0$ obtained from $A$.}%
        \label{figure:parsimonification-after}
    \end{subfigure}
    \caption{Example of construction to ensure parsimony.}%
    \label{figure:parsimonification}
\end{figure}

\begin{exa}
Recall the automaton $A$ from \cref{figure:run-confusion}, depicted in \cref{figure:parsimonification-before}.
Part of the resulting PA $A_0$ is drawn in \cref{figure:parsimonification-after}.
For instance, since $q_5 \in \delta(q_3, \ltr{a})$ in $A$, we have $q_5 \in \delta_0(q_3, \ltr{a})$ and $\top \in \delta_0(q_3, \ltr{a})$ in $A_0$ by the first and second rules generating $\delta_0$ respectively.

Furthermore, since $1\in L_A(q_4)$ and $q_2 \in \gamma(q_1, \mset{q_3, q_4})$ while $1 \in L_A(q_2)$, we have that $\top \in \gamma_0(q_1, \mset{q_3, q_4})$ as well as $\top \in \gamma_0(q_1, \mset{q_3})$ by the second rule generating $\gamma_0$.
Not drawn are the transitions witnessed by $q_2 \in \gamma_0(q_1, \mset{q_3, q_4})$ and $q_2 \in \gamma_0(q_1, \mset{q_3})$, both of which are consequences of the first rule generating $\gamma_0$, but do not contribute to the language of the automaton.
Indeed in the new automaton, no accepting state is present in the support of $q_2$, so visiting this state never leads to an accepting run.

Note that even though $A$ was $2$-forking, this is no longer the case in $A_0$, as a result of the fact that $\top \in \gamma_0(q_1, \mset{q_3})$.
We will remedy the appearance of unary forks later on.
\end{exa}

Our construction is correct, in the following sense.

\begin{restatable}{lem}{restateparsimonificationcorrectness}%
\label{lemma:parsimonification-correctness}
$A_0$ is bounded, parsimonious, and weakly implements $A$.
\end{restatable}

\subsection{Removing nullary forks}%
\label{sec:remove-epsilon}

Let $A = \angl{Q, F, \delta, \gamma}$ be a bounded, parsimonious, and fork-acyclic automaton.
We want to implement $A$ with a $1$-forking PA $A_1$ while maintaining boundedness, parsimony and fork-acyclicity.
As mentioned, the \emph{nullary} forks that we want to eliminate --- i.e., those where $q \in \gamma(p, \mempty)$ --- essentially furnish silent transitions $p \epstrans q$, analogous to classic NFAs.
We proceed eliminate these by means of the classic method, i.e., by saturating the transition relations.

\begin{defi}[$A_1$]
The PA $A_1$ is defined to be $\angl{Q,F,\delta_1,\gamma_1}$, where $\delta_1$ and $\gamma_1$ are generated by the following inference rules for all $\ltr{a} \in \Sigma$ and $\phi \in \M(Q)$ with $\phi \neq \mempty$.
\begin{mathpar}
\inferrule{%
    p \epstrans q \\
    r \in \delta(q, \ltr{a}) \\
    r \epstrans s
}{%
    s \in \delta_1(p, \ltr{a})
}
\and
\inferrule{%
    p \epstrans q \\
    r \in \gamma(q, \phi) \\
    r \epstrans s
}{%
    s \in \gamma_1(p, \phi)
}
\end{mathpar}
\end{defi}

\begin{figure}
    \begin{subfigure}{0.49\textwidth}
        \centering
        \begin{tikzpicture}
          \node[state] (q1) {$q_1$};
          \node[state,above right=7mm of q1] (q3) {$q_3$};
          \node[state,below right=7mm of q1] (q4) {$q_4$};
          \node[state,right=35mm of q3] (q5) {$q_5$};
          \node[state,right=15mm of q1] (q2) {$q_2$};
          \node[state,accepting,right=10mm of q2] (q6) {$q_6$};

          \draw[-] (q1) edge (q3);
          \draw[-] (q1) edge (q4);
          \draw[dashed] (q1) edge[-latex] (q2);
          \draw (q1) + (-45:7mm) arc (-45:45:7mm);

          \draw[dashed] (q3) edge[-latex] (q5);
          \draw[dashed] (q2) edge[-latex] (q6);
          \draw[-] (q5) edge[-latex] node[above,sloped] {$\ltr{a}$} (q6);
          \draw[-] (q4) edge[-latex] node[below,sloped] {$\ltr{b}$} (q2);
        \end{tikzpicture}
        \caption{A PA $A$ with nullary forks.}%
        \label{figure:nullary-fork-removal-before}
    \end{subfigure}
    \begin{subfigure}{0.49\textwidth}
        \centering
        \begin{tikzpicture}
          \node[state] (q1) {$q_1$};
          \node[state,above right=7mm of q1] (q3) {$q_3$};
          \node[state,below right=7mm of q1] (q4) {$q_4$};
          \node[state,right=35mm of q3] (q5) {$q_5$};
          \node[state,right=15mm of q1] (q2) {$q_2$};
          \node[state,accepting,right=10mm of q2] (q6) {$q_6$};

          \draw[-] (q1) edge (q3);
          \draw[-] (q1) edge (q4);
          \draw[dashed] (q1) edge[-latex] (q2);
          \draw[dashed] (q1) edge[-latex, bend left=25] (q6);
          \draw (q1) + (-45:7mm) arc (-45:45:7mm);

          \draw[-] (q5) edge[-latex] node[above,sloped] {$\ltr{a}$} (q6);
          \draw[-] (q4) edge[-latex] node[below,sloped] {$\ltr{b}$} (q2);
          \draw[-] (q3) edge[-latex, bend left] node[above,sloped] {$\ltr{a}$} (q6);
          \draw[-] (q4) edge[-latex, bend right] node[below,sloped] {$\ltr{b}$} (q6);
        \end{tikzpicture}
        \caption{The PA $A_1$ obtained from $A$.}%
        \label{figure:nullary-fork-removal-after}
    \end{subfigure}
    \caption{Example of nullary fork removal.}%
    \label{figure:nullary-fork-removal}
\end{figure}
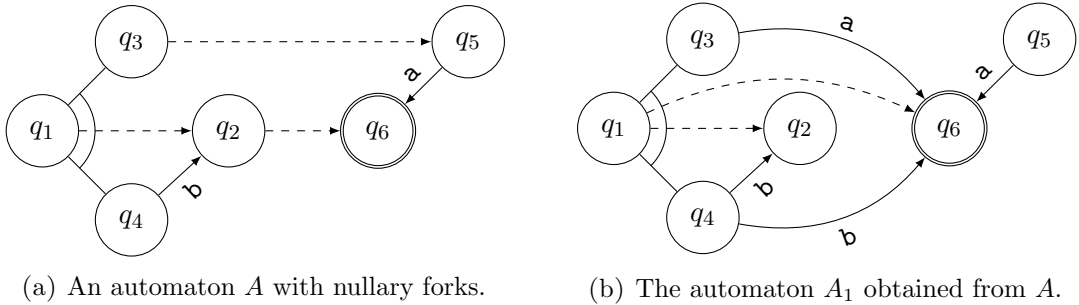

\begin{exa}
Suppose $A$ is the automaton in \cref{figure:nullary-fork-removal-before}.
This automaton has two nullary forks: $q_5 \in \gamma(q_3, \mempty)$ and $q_6 \in \gamma(q_2, \mempty)$.
We have drawn the automaton $A_1$ obtained from $A$ in \cref{figure:nullary-fork-removal-after}.
Here, $q_6 \in \delta_1(q_3, \ltr{a})$, because $q_3 \epstrans q_5$, $q_6 \in \delta(q_5, \ltr{a})$ and $q_6 \epstrans q_6$.
Similarly, $q_6 \in \gamma_1(q_1, \mset{q_3, q_4})$ because $q_1 \epstrans q_1$, $q_2 \in \gamma(q_1, \mset{q_3, q_4})$ and $q_2 \epstrans q_6$.
\end{exa}

We conclude by stating correctness of our translation, in the following sense.

\begin{restatable}{lem}{restateepsiloneliminationcorrectness}%
\label{lemma:epsilon-elimination-correctness}
$A_1$ is bounded, $1$-forking, parsimonious, and weakly implements $A$.
\end{restatable}

\subsection{Removing unary forks}

We now show that, given a fork-acyclic, bounded, parsimonious and $1$-forking PA $A = \angl{Q, F, \delta, \gamma}$, we can implement it using a $2$-forking PA $A_2$ that retains the properties of $A$.
The main idea is to simulate unary forks by keeping a ``call stack'' in the state.
When $A$ follows a unary fork, e.g., $q' \in \gamma(q, \mset{r})$, $A_2$ will push $q'$ on the stack to ``remember'' where we should continue after completing the computation in $r$; once we reach an accepting state, the transitions of $q'$ will become available.
By fork-acyclicity we can bound the height of this stack by the depth of the automaton.

To keep track of this call stack, we need to know which unary forks can occur in sequence from any given state.
This is captured by the following.

\begin{defi}
The relation $\uparrow$, is the least subset of $Q \times Q^*$ satisfying the following rules:
\begin{mathpar}
\inferrule{~}{%
    q \mathrel{\uparrow} q
}
\and
\inferrule{%
    r \mathrel{\uparrow} w \\
    q' \in \gamma(q, \mset{r})
}{%
    q \mathrel{\uparrow} wq'
}
\end{mathpar}
\end{defi}
Intuitively, if $q \mathrel{\uparrow} q_1\cdots{}q_n$, then $q$ can perform a series of unary forks leading to state $q_1$.
Once the computation starting in $q_1$ reaches an accepting state, that state will be at the top of the stack; since it is accepting, we can then pop it off the stack to continue trying to resolve $q_2$, and so on.
The first rule covers the case where no fork takes place, while the second rule allows to extend an existing series of forks with one more.

\begin{figure}
    \begin{subfigure}[b]{0.49\textwidth}
        \centering
        \begin{tikzpicture}
            \node[state] (q1) {$q_1$};
            \node[state,accepting,above=of q1] (q2) {$q_2$};
            \node[state,right=of q1] (q3) {$q_3$};
            \node[state,above=of q3] (q4) {$q_4$};
            \node[state,accepting,right=of q4] (q5) {$q_5$};
            \node[state,right=of q3] (q6) {$q_6$};

            \draw[-] (q1) edge (q3);
            \draw[dashed] (q1) edge[-latex] (q2);
            \draw (q1) + (0:7mm) arc (0:90:7mm);

            \draw[-] (q3) edge (q6);
            \draw[dashed] (q3) edge[-latex] (q4);
            \draw (q3) + (0:7mm) arc (0:90:7mm);

            \draw[-] (q4) edge[-latex] node[above] {$\ltr{c}$} (q5);
            \draw[-] (q3) edge[-latex] node[above,sloped] {$\ltr{a}$} (q2);
            \draw[-] (q6) edge[-latex] node[right] {$\ltr{b}$} (q5);
        \end{tikzpicture}
        \caption{A PA $A$ with unary forks.}%
        \label{figure:unary-fork-removal-before}
    \end{subfigure}
    \begin{subfigure}[b]{0.49\textwidth}
        \centering
        \begin{tikzpicture}
            \begin{scope}[every node/.style={draw,rounded corners=5pt,inner sep=2mm}]
                \node (q1) {$q_1$};
                \node[accepting,above=of q1] (q2q2) {$q_2 q_2$};
                \node[right=of q1] (q5q4q2) {$q_5 q_4 q_2$};
                \node[accepting,above=of q5q4q2] (q5q2) {$q_5 q_2$};
            \end{scope}

            \draw[-] (q1) edge[-latex] node[right] {$\ltr{a}$} (q2q2);
            \draw[-] (q1) edge[-latex] node[above] {$\ltr{b}$} (q5q4q2);
            \draw[-] (q5q4q2) edge[-latex] node[right] {$\ltr{c}$} (q5q2);
        \end{tikzpicture}
        \vspace{3mm}
        \caption{Part of the PA $A_2$ obtained from $A$.}%
        \label{figure:unary-fork-removal-after}
    \end{subfigure}
    \caption{An example of unary fork removal.}%
    \label{figure:unary-fork-removal}
\end{figure}
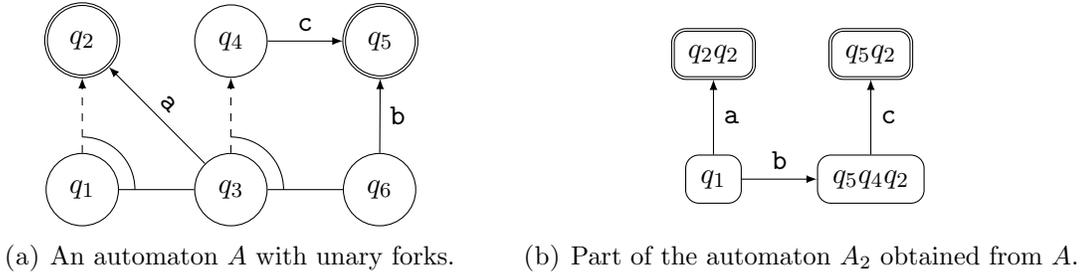

\begin{exa}%
\label{example:fork-stack}
Suppose $A$ is the PA in \cref{figure:unary-fork-removal-before}.
We first note that $q_6 \mathrel{\uparrow} q_6$ by the first rule; hence, since $q_4 \in \gamma(q_3, \mset{q_6})$, we have $q_3 \mathrel{\uparrow} q_6 q_4$ by applying the second rule.
Furthermore, since $q_2 \in \gamma(q_1, \mset{q_3})$, we find that $q_1 \mathrel{\uparrow} q_6 q_4 q_2$, again by the second rule.
Hence, $q_1$ can fork into $q_6$, and after completing a computation there and in $q_4$, we can carry on in $q_2$.
\end{exa}

We can now define our transformation, as follows.
\begin{defi}
  The PA $A_2$ is defined to be $\angl{Q_2, F_2, \delta_2, \gamma_2}$, where $Q_2 = Q^*$ and $F_2 = F^*$.
  Also, $\delta_2$ and $\gamma_2$ are generated by the following rules for $\ltr{a} \in \Sigma$ and $\phi \in \M(Q)$ with $|\phi| \geq 2$:
\begin{mathpar}
\inferrule{%
    q \mathrel{\uparrow} rw \\
    q' \in \delta(r, \ltr{a}) \\
    q'wx \in Q^*
}{%
    q'wx \in \delta_2(qx, \ltr{a})
}
\and
\inferrule{%
    w' \in \delta_2(w, \ltr{a}) \\
    q \in F \\
    qw \in Q^*
}{%
    w' \in \delta_2(qw, \ltr{a})
}
\\
\inferrule{%
    q \mathrel{\uparrow} rw \\
    q' \in \gamma(r, \phi) \\
    q'wx \in Q^*
}{%
    q'wx \in \gamma_2(qx, \phi)
}
\and
\inferrule{%
    w' \in \gamma_2(w, \phi) \\
    q \in F \\
    qw \in Q^*
}{%
    w' \in \gamma_2(qw, \phi)
}
\end{mathpar}
\end{defi}

\noindent
The first rule allows us to look at the state $q$ on top of the stack, and see where it can fork to; if the state $r$ where we end up allows a $\delta$-transition to $q'$, we push $q'$ onto the stack, along with the unresolved states $w$ gained from the unary fork.%
\footnote{%
    Note that since $q \mathrel{\uparrow} q$, we have that $q'x \in \delta_2(qx, \ltr{a})$ whenever $q' \in \delta(q, \ltr{a})$.
}
The third rule works analogously, for (non-unary) parallel transitions.

The second rule says that we can also look at states further down the stack, provided that they are preceded only by accepting states.
This allows us to pop states off the stack when their computation has reached an accepting state, while continuing in the next state.

\begin{exa}
Let $A$ be the PA in \cref{figure:unary-fork-removal-before}.
We have drawn the support of $q_1$ in the automaton $A_2$ obtained from $A$ in \cref{figure:unary-fork-removal-after}.
There, we have $q_2 q_2, q_5 q_2 \in F_2$ because $q_2, q_5 \in F$.
Also, since $q_1 \mathrel{\uparrow} q_3 q_2$ (see \cref{example:fork-stack}) and $q_2 \in \delta(q_3, \ltr{a})$, we have $q_2 q_2 \in \delta_2(q_1, \ltr{a})$ by the first rule above.
Furthermore, since $q_1 \mathrel{\uparrow} q_6 q_4 q_2$ (see \cref{example:fork-stack}) and $q_5 \in \delta(q_6, \ltr{b})$, we have $q_5 q_4 q_2 \in \delta_2(q_1, \ltr{b})$, again by the first rule above.
Lastly, $q_4 \mathrel{\uparrow} q_4$, so $q_5 q_2 \in \delta_2(q_4 q_2, \ltr{c})$ by the first rule.
Since $q_5 \in F$, we find that $q_5 q_2 \in \delta_2(q_5 q_4 q_2, \ltr{c})$ by the second rule.

The effect of the third and fourth rule is not shown, but is largely analogous.
\end{exa}

We again conclude with a statement of correctness of the transformation.

\begin{restatable}{lem}{restateunaryforkremovalcorrectness}%
\label{lemma:unary-fork-removal-correctness}
$A_2$ is bounded, parsimonious, $2$-forking, and implements $A$.
\end{restatable}

\subsection{Flat-branching}%
\label{sec:remove-forks}

In this section, we enforce flat-branching.
We start from a PA $A$ that is assumed to be fork-acyclic, bounded, parsimonious, and $2$-forking, and we want to construct a flat-branching PA $A_3$ that weakly implements $A$, but keeps the properties of $A$.

The first idea of this construction is fairly obvious: to remove chains of forks while retaining the same language, we will saturate the parallel transitions by unfolding every possible chain of forks.
For instance, if $q\in\gamma(p,\mset{q_1,q_2})$ and $\gamma(q_1, \mset{r_1,r_2}) \cap F \neq \emptyset$, we want to have $q \in \gamma_3(p, \mset{r_1,r_2,q_2})$.
Fork-acyclicity is instrumental for this construction to terminate, as fork cycles could introduce infinitely many $\phi \in \M(Q)$ such that $\gamma(p, \phi) \neq \emptyset$.
We make this idea formal as follows:

\begin{defi}
We define $\blacktriangleleft$ as the smallest reflexive relation on $\M(Q)$ satisfying:
\begin{mathpar}
\inferrule{
    \gamma(p, \chi) \cap F \neq \emptyset \\
    \phi \blacktriangleleft \psi \sqcup \mset{p}
}{%
    \phi \blacktriangleleft \psi \sqcup \chi
}
\end{mathpar}
\end{defi}

Intuitively, $\phi \blacktriangleleft \psi$ when a fork into $\phi$ can be expanded to a fork into $\psi$, by forking from one or more of its states, provided the new fork reaches an accepting state.

\begin{exa}
Recall the left half of the automaton in \cref{figure:fork-associativity}, depicted in \cref{figure:flat-branching-construction-before}.
Since $\blacktriangleleft$ is reflexive, we have $\mset{q_3} \blacktriangleleft \mset{q_3}$.
Then, because $\gamma(q_3, \mset{q_5, q_6}) \cap F \neq \emptyset$, we have that $\mset{q_3} \blacktriangleleft \mset{q_5, q_6}$ applying the rule.
Similarly, we have that $\mset{q_1} \blacktriangleleft \mset{q_3, q_4}$.
Combining these two using the rule above then tells us that $\mset{q_1} \blacktriangleleft \mset{q_5, q_6, q_4}$.
Thus, any fork into $\mset{q_1}$ to reach some $q'$ can be expanded to a fork into $\mset{q_5, q_6, q_4}$ to reach $q'$.
\end{exa}

Next, we want to make sure that the original forks cannot be executed in succession anymore, by forcing all forks to expand maximally before continuing with some non-forking transition.
The main idea is to split each state $q$ into $\sst{q}$ and $\pst{q}$.
The state $\sst{q}$ will ensure that $\gamma_3(\sst{q}, \phi) \cap F = \emptyset$ for any multiset $\phi$, i.e., no forks are allowed.
We leverage this property to get flat-branching, by making sure that for any state $p \in Q_3$ of the new automaton, $\gamma_3(p, \phi) \neq \emptyset$ implies that every state in $\phi$ is of the $\sst{q}$ variety.
On the other hand, from the state $\pst{q}$, one cannot perform $\delta$-transitions, and furthermore for any multiset $\phi$ we have $\gamma_3(\pst{q}, \phi) \subseteq \set{\top}$, where $\top$ is the unique accepting state of $A_3$ (as in \cref{sec:parsimony}).

\begin{defi}[$A_3$]
The PA $A_3$ is the quadruple $\angl{Q_3,F_3,\delta_3,\gamma_3}$, where
\begin{mathpar}
Q_3 = \setcompr{\pst{q}}{q\in Q} \cup \setcompr{\sst{q}}{q \in Q} \cup \set \top
\and
F_3 = \set{\top}
\end{mathpar}
with $\top$ a fresh state, such that for all $\ltr{a} \in \Sigma$ and $\phi \in \M(Q_3)$ we have $\delta_3(\top, \ltr{a}) = \gamma_3(\top, \phi) = \emptyset$.
Furthermore, the action of $\delta_3$ and $\gamma_3$ on states different from $\top$ is generated by the following rules for all $\ltr{a} \in \Sigma$ and all $\psi, \phi \in \M(Q)$ with $\psi \blacktriangleleft \phi$:
\begin{mathpar}
\inferrule{%
    q \in \delta(p, \ltr{a})
}{%
    \sst{q}, \pst{q} \in \delta_3(\sst{p}, \ltr{a})
}
\and
\inferrule{%
    \delta(p, \ltr{a}) \cap F \neq \emptyset
}{%
    \top \in \delta_3(\sst{p}, \ltr{a})
}
\and
\inferrule{%
    q \in \gamma(p, \psi)
}{%
    \sst{q}, \pst{q} \in \gamma_3(\sst{p}, \sst{\phi})
}
\and
\inferrule{%
    \gamma(p, \psi) \cap F \neq \emptyset
}{%
    \top \in \gamma_3(\pst{p}, \sst{\phi})
}
\end{mathpar}
in which $\sst{\phi} = \mset{\sst{q_1}, \dots, \sst{q_n}}$ whenever $\phi = \mset{q_1, \dots, q_n}$.
\end{defi}

\begin{figure}
    \begin{subfigure}[b]{0.49\textwidth}
        \centering
        \begin{tikzpicture}
          \node[state] (q1) {$q_1$};
          \node[state,above right=7mm of q1] (q3) {$q_3$};
          \node[state,below right=7mm of q1] (q4) {$q_4$};
          \node[state,above=7mm of q3] (q5) {$q_5$};
          \node[state,right=7mm of q3] (q6) {$q_6$};
          \node[state,accepting,below right=7mm of q6] (q2) {$q_2$};

          \draw[-] (q1) edge (q3);
          \draw[-] (q1) edge (q4);
          \draw[-] (q3) edge (q5);
          \draw[-] (q3) edge (q6);

          \draw[dashed] (q1) edge[-latex] (q2);
          \draw[dashed,rounded corners=5pt,-latex] (q3) -- +(1,1) -| (q2);

          \draw (q1) + (-45:7mm) arc (-45:45:7mm);
          \draw (q3) + (0:7mm) arc (0:90:7mm);

          \draw (q6) edge[-latex] node[above right] {$\ltr{b}$} (q2);
          \draw[rounded corners=5pt,-latex] (q4) -| node[above,xshift=-4em] {$\ltr{c}$} (q2);
          \draw[rounded corners=5pt,-latex] (q5) -| +(3.8,-1) node[right,yshift=-1.5em] {$\ltr{a}$} |- (q2.east);
        \end{tikzpicture}
        \caption{A PA $A$ with nested forks.}%
        \label{figure:flat-branching-construction-before}
    \end{subfigure}
    \begin{subfigure}[b]{0.49\textwidth}
        \centering
        \begin{tikzpicture}
          \node[state] (pq1) {$\pst{q_1}$};
          \node[state,below right=7mm of pq1] (sq4) {$\sst{q_4}$};
          \node[state,above=7mm of pq1] (sq5) {$\sst{q_5}$};
          \node[state,accepting,right=20mm of pq1] (top) {$\top$};
          \node[state,above=7mm of top] (sq6) {$\sst{q_6}$};

          \node[state,left=15mm of sq5] (sq2) {$\sst{q_2}$};
          \node[state,left=15mm of pq1] (pq2) {$\pst{q_2}$};

          \draw[-] (pq1) edge (sq4);
          \draw[-] (pq1) edge (sq5);
          \draw[-] (pq1) edge (sq6);

          \draw[dashed] (pq1) edge[-latex] (top);

          \draw (pq1) + (-45:7mm) arc (-45:90:7mm);

          \draw (sq6) edge[-latex] node[right] {$\ltr{b}$} (top);
          \draw (sq5) edge[-latex] node[above] {$\ltr{a}$} (sq2);
          \draw (sq5) edge[-latex] node[above,sloped] {$\ltr{a}$} (pq2);
          \draw (sq5) edge[-latex] node[above,sloped,pos=0.25] {$\ltr{a}$} (top);
          \draw (sq4) edge[-latex] node[above,sloped] {$\ltr{c}$} (top);
        \end{tikzpicture}
        \vspace{1em}
        \caption{Part of the PA $A_3$ obtained from $A$.}%
        \label{figure:flat-branching-construction-after}
    \end{subfigure}
    \caption{Example of construction to ensure flat-branching.}%
    \label{figure:flat-branching-construction}
\end{figure}
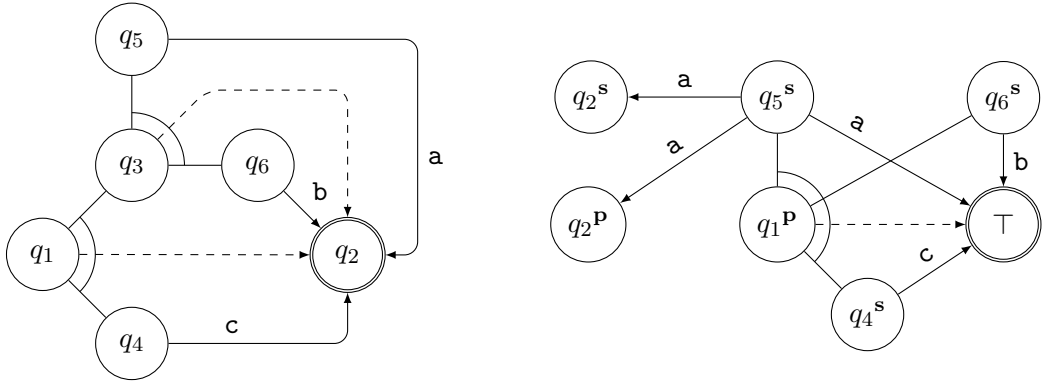

\begin{exa}
Let $A$ be the PA in \cref{figure:flat-branching-construction-before}.
Part of the support of $\pst{q_1}$ is drawn in \cref{figure:flat-branching-construction-after}.
There, we find that, since $q_2 \in \gamma(q_1, \mset{q_3, q_4}) \cap F$ and $\mset{q_3, q_4} \blacktriangleleft \mset{q_4, q_5, q_6}$ (refer to the previous example), also $\top \in \gamma_3(\pst{q_1}, \mset{\sst{q_4}, \sst{q_5}, \sst{q_6}})$ by the last rule above.
Furthermore, since $q_2 \in \delta(q_5, \ltr{a}) \cap F$, we have that $\top \in \delta_3(\sst{q_5}, \ltr{a})$ by the second rule above, but also $\sst{q_2}, \pst{q_2} \in \delta_3(\sst{q_5}, \ltr{a})$, by the first rule above.
Not drawn are the transitions $\sst{q_2}, \pst{q_2} \in \gamma_3(\sst{q_1}, \mset{\sst{q_3}, \sst{q_4}})$ which would result from applying the third rule.
These transitions do not contribute to the language, since they lead to states from which no accepting run is possible.
\end{exa}

We conclude this transformation by stating the desired correctness.

\begin{restatable}{lem}{restateflatbranchingcorrectness}%
\label{lem:flat-branching}
$A_3$ is bounded, $2$-forking, parsimonious, flat-branching (i.e., well-structured) and weakly implements $A$.
\end{restatable}

In total, we have proved the following theorem.
\begin{thm}
Let $A$ be a finite and fork-acyclic PA\@.
We can construct a finite and fork-acyclic PA $A'$ that is well-structured and implements $A$.
\end{thm}

The complexity of the construction above is hard to pin down, because the number of states supporting any single state in the well-nested automaton is obscured by the infinite nature of the intermediate automata.
Nevertheless, we anticipate that the removal of unary forks can be expected to provide the worst blowup in automaton size (possibly exponential in $D_A$), as the number of states grows to accommodate all possible ``call stacks''.
The removal of nullary forks and the flat-branching construction both contribute a relatively small number of new states; the same is true for parsimonification, although this step may add a significant number of parallel transitions.

\section{From expressions to automata}%
\label{section:expressions-to-automata}

Recall that sp-languages can be defined denotationally, in terms of a sr-expression $e \in \terms$.
In this section, we show how to obtain a fork-acyclic and finite PA that accepts the semantics of $e$, starting at some state.
Concretely, we define a bounded pomset automaton where \emph{every} sr-expression is a state, and the language of this state is intended to be the semantics of said sr-expression.
The transition functions are set up such that, starting from $e \in \terms$, reading a particular pomset brings us to $e' \in \terms$, which describes the sr-language that remains to be read.
This generalizes Antimirov's (partial) derivatives of rational expressions~\cite{antimirov-1996}; we therefore refer to the transition functions on expressions as \emph{derivatives}.

Before we define the derivatives, it is convenient to introduce two shorthands.
The first of these allows us to include a set of terms based on whether or not another term is in $\sacc$, the set of sr-expressions that accept the empty pomset (c.f. \cref{definition:accepting-terms}).
The second shorthand allows us to concatenate a set of terms with another term on the right.

\begin{defi}[Derivatives]
We define $\sderiv: \terms \times \Sigma \to 2^\terms$ and $\pderiv: \terms \times \M(\terms) \to 2^\terms$ inductively.
To this end, the following shorthands are useful, for all $e \in \terms$ and $T \subseteq \terms$:
\begin{mathpar}
e \star T =
    \begin{cases}
    T & e \in \sacc \\
    \emptyset & \text{otherwise}
    \end{cases}
\and
T \fatsemi e = \setcompr{f \cdot e}{f \in T}
\end{mathpar}
This then allows us to define $\sderiv$ and $\pderiv$ inductively, as follows:
\begin{align*}
\sderiv(0, \ltr{a}) &= \emptyset
    & \pderiv(0, \phi) &= \emptyset \\
    \sderiv(1, \ltr{a}) &= \emptyset
    & \pderiv(1, \phi) &= \emptyset \\
    \sderiv(\ltr{b}, \ltr{a}) &= \setcompr{1}{\ltr{a} = \ltr{b}}
    & \pderiv(\ltr{b}, \phi) &= \emptyset \\
    \sderiv(e + f, \ltr{a}) &= \sderiv(e, \ltr{a}) \cup \sderiv(f, \ltr{a})
    & \pderiv(e + f, \phi) &= \pderiv(e, \phi) \cup \pderiv(f, \phi) \\
    \sderiv(e \cdot f, \ltr{a}) &= \sderiv(e, \ltr{a}) \fatsemi f \; \cup \; e \star \sderiv(f, \ltr{a})
    & \pderiv(e \cdot f, \phi) &= \pderiv(e, \phi) \fatsemi f \; \cup \; e \star \pderiv(f, \phi) \\
    \sderiv(e \parallel f, \ltr{a}) &= \emptyset
    & \pderiv(e \parallel f, \phi) &= \setcompr{1}{\phi = \mset{e, f}} \\
    \sderiv(e^*, \ltr{a}) &= \sderiv(e, \ltr{a}) \fatsemi e^*
    & \pderiv(e^*, \phi) &= \pderiv(e, \phi) \fatsemi e^*
  \end{align*}
\end{defi}

We can now define our pomset automaton that operates on  terms.

\begin{defi}
  The \emph{syntactic PA}, denoted by $A_\Sigma$, is the PA $\angl{\terms, \sacc, \sderiv, \pderiv}$.
\end{defi}

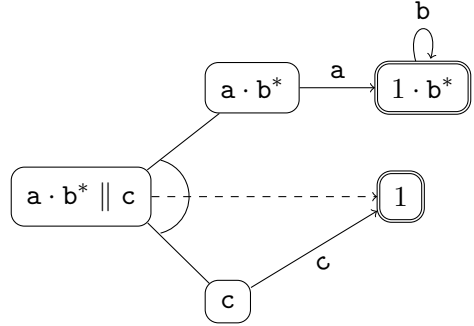
\begin{figure}
    \begin{tikzpicture}
        \begin{scope}[every node/.style={draw,rounded corners=5pt,inner sep=2mm}]
            \node (e2) {$\ltr{a} \cdot \ltr{b}^* \parallel \ltr{c}$};
            \node[above right=1cm of e2] (e3) {$\ltr{a} \cdot \ltr{b}^*$};
            \node[accepting,right=1cm of e3] (e4) {$1 \cdot \ltr{b}^*$};
            \node[below right=1cm of e2] (e5) {$\ltr{c}$};
            \node[accepting,right=3cm of e2] (e6) {$1$};
        \end{scope}

        \draw[-,shorten <=-2pt] (e2.north east) edge (e3);
        \draw[-,shorten <=-2pt,shorten >=-2pt] (e2.south east) edge (e5);
        \draw (e2.east) + (-76:5mm) arc (-76:76:5mm);
        \draw[dashed,->] (e2) edge (e6);
        \draw[->] (e3) edge node[above] {$\ltr{a}$} (e4);
        \draw[->] (e5) edge node[below,sloped] {$\ltr{c}$} (e6);
        \draw[->] (e4) edge[loop above] node[above] {$\ltr{b}$} (e4);
    \end{tikzpicture}
    \caption{Part of the syntactic pomset automaton.}\label{figure:example-syntactic-pa}
\end{figure}

We simplify subscripts, writing $\rightarrow_\Sigma$ instead of $\rightarrow_{A_\Sigma}$, and so forth.

\begin{exa}
We have drawn part of the syntactic PA, specifically the support of $\ltr{a} \cdot \ltr{b}^* \parallel \ltr{c}$, in \cref{figure:example-syntactic-pa}.
There, we see that $1 \cdot \ltr{b}^*$ is an accepting state, because $1, \ltr{b}^* \in \sacc$.
The sequential transitions are generated by $\sderiv$; for instance, $1 \cdot \ltr{b}^* \in \sderiv(\ltr{a} \cdot \ltr{b}^*, \ltr{a})$, because $1 \in \sderiv(\ltr{a}, \ltr{a})$ and $\sderiv(\ltr{a}, \ltr{a}) \fatsemi \ltr{b}^* \subseteq \sderiv(\ltr{a} \cdot \ltr{b}^*, \ltr{a})$; also, $1 \cdot \ltr{b}^* \in \sderiv(1 \cdot \ltr{b}^*, \ltr{b})$, because $1 \in \sderiv(\ltr{b}, \ltr{b})$, and $\sderiv(\ltr{b}, \ltr{b}) \fatsemi \ltr{b}^* \subseteq \sderiv(1 \cdot \ltr{b}^*, \ltr{b})$.
Lastly, $1 \in \pderiv(\ltr{a} \cdot \ltr{b}^* \parallel \ltr{c}, \mset{\ltr{a} \cdot \ltr{b}^*, \ltr{c}})$ by definition of $\pderiv$.
\end{exa}

To show that the language of a state $e$ in the syntactic PA is exactly $\sem{e}$, we first need to show a number of intermediary lemmas that analyze runs that can occur.

\subsection{Deconstruction lemmas}
The first set of lemmas are called \emph{deconstruction lemmas}: they tell us how, given a run originating in $e$, we can obtain one or more runs of terms originating in subterms of $e$.
First, we show how to do this for sums.

\begin{lem}%
  \label{lemma:run-deconstruct-plus}
  Let $e_1, e_2 \in \terms$, $f \in \sacc$ and $U \in \SP(\Sigma)$, such that $e_1 + e_2 \sarun{U} f$.
  There exists an $f' \in \sacc$ with $e_1 \sarun{U} f'$ or $e_2 \sarun{U} f'$.
\end{lem}
\begin{proof}
  Let $\ell$ be the length of $e_1 + e_2 \sarun{U} f$.
  There are two cases to consider.
  \begin{itemize}
  \item
    If $\ell = 0$, then $e_1 + e_2 \sarun{U} f$ is trivial.
    In that case, $f = e_1 + e_2$, and so $e_i \in \sacc$ for some $i \in \{1,2\}$; if we choose $f' = e_i$, we find $e_i \sarun{U} f'$, and the claim is satisfied.

  \item
    Otherwise, if $\ell > 0$, then $U = U_0 \cdot U'$ and there exists a $g \in \terms$ such that $e_1 + e_2 \sarun{U_0} g$ is a unit run, and $g \sarun{U'} f$ is of length $\ell - 1$.
    If $e_1 + e_2 \sarun{U_0} g$ is a sequential unit run, then $U_0 = \ltr{a}$ for some $\ltr{a} \in \Sigma$, and $g \in \sderiv(e_1 + e_2, \ltr{a})$.
    Without loss of generality, let $g \in \sderiv(e_1, \ltr{a})$; in that case, $e_1 \sarun{U_0} g$.
    If we then choose $f' = f$, we find that $e_1 \sarun{U} f'$.
    The case where $e_1 + e_2 \sarun{U_0} g$ is a parallel unit run is similar.
    \qedhere
  \end{itemize}
\end{proof}

\noindent
Next, we show how a run that starts in a sequential composition gives rise to runs that originate in each of the components.
In this case, it is useful to also relate the length of the original run to the lengths of the deconstructed runs, so that we can use this result in the analogous deconstruction result for the Kleene star (\cref{lemma:run-deconstruct-star}).

\begin{restatable}{lem}{restaterundeconstructsequential}%
\label{lemma:run-deconstruct-sequential}
Let $e_1, e_2 \in \terms$, $f \in \sacc$ and $U \in \SP(\Sigma)$, with $e_1 \cdot e_2 \sarun{U} f$ (of length $\ell$).
Then $U = U_1 \cdot U_2$ and there exist $f_1, f_2 \in \sacc$ with $e_i \sarun{U_i} f_i$ for $i \in \set{1, 2}$ (of length at most $\ell$).
\end{restatable}

There is also a deconstruction lemma for runs originating in a parallel composition, where the pomset accepted is a parallel composition of pomsets accepted by the operands.

\begin{restatable}{lem}{restaterundeconstructparallel}%
\label{lemma:run-deconstruct-parallel}
Let $e_1, e_2 \in \terms$, $f \in \sacc$ and $U \in \SP(\Sigma)$ such that $e_1 \parallel e_2 \sarun{U} f$.
Then $U = U_1 \parallel U_2$, and there exist $f_1, f_2 \in \sacc$ such that $e_i \sarun{U_i} f_i$ for $i \in \set{1, 2}$.
\end{restatable}

For the last deconstruction lemma, we consider terms with the Kleene star as the topmost operator.
Here, we show how to obtain a number of runs, each of which originates in the expression directly below the Kleene star.

\begin{lem}%
\label{lemma:run-deconstruct-star}
Let $e \in \terms$, $f \in \sacc$ and $U \in \SP(\Sigma)$ be such that $e^* \sarun{U} f$.
There exist $f_1, f_2, \dots, f_n \in \sacc$ such that $U = U_1 \cdot U_2 \cdots U_n$ and for $1 \leq i \leq n$ it holds that $e \sarun{U_i} f_i$.
\end{lem}
\begin{proof}
The proof proceeds by induction on the length $\ell$ of $e^* \sarun{U} f$.
In the base, where $\ell = 0$, we have that $f = e^*$ and $U = 1$; it suffices to choose $n = 0$.

In the inductive step, let $e^* \sarun{U} f$ be of length $\ell + 1$.
We find $g \in \terms$ and $U = U_0 \cdot U'$ such that $e^* \sarun{U_0} g$ is a unit run, and $g \sarun{U'} f$ is of length $\ell$.
If $e^* \sarun{U_0} g$ is a sequential unit run, then $U_0 = \ltr{a}$ for some $\ltr{a} \in \Sigma$, and $g \in \sderiv(e^*, \ltr{a}) = \sderiv(e, \ltr{a}) \fatsemi e^*$.
Hence, $g = g' \cdot e^*$, with $g' \in \sderiv(e, \ltr{a})$.
By \cref{lemma:run-deconstruct-sequential}, we find $f'', f' \in \sacc$ such that $U' = V \cdot W$ as well as $g' \sarun{V} f''$ and $e^* \sarun{W} f'$, of length at most $\ell$.
By induction, we find $f_2, f_3, \dots, f_n \in \sacc$ such that $W = U_2 \cdot U_3 \cdots U_n$, and for $1 < i \leq n$ we have $e \sarun{U_i} f_i$.

We then choose $f_1 = f''$ and $U_1 = U_0 \cdot V$.
For these choices, $U = U_0 \cdot U' = U_0' \cdot V \cdot W = U_1 \cdot U_2 \cdots U_n$.
Since $e \sarun{U_0} \sderiv(e, \ltr{a}) \sarun{V} f'$, we also have $e \sarun{U_1} f_1$, satisfying the claim.

The case where $e^* \sarun{U_0} g$ is a parallel unit run is similar.
\end{proof}

\subsection{Construction lemmas}
Next, we prove the \emph{construction lemmas}, so called because they tell us how to combine runs that originate in some terms into runs that originate in some composition of those terms.
Keeping the same order of operators as before, we start by showing how to construct runs originating from terms that are sums.

\begin{lem}%
  \label{lemma:run-construct-plus}
  Let $e_1, e_2 \in \terms$, $f_1, f_2 \in \sacc$ and $U \in \SP(\Sigma)$ be such that $e_1 \sarun{U} f_1$ or $e_2 \sarun{U} f_2$.
  There exists an $f \in \sacc$ such that $e_1 + e_2 \sarun{U} f$.
\end{lem}
\begin{proof}
  We treat the case where $e_1 \sarun{U} f_1$; the case where $e_2 \sarun{U} f_2$ can be treated similarly.
  There are two possibilities to consider, based on the length $\ell$ of $e_1 \sarun{U} f_1$.
  \begin{itemize}
  \item
    If $\ell = 0$, then $f_1 = e_1$ and $U = 1$.
    We then choose $f = e_1 + e_2$.

  \item
    If $\ell > 0$, then we find $e_1' \in \terms$ and $U = U_0 \cdot U'$ such that $e_1 \sarun{U_0} e_1'$ is a unit run, and $e_1' \sarun{U'} f_1$.
    If $e_1 \sarun{U_0} e_1'$ is a sequential unit run, then $U_0 = \ltr{a}$ for some $\ltr{a} \in \Sigma$, and $e_1' \in \sderiv(e_1, \ltr{a})$.
    But then $e_1' \in \sderiv(e_1 + e_2, \ltr{a})$ as well, and hence $e_1 + e_2 \sarun{U_0} e_1'$.
    Putting this together, we find that $e_1 + e_2 \sarun{U} f_1$; choosing $f = f_1$ then satisfies the claim.

    The case where $e_1 \sarun{U_0} e_1'$ is a parallel unit run is similar.
    \qedhere
  \end{itemize}
\end{proof}

\noindent
The construction lemma for sequential composition says that we can compose runs to accepting states in the following way.

\begin{restatable}{lem}{restaterunconstructsequential}%
\label{lemma:run-construct-sequential}
Let $e_1, e_2 \in \terms$, $f_1, f_2 \in \sacc$ and $U, V \in \SP(\Sigma)$ such that $e_1 \sarun{U} f_1$ and $e_2 \sarun{V} f_2$.
There exists an $f \in \sacc$ with $e_1 \cdot e_2 \sarun{U \cdot V} f$.
\end{restatable}

The next operator is parallel composition; here, the construction lemma is easy.

\begin{restatable}{lem}{restaterunconstructparallel}%
\label{lemma:run-construct-parallel}
Let $e_1, e_2 \in \terms$, $f_1, f_2 \in \sacc$, and $U, V \in \SP(\Sigma)$ such that $e_1 \sarun{U} f_1$ and $e_2 \sarun{V} f_2$.
Then $e_1 \parallel e_2 \sarun{U \parallel V} 1$.
\end{restatable}

The last construction lemma concerns the Kleene star.
Here, a succession of runs originating in the expression below the star allows us to create a run originating in the starred expression, provided each of these runs reaches an accepting state.

\begin{lem}%
\label{lemma:run-construct-star}
Let $e \in \terms$, $f_1, f_2, \dots, f_n \in \sacc$, and $U_1, U_2, \dots, U_n \in \SP(\Sigma)$ be such that for $1 \leq i \leq n$ it holds that $e \sarun{U_i} f_i$.
There exists an $f \in \sacc$ such that $e^* \sarun{U_1 \cdot U_2 \cdots U_n} f$.
\end{lem}
\begin{proof}
We can assume w.l.o.g.\ that for $0 \leq i < n$ it holds that $e \sarun{U_i} f_i$ is non-trivial.
We proceed by induction on $n$.
In the base, where $n = 0$, we can choose $f = e^*$.

For the inductive step, let $n > 0$ and assume that the claim holds for $n-1$.
By induction, we find $f' \in \sacc$ with $e^* \sarun{U_2 \cdot U_3 \cdots U_n} f'$.
Since $e \sarun{U_1} f_1$ is non-trivial, we find $e' \in \terms$ and $U_1 = U_0 \cdot U_1'$ such that $e \sarun{U_0} e'$ is a unit run, and $e' \sarun{U_1'} f_1$.
By \cref{lemma:run-construct-sequential}, we find $f \in \sacc$ such that $e' \cdot e^* \sarun{U_1' \cdot U_2 \cdot U_3 \cdots U_n} f$.
If $e \sarun{U_0} e'$ is a sequential unit run, then $U_0 = \ltr{a}$ for some $\ltr{a} \in \Sigma$, and $e' \in \sderiv(e, \ltr{a})$.
But then $e' \cdot e^* \in \sderiv(e^*, \ltr{a})$, and hence $e^* \sarun{U_0} e' \cdot e^*$.
Thus, $e^* \sarun{U_0} e' \cdot e^* \sarun{U_1' \cdot U_2 \cdot U_3 \cdots U_n} f$ and therefore $e^* \sarun{U_1 \cdot U_2 \cdots U_n} f$.

The case where $e \sarun{U_0} e'$ is a parallel unit run is similar.
\end{proof}

\subsection{Correctness of the syntactic PA}
The (de)construction lemmas combine to prove the equations claimed by the lemma below; the deconstruction lemmas show inclusion from left to right, whereas construction lemmas show the inclusion from right to left. 

\begin{restatable}{lem}{restatebrzozowskimorphism}%
  \label{lemma:brzozowski-morphism}
  Let $e, f \in \terms$.
  The following hold:
  \begin{align*}
    L_\Sigma(0) &= \emptyset
    & L_\Sigma(e + f) &= L_\Sigma(e) \cup L_\Sigma(f)
    & L_\Sigma(e^*) &= {L_\Sigma(e)}^* \\
    L_\Sigma(1) &= \set{1}
    & L_\Sigma(e \cdot f) &= L_\Sigma(e) \cdot L_\Sigma(f) \\
    L_\Sigma(\ltr{a}) &= \set{\ltr{a}}
    & L_\Sigma(e \parallel f) &= L_\Sigma(e) \parallel L_\Sigma(f)
  \end{align*}
\end{restatable}

\noindent
A straightforward inductive argument now helps us validate the following:

\begin{restatable}{lem}{restatebrzozowskisoundness}%
\label{lemma:brzozowski-soundness}
For all $e \in \terms$, we have $L_\Sigma(e) = \sem{e}$.
\end{restatable}

This covers the correctness of our expressions-to-automata construction as far as the languages are concerned.
It remains to show that the automaton is also fork-acyclic, which comes down to showing that if $e \preceq_\Sigma f$, then the nesting depth of $\parallel$ in $e$ is bounded by that of $f$.
Since forks only occur to expressions with strictly lower parallel nesting depth, neither $e$ nor $f$ can support $e \parallel f$.

\begin{restatable}{lem}{restatesyntacticpaforkacyclic}%
  \label{lemma:syntactic-pa-fork-acyclic}
  The syntactic pomset automaton is fork-acyclic.
\end{restatable}
\begin{proof}[Proof sketch]
We construct a function $\pdepth{-}: \terms \to \naturals$, and we show that if $e \preceq_\Sigma f$, then $\pdepth{e} \leq \pdepth{f}$.
In the process, we also show that if $e, f \in \terms$ and $\phi \in \M(\terms)$ with $\gamma(e, \phi) \neq \emptyset$ and $f \in \phi$, then $\pdepth{f} < \pdepth{e}$; together with the earlier proprety of $\pdepth{-}$, this tells us that $f \prec_\Sigma e$, as desired.
\end{proof}

Furthermore, each state in the syntactic PA has finite support.
This is done by constructing a (finite) overestimation of the expressions that support $e$ in $A_\Sigma$.

\begin{restatable}{lem}{restatesyntacticpabounded}%
  \label{lemma:syntactic-pa-bounded}
  The syntactic pomset automaton is bounded.
\end{restatable}
\begin{proof}[Proof sketch]
  We should show that for $e \in \terms$, it holds that $\ssupp(e)$ is finite.
  Since $\ssupp(e)$ is the smallest closed set that contains $e$, it suffices to find a finite closed set $S(e)$ such that $e \in S(e)$.
  The construction of this set and the verification of its closure is entirely routine; we refer to the appendix for details.
\end{proof}

This then allows us to state one half of our Kleene theorem, as follows.

\begin{thm}[Expressions to automata]%
\label{theorem:expressions-to-automata}
Let $e \in \terms$.
We can obtain a fork-acyclic and finite PA $A$ with a state $q$ such that $L_A(q) = \sem{e}$.
\end{thm}
\begin{proof}
Choose $A = A_\Sigma[e]$.
By \cref{lemma:syntactic-pa-bounded,lemma:syntactic-pa-fork-acyclic} as well as \cref{lemma:restrict-bounded}, $A$ is finite and fork-acyclic.
Finally, by \cref{lemma:restrict-bounded,lemma:brzozowski-soundness}, we have that $L_{A[\pi_\Sigma(e)]}(e) = L_\Sigma(e) = \sem{e}$.
\end{proof}

It also shows that equivalence of sr-expressions is decidable.

\begin{cor}
  Given $e, f \in \terms$, it is decidable whether $\sem{e} = \sem{f}$.
\end{cor}

\section{From automata to expressions}%
\label{section:automata-to-expressions}

Let $A=\angl{Q, F, \delta, \gamma}$ be a finite and fork-acyclic pomset automaton.
In this section, we set out to obtain the converse of the result in the previous section, that is, to construct for every state $q \in Q$ of the automaton $A$ a series-rational expression $e_q$ such that $\sem{e_q} = L_A(q)$.
Since $A$ is finite and fork-acyclic, the strict support relation $\prec_A$ is well-founded.
We can therefore proceed building $e_q$ by $\prec_A$-induction.
Our \emph{main induction hypothesis} is as follows:

\begin{quote}
If $r \in Q$ with $r \prec_A q$, then there exists an sr-expression $e_r$ with $\sem{e_r} = L_A(r)$.
\end{quote}

\noindent
Our task is to define for $q$ a series-rational expression $e_q$ such that $\sem{e_q} = L_A(q)$.
As a matter of fact, it is convenient to establish a slightly more general result: we let
\begin{equation*}
Q' = \setcompr{p \in Q}{q \preceq_A p \preceq_A q}
\end{equation*}
and simultaneously define for every $p \in Q'$ an sr-expression $e_p$ such that $\sem{e_p} = L_A(p)$.

To define these expressions, we use an approach due to McNaughton and Yamada~\cite{mcnaughton-yamada-1960}.
The first step in this approach is to define, for all $p \in Q'$ and $r \in Q$, an sr-expression $e_{pr}$ that captures the pomsets that can be read when transitioning from $p$ to $r$ by means of a (sequential or parallel) unit run.
This can be done as follows:
\[
    e_{pr} =
        \sum_{r \in \delta(p,\ltr{a})} \ltr{a} +
        \sum_{r \in \gamma(p,\mset{s_1,\dots,s_n})} e_{s_1} \parallel \cdots \parallel e_{s_n}
\]
Intuitively, the first sum accounts for sequential unit runs: if $r \in \delta(p, \ltr{a})$, then $p \arun{\ltr{a}}_A r$ is a unit run.
For the second term, we note that if $s_1, \dots, s_n \in Q$ and $r \in \gamma(p, \mset{s_1, \dots, s_n})$, then $s_1, \dots, s_n \prec_A p$ by fork-acyclicity; therefore, by the main induction hypothesis, $e_{s_i}$ is already defined such that $L_A(s_i) = \sem{e_{s_i}}$.
Thus, if for $1 \leq i \leq n$, we have $U_i \in \sem{e_{s_i}}$, then the pomset $U_i$ is accepted by $s_i$, which means that $p \arun{U_1 \parallel \dots \parallel U_n}_A r$ is a parallel unit run.

The following asserts that $e_{pr}$ is fit for purpose:
\begin{restatable}{lem}{restateautomatatoexpressionssmallstep}%
\label{lemma:automata-to-expressions-small-step}
Let $p \in Q'$ and $r \in Q$.
Now $U \in \sem{e_{pr}}$ if and only if $p \arun{U}_A r$ is a unit run.
\end{restatable}

The next step is to extend the expression $e_{pr}$ into an expression that describes the labels of (non-trivial) runs that may pass through intermediate states in $Q'$.
More precisely, for $S \subseteq Q'$, we choose $e_{pr}^S$ to describe the language of pomsets that can be read while transitioning from $p$ to $r$, while passing through states in $S$ as intermediates.%
\footnote{The inductive step is ambiguous, in that we do not specify \emph{which} $s \in S$ to choose --- any choice will do.}
\[
    e_{pr}^S =
    \begin{cases}
    e_{pr} & S = \emptyset \\[2mm]
    e_{pr}^{S'} + e_{ps}^{S'} \cdot {\left(e_{ss}^{S'}\right)}^{*} \cdot e_{sr}^{S'} & S = S' \cup \set{s},\ s \not\in S'
    \end{cases}
\]
Intuitively, when $S = \emptyset$, no intermediate states are allowed, so we need a unit run from $p$ to $r$.
When $S$ is non-empty, we single out the state $s$; a run from $p$ to $r$ either does not pass through $s$ (the first term), or it passes through $s$ at least once, where it may loop several times, before continuing on to $r$ (the second term).

The following asserts the correctness of $e_{pr}^{S}$ w.r.t.\ our intention:

\begin{lem}%
\label{lemma:automata-to-expressions-big-step}
Let $p \in Q'$, $r \in Q$, and $S \subseteq Q'$.
Now it holds that $U \in \sem{e_{pr}^S}$ if and only if there exist $\ell \geq 1$ and $s_1, \dots, s_{\ell-1} \in S$ and $U = U_1 \cdots U_n$, such that
\[
    p = s_0 \arun{U_1}_A s_1 \arun{U_2}_A s_2 \arun{U_3}_A \cdots \arun{U_{\ell-1}}_A s_{\ell-1} \arun{U_\ell}_A s_\ell = r
\]
in which each of the above is a unit run.
\end{lem}
\begin{proof}
We proceed by induction on $S$.
The case where $S = \emptyset$, is covered by \cref{lemma:automata-to-expressions-small-step}.

In the inductive step, let $S = S' \cup \set{s}$, with $s \not\in S'$.
For the implication from left to right, suppose $U \in \sem{e_{pr}^S}$.
Then, there are two cases:
\begin{itemize}
    \item
    If $U\in\sem{e^{S'}_{pr}}$, then the claim follows immediately by induction.

    \item
    If $U\in\sem{e^{S'}_{ps} \cdot {(e^{S'}_{ss})}^* \cdot e^{S'}_{sr}}$, then, in accordance with the semantics of sequential composition and Kleene star, there exist $U_0 \in \sem{e^S_{ps}}$, $U_1, \dots, U_n \in \sem{e^{S'}_{ss}}$ and $U_{n+1} \in \sem{e^{S'}_{sr}}$ such that $U = U_0 \cdot U_1 \cdots U_n \cdot U_{n+1}$.
    With $n+2$ applications of the induction hypothesis we then we get an appropriate sequence of unit runs.
\end{itemize}

\noindent
For the converse implication, we can identify all occurrences of $s$ among the $s_1, \dots, s_{\ell-1}$, and partition the sequence into subsequences that do not travel through $s$ as an intermediate state.
We can then apply the induction hypothesis to find that $U \in \sem{e^{S'}_{ps} \cdot {(e^{S'}_{ss})}^* \cdot e^{S'}_{sr}}$.
\end{proof}

The final step is to create an expression $e_p$ that characterises $L_A(p)$, i.e., the pomsets $U$ such that $p \arun{U}_A r$ for some $r \in F$.
Note that, if this run has exited $Q'$ to reach $s$, then $p \prec_A s$; in particular, this means that once a run has left $Q'$, it cannot pass through (or fork into) states in $Q'$ anymore.
Therefore, $p \arun{U}_A r$ only passes through states in $Q'$, or $p \arun{U}_A r$ eventually leaves $Q'$, never to return again.

This guides us to define $e_p$ as follows:
\begin{mathpar}
e_p =
    \sum_{r \in Q' \cap F} e_{pr}^{Q'} +
    \sum_{r \prec_A p} e_{pr}^{Q'} \cdot e_r +
    \begin{cases}
    1 & p \in F \\
    0 & \text{otherwise}
    \end{cases}
\end{mathpar}
Intuitively, the first term accounts for the possibility that the run from $p$ to some accepting state stays within $Q'$; the second term incorporates that we can reach some $r \prec_A p$ (for which we already have the expression $e_r$, by the main induction hypothesis), whence we reach an accepting state.
The last term is there to allow for $p$ to accept immediately.

We can then prove that $e_p$ indeed describes the language of $p$.

\begin{lem}%
\label{lemma:automata-to-expressions-correctness}
For all $p \in Q'$ it holds that $\sem{e_p} = L_A(p)$.
\end{lem}
\begin{proof}
For the inclusion from left to right, suppose that $U \in \sem{e_p}$; we distinguish three cases:
\begin{itemize}
    \item
    If $U \in \sem{e^{Q'}_{pr}}$ for some $r \in Q'\cap F$, then by \cref{lemma:automata-to-expressions-big-step} we have that $p \arun{U}_A r$ and hence, since $r \in F$, we find that $U \in L_A(p)$.

    \item
    If $U \in \sem{e^{Q'}_{pr} \cdot e_r}$ for some $r \prec_A p$, then there exist pomsets $V$ and $W$ such that $U = V \cdot W$ with $V \in \sem{e^{Q'}_{pr}}$ and $W \in \sem{e_r}$.
    From $V \in \sem{e^{Q'}_{pr}}$ it follows by \cref{lemma:automata-to-expressions-big-step} that $p \arun{V}_A r$; from $W \in \sem{e_p}$ it follows by the main induction hypothesis that $W \in L_A(r)$, and hence $r \arun{W}_A r'$ for some $r' \in F$.
    In total, we have that $p \arun{V \cdot W}_A r'$, and therefore $U = V \cdot W \in L_A(p)$.

    \item
    If $U = 1$ and $p \in F$, then $p \arun{U} p \in F$, and hence $U \in L_A(p)$ immediately.
\end{itemize}

\noindent
For the other inclusion, suppose that $U \in L_A(p)$.
Then there exists $r \in F$ such that $p \arun{U}_A r$, and hence, by \cref{lemma:run-deconstruct}, there exist states $p = s_0, \dots, s_n = r$ and pomsets $U_1, \dots, U_n$ such that $U = U_1 \cdots U_n$ and $s_{i-1} \arun{U_i}_A s_i$ is a unit run for all $1 \leq i \leq n$.
If $n = 0$, then $p \arun{U} r$ is trivial, meaning $p \in F$ and $U = 1$, and hence $U \in \sem{e_p}$ immediately.
If $n \geq 1$ and $s_i \in Q'$ for all $0 \leq i < n$, then by \cref{lemma:automata-to-expressions-big-step}, we have $U \in \sem{e^{Q'}_{pr}} \subseteq \sem{e_p}$.
Otherwise, let $s_i$ be the first intermediate state such that $s_i \not\in Q'$.
Since $s_i \preceq_A p$, we have that $s_i \prec_A p$, so from $U_i \cdots U_n \in L_A(s_i)$ it follows by the main induction hypothesis that $U_i \cdots U_n \in \sem{e_{s_i}}$.
Moreover, since $s_j \in Q'$ for all $0 \leq j < i$, we have by \cref{lemma:automata-to-expressions-big-step} that $U_1 \cdots U_{i-1} \in \sem{e^k_{p s_i}}$.
We conclude that $U_1 \cdots U_n \in \sem{e^{Q'}_{ps_i} \cdot e_{s_i}} \subseteq \sem{e_p}$.
\end{proof}

By induction on strict support, we have now proved the following theorem.
\begin{thm}[Automata to expressions]%
  \label{theorem:automata-to-expressions}
  Let $A = \angl{Q, F, \delta, \gamma}$ be a finite and fork-acyclic PA\@.
  For every $q \in Q$, we can construct an $e_q \in \terms$ such that
  $\sem{e_q} = L_A(q)$.
\end{thm}

\cref{theorem:expressions-to-automata,theorem:automata-to-expressions} now combine to form our Kleene theorem.
\begin{thm}[Kleene theorem]
  Let $L$ be a pomset language.
  The following are equivalent:
  \begin{enumerate}
   \item $L$ is series-rational;
   \item $L$ is  accepted by a state of a finite and fork-acyclic PA\@.
   \end{enumerate}
\end{thm}

\section{Further work}%
\label{section:further-work}

The decision procedure for language equivalence in pomset automata as outlined in \cref{section:language-equivalence} is not very efficient.
For one, the algorithm computes the atoms generated by the languages of \emph{all} states, rather than just the ones being compared.
Consequently, to use equivalence checking of pomset automata as a proxy for equivalence checking of sr-expressions in an applied setting, a more efficient procedure is necessary.
Bisimulation up-to equivalence~\cite{hopcroft-karp-1971} and extensions thereof would be an obvious class of  optimizations to explore.

Another direction for future work is to transfer the techniques from~\cite{brunet-pous-struth-2017} to fork-acyclic pomset automata.
The advantage of such an approach is that it would most likely not require the input PA to be well-structured.
Furthermore, since the techniques from~\cite{brunet-pous-struth-2017} can also be used to compare safe Petri nets ``up-to sequentialisation'', an adaptation of this algorithm could help us decide an analogous property for pomset automata.

In~\cite{kappe-brunet-luttik-silva-zanasi-2019}, we presented a Kleene theorem that  includes the parallel variant of the Kleene star, sometimes known as \emph{parallel star} or \emph{replication}.
This version of the theorem covers a strictly larger class of pomset automata, including those that are not fork-acyclic.
However, the structural condition that describes this larger class is not overly complicated and we would like to extend the decision procedure of the present paper to work for this class as well.
Doing so would distinguish our algorithm from~\cite{brunet-pous-struth-2017}, which does not (and cannot) include this operator, as it would necessarily construct Petri nets that are unbounded.

\paragraph{\bf Acknowledgements}
We would like to thank the anonymous reviewers for their insightful comments, which helped to substantially improve the initial version of this paper.

\bibliographystyle{alpha}
\bibliography{bibliography}
\appendix
\pagebreak

\section{Lemmas about pomsets and pomset automata}

To prove \cref{lemma:parallel-factorisation-unique}, we first recall the following lemma from~\cite{kappe-brunet-silva-zanasi-2018}.

\begin{lem}%
  \label{lemma:pomset-levi-parallel}
  Let $U, V, W, X$ be sp-pomsets such that $U \parallel V = W \parallel X$.
  There exist sp-pomsets $Y_0, Y_1, Z_0, Z_1$ such that
  \begin{mathpar}
    U = Y_0 \parallel Y_1
    \and
    V = Z_0 \parallel Z_1
    \and
    W = Y_0 \parallel Z_0
    \and
    X = Y_1 \parallel Z_1
  \end{mathpar}
\end{lem}

\restateparallelfactorisationunique*
\begin{proof}
  Let $V_1, \dots, V_n$ and $W_1, \dots, W_m$ be sp pomsets that form parallel factorisations of $U$.
  We set out to prove that $\mset{V_1, \dots, V_n} = \mset{W_1, \dots, W_m}$.
  The proof proceeds by induction on $\min(n, m)$.
  In the base, there are two cases to consider.
  If $\min(n, m) = 0$, then the claim follows by \cref{lemma:sp-unique-kind}.
  If $\min(n, m) = 1$, the claim follows by definition of parallel primes.

  For the inductive step, we assume that the claim holds for $n', m'$ with $\min(n', m') < \min(n, m)$, and that $\min(n, m) > 1$.
  By \cref{lemma:pomset-levi-parallel}, we find sp-pomsets $Y_0, Y_1, Z_0, Z_1$ s.t.
  \begin{mathpar}
    V_1 = Y_0 \parallel Y_1
    \and
    V_2 \parallel \dots \parallel V_n = Z_0 \parallel Z_1
    \and
    W_1 = Y_0 \parallel Z_0
    \and
    W_2 \parallel \dots \parallel W_m = Y_1 \parallel Z_1
  \end{mathpar}
  Since $V_1$ is a parallel prime, we find that $Y_0$ or $Y_1$ is empty (but not both).
  In the former case, we find that $V_1 = Y_1$ and $W_1 = Z_0$.
  Now, let $X_1, \dots, X_k$ be a parallel factorisation of $Z_1$.
  By induction, we then find that $\mset{V_2, \dots, V_n} = \mset{Z_0, X_1, \dots, X_k}$, as well as $\mset{W_2, \dots, W_m} = \mset{Y_1, X_1, \dots, X_k}$.
  Consequently, we have that
  \begin{align*}
    \mset{V_1, \dots, V_n}
    &= \mset{V_1, Z_0, X_1 \dots, X_k} \\
    &= \mset{Y_1, W_1, X_1, \dots, X_k} \\
    &= \mset{W_1, \dots, W_m}
  \end{align*}
  In the latter case, we find that $Y_0$ is not empty, and thus $Z_0$ is empty, by the fact that $W_1$ is a parallel prime --- thus, $W_1 = Y_0 = V_1$.
  Furthermore, $V_2 \parallel \dots \parallel V_m = W_2 \parallel \dots \parallel W_m$.
  Hence, by induction it follows that $\mset{V_2, \dots, V_n} = \mset{W_2, \dots, W_m}$; the claim then follows.
\end{proof}

\restateatomicproperties*
\begin{proof}
We first claim that for every $U \in \SP(\Delta)$ there exists at most one $w \in \Sigma^*$ such that $U \in \zeta(w)$.
To see this, suppose that $w, w' \in \Sigma^*$ such that $U \in \zeta(w)$ and $U \in \zeta(w')$.
This means that we can write $w = \ltr{a}_1\cdots\ltr{a}_n$ and $w' = \ltr{a}_1'\cdots\ltr{a}_{n'}'$, as well as $U = U_1\cdots{}U_n$ and $U = U_1'\cdots{}U_{n'}'$ such that for $1 \leq i \leq n$ we have $U_i \in \zeta(\ltr{a}_i)$, and for $1 \leq i \leq n'$ we have $U_i' \in \zeta(\ltr{a}_i')$.
Indeed, by the first restriction on $\zeta$ and \cref{lemma:sequential-factorisation-unique}, we find that $n = n'$, and that for $1 \leq i \leq n$ we have that $U_i = U_i'$.
This also means that for $1 \leq i \leq n$ we have that $U_i \in \zeta(\ltr{a}_i) \cap \zeta(\ltr{a}_i')$; the second restriction on $\zeta$ then tells us that for $1 \leq i \leq n$ we have that $\ltr{a}_i = \ltr{a}_i'$, which entails that $w = w'$.
Hence, we conclude that $w = w'$.

We now treat the claims in the order given.
\begin{enumerate}[label={(\roman*)}]
    \item
    First suppose that $U \in \zeta(L \cap L')$; then there exists a $w \in L \cap L'$ such that $U \in \zeta(w)$.
    We then have that $U \in \zeta(L)$ and $U \in \zeta(L')$, meaning that $U \in \zeta(L) \cap \zeta(L')$.

    For the other inclusion, let $U \in \zeta(L) \cap \zeta(L')$.
    Then there exist $w \in L$ and $w' \in L'$ such that $U \in \zeta(w)$ and $U \in \zeta(w')$; by the above, we conclude that $w = w'$, and hence $w \in L \cap L'$, which means that $U \in \zeta(L \cap L')$.

  \item
    For the third equality, first suppose that $U \in \zeta(L \cap L')$.
    In that case there exists a $w \in L \setminus L'$ such that $U \in \zeta(w)$.
    In particular, this means that $w \in L$ and hence $U \in \zeta(L)$.
    Furthermore, suppose that $U \in \zeta(L')$; we then find a $w' \in L'$ such that $U \in \zeta(w')$.
    By the above, this would mean that $w = w'$, and hence $w \in L'$ --- a contradiction.
    We therefore know that $U \not\in \zeta(L')$, and hence $U \in \zeta(L) \setminus \zeta(L')$.

    For the other inclusion, let $U \in \zeta(L) \setminus \zeta(L')$.
    Then we obtain $w \in L$ such that $U \in \zeta(L)$, and furthermore there cannot be a $w' \in L'$ with $U \in \zeta(w')$, which in particular means that $w \not\in L'$.
    We conclude that $w \in L \setminus L'$, and hence $U \in \zeta(L \setminus L')$.

  \item
    Lastly, if $L = \emptyset$, then $\zeta(L) = \emptyset$ by definition.
    For the other direction, suppose that $L \neq \emptyset$, i.e., that $w \in L$.
    By the second condition, we have that $\zeta(\ltr{a}) \neq \emptyset$ for all $\ltr{a} \in \Sigma$, and hence $\zeta(w)$ cannot be empty, which means that $\zeta(L)$ cannot be empty either.
    \qedhere
  \end{enumerate}
\end{proof}

\restatesemvssacc*
\begin{proof}
  The proof from left to right proceeds by induction on the construction of $\sacc$.
  In the base, we have $e = 1$ or $e = f^*$ for some $f \in \terms$; in both cases, we find $1 \in \sem{e}$ by definition.

  For the inductive step, there are three cases to consider.
  First, if $e \in \sacc$ because $e = f + g$ and $f \in \sacc$ or $g \in \sacc$, then $1 \in \sem{f}$ or $1 \in \sem{g}$ by induction, and hence $1 \in \sem{e}$.
  Second, if $e \in \sacc$ because $e = f \cdot g$ and $f, g \in \sacc$, then by induction $1 \in \sem{f}$ and $1 \in \sem{g}$, meaning that $1 = 1 \cdot 1 \in \sem{e}$.
  Third, if $e \in \sacc$ because $e = f \parallel g$ and $f, g \in \sacc$, then by induction $1 \in \sem{f}$ and $1 \in \sem{g}$, meaning that $1 = 1 \parallel 1 \in \sem{e}$.

  The proof from right to left proceeds by induction on the structure of $e$.
  In the base, the only case to consider is $e = 1$, where we find $e \in \sacc$ immediately.

  For the inductive step, there are four cases to consider.
  First, $e = f + g$, then $1 \in \sem{e}$ implies that $1 \in \sem{f}$ or $1 \in \sem{g}$.
  In the former case, $f \in \sacc$ by induction, and hence $e \in \sacc$ by definition.
  Second, if $e = f \cdot g$, then $1 \in \sem{e}$ implies that $1 \in \sem{f}$ and $1 \in \sem{g}$.
  By induction, we then find that $f, g \in \sacc$, and thus $e \in \sacc$ by definition.
  The case where $e = f \parallel g$ can be argued similarly.
  Lastly, if $e = f^*$, then $e \in \sacc$ by definition, and there is nothing to prove.
\end{proof}

\restaterundeconstruct*
\begin{proof}
  We proceed by induction on $\rightarrow_A$.
  In the base, there are two cases to consider.
  \begin{itemize}
  \item
    If $q \arun{U}_A q'$ because $q = q'$ and $U = 1$, then we choose $\ell = 0$ to satisfy the claim.

  \item
    If $q \arun{U}_A q'$ because $U = \ltr{a}$ for some $\ltr{a} \in \Sigma$ and $q' \in \delta(q, \ltr{a})$, then choose $\ell = 1$ and $U_1 = \ltr{a}$.
  \end{itemize}

  \noindent
  For the inductive step, there are also two cases to consider.
  \begin{itemize}
  \item
    On the one hand, if $q \arun{U}_A q'$ because $U = V \cdot W$ and there exists a $q'' \in Q$ such that $q \arun{V}_A q''$ and $q'' \arun{W}_A q'$, then we can obtain the necessary states and pomsets by induction from those runs.

  \item
    On the other hand, if $q \arun{U}_A q'$ because $U = U_1 \parallel \dots \parallel U_n$ and there exist $r_1, \dots, r_n \in Q$ as well as $r_1', \dots, r_n' \in F$ such that for $1 \leq i \leq n$ we have $r_i \arun{U_i'} r_i'$, and furthermore $q' \in \gamma(q, \mset{r_1, \dots, r_n})$, then we can choose $\ell = 1$ and $U_1 = U$ to satisfy the claim.
    \qedhere
  \end{itemize}
\end{proof}

\section{Lemmas towards the decision procedure}

\restateassociativity*
\begin{proof}
Let $p, q \in Q$ and $\phi \in \M(Q)$ such that $\gamma(p, \phi) \neq \emptyset$ and $q \in \phi$.
Suppose, towards a contradiction, that $U \in L_A(q)$ is parallel.
We then know that there exists a $q' \in F$ with $q \arun{U}_A q'$, which is a parallel unit run by \cref{lemma:well-structured-no-confusion}.
In particular, $q' \in \gamma(q, \phi')$ for some $\phi' \in \M(Q)$.
However, the latter contradicts the fact that $q' \in F$, by the premise that $A$ is well-structured.
It thus follows that $U$ must be either primitive or sequential.
\end{proof}

\restatesubstitutionvsruns*
\begin{proof}
  It suffices to prove that, for $q, q' \in Q$ as well as $U \in \SP(\Sigma)$, it holds that $q \arun{U}_A q'$ is a unit run if and only if there exists an $\ltr{a} \in \Delta$ with $U \in \zeta(\ltr{a})$ and $q' \in \delta'(q, \ltr{a})$; straightforward inductive arguments on run length (for the forward inclusion, using \cref{lemma:run-deconstruct}) and word length (for the backwards inclusion) then complete the proof.

  First suppose that $q \arun{U}_A q'$ is a unit run.
  We have two cases to consider.
  \begin{itemize}
  \item
    If $q \arun{U}_A q'$ is a sequential unit run, then $U = \ltr{b}$ for some $\ltr{b} \in \Sigma$, and $q' \in \delta(q, \ltr{b})$.
    We can then choose $\ltr{a} = \ltr{b}$ to find that $q' \in \delta'(q, \ltr{a})$ by definition of $\delta'$.

    \item
    If $q \arun{U}_A q'$ is a parallel unit run, then $U = U_1 \parallel \cdots \parallel U_n$ and we obtain $q_1, \ldots, q_n \in Q$ such that $q' \in \gamma(q, \mset{q_1, \dots, q_n})$ as well as $U_i \in L_A(q_i)$.
    For $1 \leq i \leq n$, we then choose $\alpha_i = \setcompr{r \in Q'}{U_i \in L_A(r)} \in \At_{A[Q']}$, and set $\ltr{a} = \mset{\alpha_1, \dots, \alpha_n} \in \Delta$.
    We now find that $U \in \zeta(\ltr{a})$ as well as $q' \in \delta'(q, \mset{\alpha_1, \dots, \alpha_n})$.
  \end{itemize}
  Conversely, suppose $\ltr{a} \in \Delta$ s.t.\ $U \in \zeta(\ltr{a})$ and $q' \in \delta'(q, \ltr{a})$.
  We have two cases to consider.
  \begin{itemize}
  \item
    If $\ltr{a} \in \Sigma$, then $q' \in \delta(q, \ltr{a})$ and $U = \ltr{a}$.
    Hence $q \arun{U}_A q'$ is a sequential unit run.

  \item
    If $\ltr{a} = \mset{\alpha_1, \dots, \alpha_n}$, then $q' \in \gamma(q, \mset{q_1, \dots, q_n})$ with $q_i \in \alpha_i$.
    Since $U \in \zeta(\ltr{a})$, also $U = U_1 \parallel \cdots \parallel U_n$ with $U_i \in L_A(q_i)$.
    Thus, $q \arun{U}_A q'$ is a parallel unit run.
    \qedhere
  \end{itemize}
\end{proof}

\restatesubstitutionatomic*
\begin{proof}
  If $\ltr{a} \in \Delta$ and $U \in \zeta(\ltr{a})$, then $U = \ltr{b}$ for some $\ltr{b} \in \Sigma$, or $U \in L_A(\alpha_1) \parallel \cdots \parallel L_A(\alpha_n)$ for atoms $\alpha_1, \dots, \alpha_n \in \At_{A[Q']}$.
  In the former case, we know that $U$ is primitive, and hence a sequential prime.
  In the latter case, we know that none of these atoms contain the empty set and $n \geq 2$ (since $A$ is well-structured); it follows that $U$ is a sequential prime.

  Furthermore, suppose $\ltr{a}, \ltr{b} \in \Delta$ such that $U \in \zeta(\ltr{a})$ and $U \in \zeta(\ltr{b})$.
  By the above, we know that $U$ is either primitive and $\ltr{a}, \ltr{b} \in \Sigma$, or $U$ is parallel and $\ltr{a}, \ltr{b} \in \Delta \setminus \Sigma$.
  In the former case, it follows that $\ltr{a} = U = \ltr{b}$ by definition of $\zeta$.
  In the latter case, we know that $U = U_1 \parallel \cdots \parallel U_n$ and $U = U_1' \parallel \cdots \parallel U_{n'}'$ as well as $\ltr{a} = \mset{\alpha_1, \dots, \alpha_n}$ and $\ltr{b} = \mset{\alpha_1', \dots, \alpha_{n'}'}$, such that for $1 \leq i \leq n$ we have that $U_i \in L_A(\alpha_i)$ and for $1 \leq i \leq n'$ we have that $U_i' \in L_A(\alpha_i')$.
  Because $A$ is well-structured, each of the $U_i$ and $U_i'$ is a parallel prime; by \cref{lemma:parallel-factorisation-unique}, we then find that $n = n'$, and without loss of generality we have that $U_i = U_i'$.
  Since $U_i \in L_A(\alpha_i) \cap L_A(\alpha_i')$, it follows that $\alpha_i = \alpha_i'$, and hence $\ltr{a} = \ltr{b}$.
\end{proof}

\section{Lemmas about well-structured automata}%
\label{sec:proofs-well-struct}

The following characterisation of runs labelled by the empty pomset will be useful.

\begin{fact}%
\label{fact:characterise-empty-pomset-runs}
Let $\leadsto_A$ be the smallest relation on $Q$ satisfying the rules
\begin{mathpar}
\inferrule{~}{%
  p \leadsto_A p
}
\and
\inferrule{%
  p \leadsto_A r \\
  r \leadsto_A q
}{%
  p \leadsto_A q
}
\and
\inferrule{%
  q \in \gamma(p, \mset{q_1, \dots, q_n})\and
  \forall 1 \leq i \leq n.\ q_i \leadsto_A q_i' \in F
}{%
  p \leadsto_A q
}
\end{mathpar}
Now $p \leadsto_A q$ holds if and only if $p \arun{1}_A q$.
\end{fact}
\begin{proof}
The first inclusion is proved by induction on the construction of $\leadsto_A$.
In the base, $p \leadsto_A q$ because $p = q$, thus $p \arun{1}_A q$ immediately.
In the inductive step, there are two cases.
\begin{itemize}
    \item
    If $p \leadsto_A q$ because there exists an $r \in Q$ such that $p \leadsto_A r$ and $r \leadsto_A q$, then by induction we have that $p \arun{1}_A r$ and $r \arun{1}_A q$.
    It then follows that $p \arun{1}_A q$.

    \item
    Suppose $p \leadsto_A q$ because there exist $q_1, \dots, q_n \in Q$ such that $q \in \gamma(p, \mset{q_1, \dots, q_n})$ and $q_1', \dots, q_n' \in F$ such that for $1 \leq i \leq n$ we have that $q_i \arun{1}_A q_i'$.
    By induction, we find for $1 \leq i \leq n$ that $q_i \arun{1}_A q_i'$.
    We can then conclude that $p \arun{1}_A q$.
\end{itemize}

\noindent
For the other inclusion, we proceed by induction on the construction of $\arun{1}_A$.
In the base, we have $p \arun{1}_A q$ because $p = q$, and so we are done.
The case where $p \arun{1}_A q$ is a sequential unit run can be disregarded.
For the inductive step, there are two cases to consider.
\begin{itemize}
    \item
    If $p \arun{1}_A q$ because $1 = U \cdot V$ and there exists an $r \in Q$ such that $p \arun{U}_A r$ and $r \arun{V}_A q$, then note that $U = V = 1$.
    Hence $p \leadsto_A r$ and $r \leadsto_A q$ by induction, meaning $p \leadsto_A q$.

    \item
    Suppose $p \arun{1}_A q$ because there exist $q_1, \dots, q_n \in Q$ such that $q \in \gamma(p, \mset{q_1, \dots, q_n})$ and there exist $q_1' \dots, q_n' \in F$ such that $1 = U_1 \parallel \cdots \parallel U_n$ and for $1 \leq i \leq n$ we have $q_i \arun{U_i}_A q_i'$.
    By induction, we find for $1 \leq i \leq n$ that $q_i \leadsto_A q_i'$.
    We can then conclude that $p \leadsto_A q$.
    \qedhere
\end{itemize}
\end{proof}

\noindent
To prove \cref{lemma:well-structured-iff-triple}, the following is useful.

\begin{fact}%
\label{fact:parsimony-vs-empty-pomset}
Let $A$ be a $1$-forking PA such that if $q$ is a fork target, then $q \not\in F$.
For any states $p, q \in Q$, we have $p \arun{1}_A q$ iff $p = q$.
Furthermore, $1 \in L_A(p)$ iff $p \in F$.
\end{fact}
\begin{proof}
For the first part, the implication from right to left holds by definition of $\runrel_A$.
For the other implication, note that $p \arun{1}_A q$ is equivalent to $p \leadsto_A q$ by \cref{fact:characterise-empty-pomset-runs}.
We proceed by induction on the derivation $p \leadsto_A q$.
In the base, we have that $p \leadsto_A q$ because $p = q$, and we are done.
In the inductive step, there are two cases to consider.
\begin{itemize}
    \item
    If $p \leadsto_A q$ because $p \leadsto_A r$ and $r \leadsto_A q$ for some $r \in Q$, then by induction $p = r = q$.

    \item
    If $p \leadsto_A q$ because there exist $q_1, \dots, q_n \in Q$ such that $r \in \gamma(p, \mset{q_1, \dots, q_n})$, and $q_1', \dots, q_n' \in F$ such that for $1 \leq i \leq n$ we have that $q_i \leadsto_A q_i'$, then by induction we have for $1 \leq i \leq n$ that $q_i = q_i'$.
    Since $A$ is $1$-forking, it follows that $q_1 = q_1' \in F$.
    But this contradicts our premise, because $q_1$ is a fork target; we can therefore exclude this case.
\end{itemize}

\noindent
As for the second part, the implication from right to left holds by definition of $L_A(p)$.
For the other implication, note that if $1 \in L_A(p)$, then $p \arun{1}_A q$ for some $q \in F$.
Since $p = q$ by the first part, we conclude that $p \in F$.
\end{proof}

\restatewellstructuredifftriple*
\begin{proof}
By definition, an automaton is well-structured iff it is $2$-forking, flat-branching, and satisfies the property that for all $p \in Q$ and $q \in \phi \in \M(Q)$ with $\gamma(p, \phi) \neq \emptyset$, it holds that $q\notin F$.
If the automaton is parsimonious, then this property is satisfied.
By \cref{fact:parsimony-vs-empty-pomset}, if $A$ is well-structured, then $q\notin F\Leftrightarrow 1\notin L_A(q)$, so it is in particular parsimonious.
\end{proof}

\restateepsilondecidable*
\begin{proof}
By \cref{fact:characterise-empty-pomset-runs}, it suffices to show that $\leadsto_A$ is decidable.
Since $A$ finite, we can build $\leadsto_A$ by saturation, starting from the identity and adding pairs until we reach a fixpoint; this fixpoint is reached after finitely many steps as a result of boundedness.
\end{proof}

\restateweakimplementationiffimplementation*
\begin{proof}
Let $A = \angl{Q, F, \delta, \gamma}$, and $A' = \angl{Q', F' ,\delta' ,\gamma'}$.
We can assume without loss of generality that $Q$ and $Q'$ are disjoint.
We define $A'' = \angl{Q'', F'', \delta'', \gamma''}$ as follows:
\begin{mathpar}
Q'' = Q' \uplus Q
\and
F'' = F' \uplus \setcompr{q \in Q}{1 \in L_A(q)}
\\
\delta''(q'',\ltr a) =
    \begin{cases}
    \delta'(q'', \ltr a) & q'' \in Q' \\[2mm]
    \bigcup_{x\in Q_{q''}} \delta'(x,\ltr a) & q'' \in Q
    \end{cases}
\and
\gamma''(q'', \phi) =
    \begin{cases}
    \gamma'(q'', \phi) & q'' \in Q' \\[2mm]
    \bigcup_{x \in Q_{q''}} \gamma'(x, \phi) & q'' \in Q
    \end{cases}
\end{mathpar}
Here, $Q_{q''}$ is the set of states of $Q$ implementing $q'' \in Q$, by weak implementation.
Note that if $q'' \in Q$, then there are finitely many $\phi \in \M(Q'')$ such that $\gamma''(q'', \phi) \neq \emptyset$, because $Q_{q''}$ is finite, and for each $x \in Q_{q''}$ there are finitely many $\phi \in \M(Q')$ such that $\gamma'(x, \phi) \neq \emptyset$.

We now prove that $A''$ implements $A$.
First, we show for $q \in Q$ that $L_A(q) = L_{A''}(q)$.
\begin{itemize}
    \item
    For the inclusion from left to right, suppose that $U \in L_A(q)$.
    If $U = 1$, then $q \in F''$ by definition of $F''$, and thus $U = 1 \in L_{A''}(q'')$ immediately.

    On the other hand, suppose $U \neq 1$.
    There exists an $x \in Q_q$ such that $U \in L_{A'}(x)$, because $A'$ weakly implements $A$.
    Hence, there must exist a $q' \in F'$ such that $x \arun{U}_{A'} q'$; a straightforward inductive argument then shows that $x \arun{U}_{A''} q'$.
    Now, since $U \neq 1$, the latter run must be non-trivial.
    We thus find $x' \in Q''$ and $U = V \cdot W$ such that $x \arun{V}_{A''} x'$ is a unit run, and $x' \arun{W}_{A''} q'$.
    By construction of $A''$, it follows that $q \arun{V}_{A''} x'$ is a unit run too, and hence $q \arun{U}_{A''} q'$.
    Since $q' \in F''$, it follows that $U \in L_{A''}(q)$.

    \item
    For the inclusion from right to left, suppose that $U \in L_{A''}(q)$.
    There must then exist a $q' \in F''$ such that $q \arun{U}_{A''} q'$.
    Now, if this run is trivial, then $q = q'$ and $U = 1$.
    This means that $q \in F''$; since $q \in Q$, it follows that $q \not\in F'$, and hence $U = 1 \in L_A(q)$.

    Otherwise, if $q \arun{U}_{A''} q'$ is non-trivial, then there exists a $q'' \in Q''$ and $U = V \cdot W$ such that $q \arun{V}_{A''} q''$ is a unit run, and $q'' \arun{W}_{A''} q'$.
    By construction of $A''$, we find an $x \in Q_q$ such that $x \arun{V}_{A''} q''$ is a unit run, and hence $x \arun{U}_{A''} q'$.
    A simple inductive argument then shows that $x \arun{U}_{A'} q'$ and $q' \in F'$, and thus $U \in L_{A'}(x)$.
    Because $A'$ weakly implements $A$, we conclude that $U \in L_A(q)$.
\end{itemize}

\noindent
We should also show that if $A$ is fork-acyclic, then so is $A''$.
To this end, note that fork-acyclicity of $A$ already implies fork-acyclicity of $A'$; hence, we need only show that our construction above preserves fork-acyclicity.
The following property is useful.
\begin{fact}%
\label{fact:support-under-new}
For any $q \in Q'$ and $q' \in Q''$, if $q' \preceq_{A''} q$ then $q' \in Q'$ and $q' \preceq_{A'} q$.
\end{fact}
\begin{proof}
It suffices to argue that if $q' \preceq_{A''} q$ arises from one of the rules that generate $\preceq_{A''}$, then $q' \in Q'$ and $q' \preceq_{A'} q$.
This is straightforward, for if $q \in Q'$ then $\delta''(q, \ltr{a})$ coincides with $\delta'(q, \ltr{a})$ for all $\ltr{a} \in \Sigma$, and $\gamma''(q, \phi)$ coincides with $\gamma'(q, \phi)$ for all $\phi \in \M(Q'')$.
\end{proof}
Let $r,q \in Q''$ and $\phi \in \M(Q'')$ with $r \in \phi$ and $\gamma''(q,\phi) \neq \emptyset$.
By construction of $A''$, we have that $r \in Q'$.
If $q \preceq_{A''} r$, then by \cref{fact:support-under-new} we have that $q \in Q'$ and $q \preceq_{A'} r$.
In that case, also $\gamma'(q, \phi) \neq \emptyset$, and hence $r \preceq_{A'} q$.
This, however, contradicts the premise that $A'$ is fork-acyclic; we thus have that $r \prec_{A''} q$.
Hence, $A''$ must be fork-acyclic as well.

We should also show that our construction preserves boundedness, $n$-forking, flat-branching and parsimony.
For boundedness, the following is useful.
\begin{fact}%
\label{fact:support-under-old}
For any $q \in Q$ and $q' \in Q''$, if $q' \preceq_{A''} q$, then $q = q'$, or $q' \in Q'$ such that $q' \preceq_{A'} x$ for some $x \in Q_q$.
\end{fact}
\begin{proof}
    To see this, we proceed by induction on the construction of $\preceq_{A''}$.
    In the base, either $q = q'$ by reflexivity, or one of three cases applies.
    \begin{itemize}
        \item
        If $q' \preceq_{A''} q$ because $q' \in \delta''(q, \ltr{a}) \subseteq Q'$ for some $\ltr{a} \in \Sigma$, then $q' \in \delta'(x, \ltr{a})$ for some $x \in Q_q$, and therefore $q' \preceq_{A'} x$.

        \item
        If $q' \preceq_{A''} q$ because $q \in \gamma''(q, \phi)$ for some $\phi \in \M(Q'')$, then $q' \in \gamma'(x, \phi) \subseteq Q'$ for some $x \in Q_q$, and therefore $q' \preceq_{A'} x$.

        \item
        If $q' \preceq_{A''} q$ because there exists a $\phi \in \M(Q'')$ such that $q' \in \phi$ and $\gamma''(q, \phi) \neq \emptyset$, then $\phi \in \M(Q')$ and there exists an $x \in Q_q$ such that $\gamma'(x, \phi) \neq \emptyset$, and hence $q' \preceq_{A'} x$.
    \end{itemize}
    For the inductive step, $q' \preceq_{A''} q$ because $q'' \in Q''$ and $q' \preceq_{A''} q''$ and $q'' \preceq_{A''} q$.
    By induction, either $q = q''$ (in which case the claim follows by applying induction to $q' \preceq_{A''} q'' = q$), or $q'' \preceq_{A'} x$ for some $x \in Q_q$, in which case $q' \preceq_{A'} q''$ by \cref{fact:support-under-new}.
\end{proof}
By \cref{fact:support-under-old} we have that the support of $q \in Q$ in $A''$ is given by $\set{q} \cup \bigcup_{x \in Q_q} \pi_{A'}(x)$; if $A'$ is bounded, this set must be finite.
Furthermore, by~\ref{fact:support-under-new}, we know that the support of $q \in Q'$ in $A''$ is given by the support of $q$ in $A'$, which is finite if $A'$ is bounded.
Thus, if $A'$ is bounded, then so is $A''$.

We now prove that our construction preserves $n$-forking, parsimony and flat-branching.
\begin{itemize}
    \item
    If $A'$ is $n$-forking, let $\phi \in \M(Q'')$ and $q \in Q''$ such that $\gamma''(q,\phi)\neq\emptyset$.
    By construction there exists $q' \in Q'$ such that $\gamma'(q',\phi) \neq \emptyset$.
    Since $A'$ is $n$-forking we may conclude $|\phi| \geq n$.

    \item
    Assume, towards a contradiction, that $A''$ is not parsimonious.
    Then there exist $q \in Q''$ and $r \in \phi \in \M(Q'')$ such that $\gamma''(q, \phi) \neq \emptyset$ but $1 \in L_{A''}(r)$.
    By construction there exists $q' \in Q'$ such that $\gamma'(q',\phi) \neq \emptyset$.
    We also know that it must be the case that $r \in Q'$, hence $1 \in L_{A'} (r)$.
    This contradicts the premise that $A'$ is parsimonious.

    \item
    Assume, towards a contradiction, that $A$ is not flat-branching.
    There must then exist $q \in Q''$ and $\phi, \psi \in \M(Q'')$ with $r \in \phi$ such that $\gamma''(q, \phi) \neq \emptyset$ and $\gamma''(r, \psi) \cap F \neq \emptyset$.
    By construction of $A''$, we find that $\phi \in \M(Q')$, and there exists a $q' \in Q'$ such that $\gamma'(q', \phi) \neq \emptyset$.
    Also by construction of $A''$ and the fact that $r \in \phi$ (and hence $r \in Q'$), we find that $\psi \in \M(Q')$ and $\gamma'(r, \psi) \cap F' \neq \emptyset$.
    This contradicts the premise that $A'$ is flat-branching.
    \qedhere
\end{itemize}
\end{proof}

\subsection{Ensuring parsimony}

\restateparsimonificationcorrectness*
\begin{proof}
We relate runs in $A$ to runs in $A_0$ and vice versa, as follows:

\begin{fact}
For any run $p \arun{U}_A q$ we have $p \arun{U}_{A_0} q$, and if in addition we know that $1 \in L_A(q)$ and $U \neq 1$, then $p \arun U_{A_0} \top$.
\end{fact}
\begin{proof}
In the base, there are two cases.
\begin{itemize}
    \item
    If $p \arun{U}_A q$ because $p = q$ and $U = 1$, then $p \arun{U}_{A_0} q$; the second claim holds vacuously.

    \item
    If $p \arun{U}_A q$ because $U = \ltr{a}$ for some $\ltr{a} \in \Sigma$ and $q \in \delta(p, \ltr{a})$, then $q \in \delta_0(p, \ltr{a})$ by construction, so indeed $p \arun{U}_{A_0} q$.
    If $1 \in L_A(q)$, then $\top \in \delta_0(p, \ltr{a})$ as well, so $p \arun{\ltr{a}}_{A_0} \top$.
\end{itemize}

\noindent
For the inductive step, we have two more cases.
\begin{itemize}
    \item
    If $p \arun{U}_A q$ because $U = V \cdot W$ and there exists $r \in Q$ such that $p \arun{V}_A r$ and $r \arun{W}_A q$, then by induction we know that $p \arun{V}_{A_0} r$ and $r \arun{W}_{A_0} q$, so $p \arun{U}_{A_0} q$.

    If furthermore $1 \in L_A(q)$ and $V \cdot W \neq 1$, then we distinguish two subcases.
    \begin{itemize}
        \item
        If $W = 1$, then $V = U$, and $1 \in L_A(r)$; therefore by induction we get $p \arun{U}_{A_0} \top$.
        \item
        If $W \neq 1$, then by induction we get $r \arun{W}_{A_0} \top$, hence $p \arun {V \cdot W}_{A_0} q$.
    \end{itemize}

    \item
    Suppose that $p \arun{U}_A q$ because there exist $q_1, \dots, q_n \in Q$ with $q \in \gamma(p, \mset{q_1, \dots, q_n})$, and there exist $q_1', \dots, q_n' \in F$ such that $U = U_1 \parallel \cdots \parallel U_n$ and for $1 \leq i \leq n$ we have $q_i \arun{U_i}_A q_i'$.
    We then partition $\mset{q_1, \dots, q_n}$ into $\phi$ and $\psi$ such that:
    \begin{mathpar}
    \mset{q_1,\dots,q_n} = \phi \sqcup \psi
    \and
    \forall i.\, q_i \in \phi \Rightarrow U_i \neq 1
    \and
    \forall i.\, q_i \in \psi \Rightarrow U_i = 1.
    \end{mathpar}
    By construction of $\gamma_0$, we have $q \in \gamma_0(p, \phi)$.
    By induction, since $\forall i. q_i \in \phi \Rightarrow U_i \neq 1$, we get $\forall i. q_i \in \phi \Rightarrow q_i \arun{U_i}_{A_0} \top$.
    Consequently, we have $p \arun{U_1 \parallel \dots \parallel U_n}_{A_0} q$.

    \smallskip
    If additionally $1 \in L_A(q)$ and $U_1\parallel\dots\parallel U_n\neq 1$, then there must exist a $1 \leq i \leq n$ with $U_i \neq 1$, so we know that $\phi \neq \mempty$; hence, $\top \in \gamma_0(p,\phi)$.
    As result, we get $p \arun{U_1 \parallel \dots \parallel U_n}_{A_0} \top$.
    \qedhere
\end{itemize}
\end{proof}

\begin{fact}
For any $p \in Q$, if $p\arun U_{A_0} q$, then if $q \in Q$ we have $p \arun U_{A} q$, and if $q = \top$ we have $U \in L_A(p)$ but $U \neq 1$.
\end{fact}
\begin{proof}
In the base, there are two cases.
\begin{itemize}
    \item
    If $p \arun{U}_{A_0} q$ because $U = 1$ and $p = q$, then $p \arun{U}_A q$ immediately.

    \item
    If $p \arun{U}_{A_0} q$ because $U = \ltr{a}$ for some $\ltr{a} \in \Sigma$ and $q \in \delta_0(p, \ltr{a})$, then there are two subcases:
    \begin{itemize}
        \item
        If $q \in Q$, then $q \in \delta(p, \ltr{a})$, and hence $p \arun{U}_A q$.

        \item
        If $q = \top$, then there exists a $q' \in \delta(p, \ltr{a})$ such that $1 \in L_A(q')$.
        In particular, this means that $p \arun{U}_A q'$ and $q' \arun{1}_A q''$ for some $q'' \in F$.
        This means that $p \arun{U}_A q''$, and hence $U \in L_A(p)$.
        Since $\ltr{a} \neq 1$, we rightfully get $U \neq 1$.
    \end{itemize}
\end{itemize}
For the inductive step, there are again two cases.
\begin{itemize}
    \item
    Suppose $p \arun{U}_{A_0} q$ because $U = V \cdot W$ and there exists an $r \in Q_0$ with $p \arun{V}_{A_0} r$ and $r \arun{W}_{A_0} q$.
    We distinguish two cases based on $r$:
    \begin{itemize}
        \item
        If $r = \top$, then by construction of $A_0$ we have $W = 1$ and $q = \top$, so by induction it follows that $U = V \in L_A(p)$ while $U = V \neq 1$;

        \item
        If $r \in Q$, then by induction hypothesis we have $p \arun{U}_A r$.
        We distinguish two cases:
        \begin{itemize}
            \item
            If $q\in Q$, then by induction $r\arun{V}_A q$, so $p\arun {U\cdot V}_A q$.
            \item
            If $q = \top$, then by induction $W \in L_A(r)$ while $W \neq 1$, so $U \in L_A(p)$ while $U \neq 1$.
        \end{itemize}
    \end{itemize}

    \item
    Suppose $p \arun{U}_{A_0} q$ because there exist $q_1, \dots, q_n \in Q_0$ such that $q \in \gamma(p, \mset{q_1, \dots, q_n})$ such that $U = U_1 \parallel \cdots \parallel U_n$ and for $1 \leq i \leq n$ we have $q_i \arun{U_i}_{A_0} \top$.
    By construction of $\gamma_0$, we also know that for $1 \leq i \leq n$ we have $q_i \in Q$.
    By induction, we then know that for $1 \leq i \leq n$ we have $U_i \in L_A(q_i)$ but $U_i \neq 1$.
    We distinguish the two cases for $q$:
    \begin{itemize}
        \item
        If $q\in Q$, then by construction of $\gamma_0$ there exists a multiset $\psi$ such that:
        \begin{mathpar}
        q\in \gamma(p, \mset{q_1, \dots, q_n}\sqcup\psi)\and
        \forall r\in \psi.\,1\in L_A(r).
        \end{mathpar}
        If $\psi=\mset{q_{n+1},\dots,q_m}$, we may then complete the run by setting $\forall n<i\leq m.\, U_i=1$ and performing the following $\gamma$ step:
        \[
            \inferrule{%
                \forall 1 \leq i \leq m.\ U_i\in L_A(q_i) \and
                q \in \gamma(p, \mset{q_1, \dots, q_m})
            }{%
                p \arun{U_1 \parallel \dots \parallel U_m}_{A} q
            }
        \]
        This concludes this case, since we have:
        \[U_1\parallel\cdots\parallel U_m=U_1\parallel\cdots\parallel U_n\parallel 1\parallel\cdots\parallel 1=U_1\parallel\cdots\parallel U_n.\]

        \item
        Similarly, if $q = \top$, then $\phi\neq \mempty$ and there exist $\psi \in \M(Q)$ and $q' \in Q$ such that:
        \begin{mathpar}
        q' \in \gamma(p, \mset{q_1, \dots, q_n} \sqcup \psi)
        \and
        \forall r \in \psi.\, 1 \in L_A(r)
        \and
        1 \in L_A(q').
        \end{mathpar}
        We may proceed as in the previous case to show that $p \arun{U_1 \parallel \dots \parallel U_n}_{A} q'$.
        Since $1\in L_A(q')$, this implies $U_1 \parallel \dots \parallel U_n\in L_A(p)$.
        To conclude, note that for all $1 \leq i \leq n$ we have $U_i \neq 1$; since $\phi \neq \mempty$, also $n > 0$.
        This entails that $U = U_1 \parallel \dots \parallel U_n \neq 1$.
        \qedhere
    \end{itemize}
\end{itemize}
\end{proof}

\noindent
Summing up, we get that $L_A(q) \setminus \set{1} = L_{A_0}(q)$.
So, if $1 \not\in L_A(q)$, then we have $L_A(q) = L_{A_0}(q)$, and otherwise $L_A(q) = L_{A_0}(q) \cup L_{A_0}(\top)$.
Therefore, $A_0$ weakly implements $A$, since fork-acyclicity is clearly preserved by this construction.

To show that $A_0$ is parsimonious, notice that by construction, if $\gamma_0(p, \phi) \neq \emptyset$ and $q \in \phi$, then $q \in Q$.
As we showed above, this implies $1 \notin L_{A_0}(q)$.

To show that $A_0$ is bounded, note that if $q, q' \in Q_0$ with $q' \prec_{A_0} q$, then either $q' = \top$, $q = \top$ or $q, q' \in Q$ with $q' \preceq_A q$.
Hence, if $q \in Q_0$, then either $q = \top$ and $\pi_{A_0}(q) = \set{\top}$, or $q \in Q$ and $\pi_{A_0}(q) \subseteq \pi_A(q)$; in either case, $\pi_{A_0}(q)$ is finite, since $A$ is bounded; hence, $A_0$ is bounded, too.
\end{proof}

\subsection{Removing nullary forks}

\restateepsiloneliminationcorrectness*
\begin{proof}
To see that $A_1$ is bounded, note that $q' \preceq_{A_1} q$ if and only if $q' \preceq_A q$; hence, $\pi_{A_1}(q) \subseteq \pi_A(q)$ for all $q \in Q$; since $A$ is bounded, this implies that $A_1$ is bounded, too.
This also show that $A_1$ is fork-acyclic if $A_1$ is.

The condition $\phi \neq \mempty$ in the definition of $\gamma_1$ ensures that $A_1$ is $1$-forking.
Before we show that $A_1$ is parsimonious, note that if we have $r \in \phi$ and $\gamma_1(q, \phi) \neq \emptyset$, then there is a state $q'$ such that $q \epstrans q'$ and $\gamma(q', \phi) \neq \emptyset$.
By parsimony of $A$, we have $1 \notin L_A(r)$, hence $r \not\in F$.
By \cref{fact:parsimony-vs-empty-pomset}, we conclude that $1 \not\in L_{A_1}(p)$, hence $A_1$ is parsimonious.

We now check that $A_1$ weakly implements $A$, by relating their runs.
\begin{fact}%
\label{fact:epsilon-elimination-correctness-forward}
If $q \epstrans p \arun{U}_A p' \epstrans q'$ with $U \neq 1$, then $q \arun{U}_{A_1} q'$.
\end{fact}
\begin{proof}[Proof of \cref{fact:epsilon-elimination-correctness-forward}]
We proceed by induction on the construction of $p' \arun{U}_A q'$.
In the base, we can exclude the case where $p' = q'$ and $U = 1$, because it contradicts the premise.
This leaves the case where $p \arun{U}_A p'$ because $U = \ltr{a}$ for some $\ltr{a} \in \Sigma$ and $p' \in \delta(p, \ltr{a})$.
By construction of $\delta_1$, we then find that $q' \in \delta_1(q, \ltr{a})$, hence $q \arun{U}_{A_1} q'$.
For the inductive step, there are two cases:
\begin{itemize}
    \item
    Suppose $p \arun{U}_A p'$ because $U = V \cdot W$ and there exists $p'' \in Q$ such that $p \arun{V}_A p''$ and $p'' \arun{W}_A p'$.
    We distinguish four subcases:
    \begin{itemize}
        \item
        If $V = 1 = W$, then $U = V \cdot W = 1$, contradicting the premise; we disregard this case.

        \item
        If $V \neq 1 = W$, then $p'' \epstrans q'$ and $U = V$; by induction, we find that $q \arun{U}_{A_1} q'$.
        \item
        If $V = 1 \neq W$, then $q \epstrans p''$ and $U = W$; by induction, we find that $q \arun{U}_{A_1} q'$.
        \item
        If $V \neq 1 \neq W$, then induction we get $q \arun{V}_{A_1} p''$ and $p'' \arun{W}_{A_1} q'$, hence $q \arun{U}_{A_1} q'$.
    \end{itemize}

    \item
    Suppose $p \arun{U}_A p$ because there exist $r_1, \dots, r_n \in Q$ such that $p' \in \gamma(p, \mset{r_1, \dots, r_n})$, and $U = U_1 \parallel \cdots \parallel U_n$ such that for $1 \leq i \leq n$ there exists $r_i' \in F$ with $r_i \arun{U_1}_A r_i'$.
    By parsimony of $A$, we know that each of the $U_i$ is different from $1$.
    Hence, $U_1\parallel\dots\parallel U_n$ is itself non-empty, and we have $n > 0$.
    Therefore $q' \in \gamma_1(q, \mset{r_1, \dots, r_n})$.
    By induction we have for every $1 \leq i \leq n$ that $r_i \arun {U_i}_{A_1} r'_i \in F$, so we may conclude $q \arun{U}_{A_1} q'$.
    \qedhere
\end{itemize}
\end{proof}

\begin{fact}%
\label{fact:epsilon-elimination-correctness-backward}
If $q \arun{U}_{A_1} q'$, then $q \arun{U}_{A} q'$.
\end{fact}
\begin{proof}[Proof of \cref{fact:epsilon-elimination-correctness-backward}]
We proceed by induction on $q \arun{U}_{A_1} q'$.
In the base, there are two cases.
\begin{itemize}
    \item
    If $q \arun{U}_{A_1} q'$ because $q = q'$ and $U = 1$, then $q \arun{U}_A q'$ immediately.

    \item
    If $q \arun{U}_{A_1} q'$ because $U = \ltr{a}$ for some $\ltr{a} \in \Sigma$ and $q' \in \delta_1(q, \ltr{a})$, then there exist $p, p' \in Q$ such that $q \epstrans p$ and $p' \epstrans q'$ with $p' \in \delta(p, \ltr{a})$, by construction of $\delta_1$.
    We can string these together to find that $q \arun{1}_A p \arun{\ltr{a}}_A p' \arun{1}_A q'$, hence $q \arun{U}_A q'$.
\end{itemize}
In the inductive step, there are again two cases:
\begin{itemize}
    \item
    Suppose $q \arun{U}_{A_1} q'$ because $U = V \cdot W$ and there exists a $q'' \in Q$ such that $q \arun{V}_{A_1} q''$ and $q'' \arun{W}_{A_1} q'$.
    By induction, $q \arun{V}_A q''$ and $q'' \arun{W}_A q'$, and hence $q \arun{U}_A q'$.

    \item
    Suppose $q \arun{U}_{A_1} q'$ because there exist $r_1, \dots, r_n \in Q$ with $q' \in \gamma_1(q, \mset{r_1, \dots, r_n})$, and $U = U_1 \parallel \cdots \parallel U_n$ and for $1 \leq i \leq n$ there exists $r_i' \in F$ with $r_i \arun{U_i}_{A_1} r_i'$.
    By induction, it must be the case that for $1 \leq i \leq n$ we have $r_i \arun{U_i}_{A} r_i'$.
    By construction of $\gamma_1$, we know that there are $p, p' \in Q$ such that $q \epstrans p$, $p' \epstrans q'$, and $p' \in \gamma(p, \mset{q_1, \dots, q_n})$.
    We then know that $p \arun{U_1 \parallel \cdots \parallel U_n}_A p'$, and hence $q \arun{U_1 \parallel \cdots \parallel U_n}_A q'$.
    \qedhere
\end{itemize}
\end{proof}

\noindent
We can now wrap up by showing that, for $q \in Q$, we have $L_A(q) = \bigcup_{q \epstrans p} L_{A_1}(p)$.
\begin{itemize}
    \item
    Let $U \in L_A(q)$, meaning there is $q' \in F$ such that $q \arun{U}_A q'$.
    On the one hand, if $U = 1$, then $1 \in L_{A_1}(q')$; since $q \epstrans q'$, $U$ is contained in the right-hand side.
    On the other hand, if $U \neq 1$, then $q \arun U_{A_1} q'$ by \cref{fact:epsilon-elimination-correctness-forward}, hence $U \in L_{A_1}(p)$.
    Since $q \epstrans q$, we are done.

    \item
    Let $p \in Q$ such that $q \epstrans p$, and let $U \in L_{A_1}(p)$.
    This means there is $p' \in F$ such that $p \arun{U}_{A_1} p'$, and hence $p \arun{U}_A p'$ by \cref{fact:epsilon-elimination-correctness-backward}.
    Therefore $q \arun{U}_A p' \in F$, so $U \in L_A(q)$.
\end{itemize}
Since boundedness and fork-acyclicity are preserved, $A_1$ weakly implements $A$.
\end{proof}

\subsection{Removing unary forks}

\restateunaryforkremovalcorrectness*
\begin{proof}
We start by showing that our construction preserves fork-acyclicity, and that $A_2$ is bounded.
For preservation of fork-acyclicity, we first verify the following.

\begin{fact}%
\label{fact:preceq-stacked-vs-preceq}
The following hold for all $w = q_1\cdots{}q_n \in Q_2$ and $w' = q_1'\cdots{}q_m' \in Q_2$:
\begin{enumerate}[label={(\roman*)}]
    \item\label{property:uparrow-vs-preceq}
    If $q \in Q$ such that $q \mathrel{\uparrow} w$, then for all $q_i$ we have $q_i \preceq_A q$.

    \item\label{property:delta-stacked-vs-preceq}
    If $w' \in \delta_2(w, \ltr{a})$, then for every $q_i'$ there exists a $q_j$ with $q_i' \preceq_A q_j$.

    \item\label{property:gamma-stacked-vs-preceq}
    If $w' \in \gamma_2(w, \phi)$, then for every $q_i'$ there exists a $q_j$ with $q_i' \preceq_A q_j$.

    \item\label{property:gamma-stacked-vs-preceq-strict}
    If $\gamma_2(w, \phi) \neq \emptyset$ with $r \in \phi$, then there exists a $q_i$ with $r \prec_A q_i$.

    \item\label{property:preceq-stacked-vs-preceq}
    If $w' \preceq_{A_2} w$, then for every $q_i'$ there exists a $q_j$ with $q_i' \preceq_A q_j$.
\end{enumerate}
\end{fact}
\begin{proof}[Proof of \cref{fact:preceq-stacked-vs-preceq}]
We treat the claims in the order given.
\begin{enumerate}[label={(\roman*)}]
    \item
    We proceed by induction on the construction of $\uparrow$.
    In the base, where $q = q_1$, the claim holds vacuously.
    For the inductive step, we have $r \in Q$ such that $r \mathrel{\uparrow} q_1 \cdots q_{n-1}$ and $q_n \in \gamma(q, \mset{r})$.
    In that case, $q_n \preceq_A q$ and $r \preceq_A q$ immediately.
    Also, we find by induction that for all $1 \leq i < n$ it holds that $q_i \preceq_A r$, and hence $q_i \preceq_A q$.

    \item
    We proceed by induction on the construction of $\delta_2$.
    In the base, we have $w' \in \delta_2(w, \ltr{a})$ because $w' = q' \cdot x \cdot y$ and $w = q \cdot y$, and there exists an $r \in Q$ such that $q \mathrel{\uparrow} r \cdot x$ and $q' \in \delta(r, \ltr{a})$.
    By~\ref{property:uparrow-vs-preceq}, we know that $q_1' = q' \preceq_A r \preceq_A q = q_1$.
    Furthermore, if $q_i'$ appears in $x$, then also by~\ref{property:uparrow-vs-preceq} we know that $q_i' \preceq_A q = q_1$.
    Lastly, if $q_i'$ appears in $y$, then note that it also appears in $w$, and hence we can conclude by $q_i' \preceq_A q_j$ for some $j$.
    In the inductive step, $w = q \cdot x$ such that $q \in F$ and $w' \in \delta_2(x, \ltr{a})$.
    The claim then follows immediately by induction.

    \item
    This proof proceeds analogously to the one above.

    \item
    We proceed by induction on the construction of $\gamma_2$.
    In the base, we have that $w = q_{1} \cdot x$ and $w' = q_1' \cdot y \cdot x$ such that there exists an $p \in Q$ with $q_1 \mathrel{\uparrow} p \cdot y$ and $q_1' \in \gamma(p, \phi)$.
    Since $A$ is fork-acyclic, it follows that $r \prec_A p$; because $p \preceq_A q_1$ by~\ref{property:uparrow-vs-preceq}, the claim follows.

    For the inductive step, we have that $\gamma_2(w, \phi) \neq \emptyset$ because $w = q \cdot x$ such that $q \in F$ and $\gamma_2(x, \phi) \neq \emptyset$.
    By induction, we then find a $q_i$ such that $r \prec_A q_i$.

    \item
    This can be shown by induction on $\preceq_{A_2}$, noting that the base cases are covered by~\ref{property:delta-stacked-vs-preceq}--\ref{property:gamma-stacked-vs-preceq-strict}.
    For the inductive step, it suffices to note that the claimed property is transitive in nature.
    \qedhere
\end{enumerate}
\end{proof}

\noindent
Now, if $w \in \phi \in \M(Q_2)$ and $x \in Q_2$ with $\gamma_2(x, \phi) \neq \emptyset$, we should show that $w \prec_{A_2} w$.
First, note that $w = r \in Q$ for some $r \in Q$ by construction of $\gamma_2$.
By \cref{fact:preceq-stacked-vs-preceq}\ref{property:gamma-stacked-vs-preceq-strict}, we know that $w = q_1 \cdots q_n$, and there exists a $1 \leq i \leq n$ with $r \prec_A q_i$.
Suppose, towards a contradiction, that $w \preceq_{A_2} r$; then by \cref{fact:preceq-stacked-vs-preceq}\ref{property:preceq-stacked-vs-preceq} we also know that $q_i \preceq_A r$, contradicting that $r \prec_A q_i$.

\smallskip
To argue that $A_2$ is bounded, we first record the following.

\begin{fact}%
\label{fact:preceq-stacked-vs-depth}
Let $w = q_1 \cdots q_n \in Q_2$ and $x = q_1' \cdots q_{n'} \in Q_2$.
If $w \preceq_{A_2} x$, then for every $1 \leq i \leq n$ there exists a $1 \leq j \leq n'$ such that $n + D_A(q_i) - i \leq n' + D_A(q_j') - j$.
\end{fact}
\begin{proof}[Proof of \cref{fact:preceq-stacked-vs-depth}]
It suffices to verify the claim for the pairs generating $\preceq_{A_2}$.
\begin{itemize}
    \item
    If $w \in \delta_2(x, \ltr{a})$ for some $\ltr{a} \in \Sigma$, then we proceed by induction on the construction of $\delta_2$.
    In the base, there exist $r \in Q$ and $1 \leq k \leq n$ such that $q_1' \mathrel{\uparrow} r \cdot q_2 \cdots q_k$ and $q_1 \in \delta(r, \ltr{a})$, while $n' + k - 1 = n$, and for $k < i \leq n'$ we have $q_i = q_{i-k+1}'$.
    We now consider two cases.
    \begin{itemize}
        \item
        When $1 \leq i \leq k$, we choose $j = 1$; since $n \leq n'$ and $i \geq 1$ and $D_A(q_i) \leq D_A(q_1')$ by \cref{fact:preceq-stacked-vs-preceq}\ref{property:uparrow-vs-preceq}, we find that $n + D_A(q_i) - i \leq n' + D_A(q_1') - 1$.

        \item
        Otherwise, when $k < i \leq n$, we choose $j = i - k + 1$ to find that $n + D_A(q_i) - i = n + D_A(q_j') - i = n' + k - 1 - D_A(q_j') - i = n' + D_A(q_j') - j$
    \end{itemize}

    \noindent
    In the inductive step, $q_1' \in F$ and $w \in \delta_2(q_2' \cdots q_{n'}', \ltr{a})$.
    The claim then follows by induction.

    \item
    If $w \in \gamma_2(x, \phi)$ for some $\phi \in \M(Q_2)$, then the proof is similar to the previous case.

    \item
    If there exists a $\phi \in \M(Q_2)$ such that $w \in \phi$ and $\gamma_2(x, \phi) \neq \emptyset$, then note that $\phi \in \M(Q)$ by construction of $\gamma_2$, and thus that $w = r$ for some $r \in Q$.
    The proof proceeds by induction on $\gamma_2$, where it suffices to show that $D_A(r) < D_A(q_j')$ for some $1 \leq j \leq n'$.
    This is a direct consequence of \cref{fact:preceq-stacked-vs-preceq}\ref{property:gamma-stacked-vs-preceq-strict}.
    \qedhere
\end{itemize}
\end{proof}

\noindent
To see that $A_2$ is bounded, let $q_1 \cdots q_n \in Q_2$ and choose $m = \max_{1 \leq i \leq n} D_A(q_i)$.
If $q_1', \dots, q_{n'}' \in Q$ such that $q_1' \cdots q_{n'} \preceq_A q_1 \cdots q_n$, then by \cref{fact:preceq-stacked-vs-depth}, we find $1 \leq j \leq n'$ such that $n' \leq n' + D_A(q_1') - 1 \leq n + D_A(q_j) - j \leq n + m$.
By \cref{fact:preceq-stacked-vs-preceq}\ref{property:preceq-stacked-vs-preceq}, $q_1', \dots, q_n' \in \pi_A(q_1) \cup \cdots \cup \pi_A(q_n)$; the latter set is finite.
Hence, the states supporting $q_1 \cdots q_n$ in $A_2$ are words of length at most $n + m$ over a finite alphabet; thus $\pi_{A_2}(q_1 \cdots q_n)$ must be finite.

\smallskip
To show that $A_2$ can accept the same languages as $A$, the following facts are useful.

\begin{fact}%
\label{fact:finish-run}
If $q \in Q$ and $q \mathrel{\uparrow} q_1 \cdots q_n$ such that for $1 \leq i \leq n$ there exist $q_i' \in Q$ and $U_i \in \SP$ with $q_i \arun{U_i}_A q_i'$, and for $1 \leq i < n$ it holds that $q_i' \in F$, then $q \arun{U_1 \cdots U_n}_A q_n'$.
\end{fact}
\begin{proof}[Proof of \cref{fact:finish-run}]
The proof proceeds by induction on the construction of $\uparrow$.
In the base, we have that $q = q_1$ and $n = 1$; it then follows immediately that $q = q_1 \arun{U_1}_A q_1' = q_n'$.

For the inductive step, we have $q \mathrel{\uparrow} q_1 \cdots q_n$ because there exists $r \in Q$ with $r \mathrel{\uparrow} q_1 \cdots q_{n-1}$ and $q_n \in \gamma(q, \mset{r})$.
By induction, we then know that $r \arun{U_1 \cdots U_{n-1}}_A q_{n-1}$; since $q_{n-1} \in F$ and $q_n \in \gamma(q, \mset{r})$ it follows that $q \arun{U_1 \cdots U_{n-1}}_A q_n \arun{U_n} q_n'$, and hence $q \arun{U_1 \cdots U_n} q_n'$.
\end{proof}

\begin{fact}%
\label{fact:split-stacked-run}
If $w \in Q_2$ and $w' \in F_2$ and $U \in \SP$ such that $w \arun{U}_{A_2} w'$, then $w = q_1\cdots{}q_n$ and $U = U_1 \cdots U_n$ such that for $1 \leq i \leq n$ there exists a $q_i' \in F$ with $q_i \arun{U_i}_A q_i'$.
\end{fact}
\begin{proof}[Proof of \cref{fact:split-stacked-run}]
We proceed by induction on the length $\ell$ of $w \arun{U}_{A_2} w'$.
In the base, where $\ell = 0$, we have $U = 1$ and $w = w' \in F_2$.
We choose for $1 \leq i \leq n$ that $U_i = 1$ and $q_i' = q_i \in F$.

For the inductive step, we have that $U = V \cdot W$ and $w'' \in Q_2$ such that $w \arun{V}_{A_2} w''$ is a unit run, and $w'' \arun{W}_{A_2} w'$ of length $\ell - 1$.
By induction, $w'' = r_1\cdots{}r_{m}$ and $W = W_1\cdots{}W_m$ such that for $1 \leq i \leq m$ there exists an $r_i' \in F$ with $r_i \arun{W_i}_A r_i'$.
Suppose that $w \arun{V}_{A_2} w''$ is a sequential unit run.
Then $V = \ltr{a}$ for some $\ltr{a} \in \Sigma$, and $w'' \in \delta_2(w, \ltr{a})$.
We proceed by induction on $\delta_2$.

In the base, we have that $w = q_1 \cdots q_n$ and $w'' = q' \cdot v \cdot q_2 \cdots q_n$ such that there exists an $r \in Q$ with $q_1 \mathrel{\uparrow} r \cdot v$ and $q' \in \delta(r, \ltr{a})$.
Note that $q' \cdot v = r_1\cdots{}r_k$ for some $k \leq m$, and that $r_{i+k-1} = q_i$ for $2 \leq i \leq n$.
We choose $U_1 = \ltr{a} \cdot W_1 \cdots W_k$.
Since $r \arun{\ltr{a} \cdot W_1}_A r_1' \in F$ and for $2 \leq i \leq k$ we have $r_i \arun{W_i}_A r_i' \in F$, it follows that $q_1 \arun{U_1}_A r_k'$ by \cref{fact:finish-run}; we set $q_1' = r_k'$.
For $i \geq 2$, we choose $q_i' = r_{i+k-1}'$ and $U_i = W_{i+k-1}$, to find that $r_{i+k-1} \arun{W_{i+k-1}}_A r_{i+k-1}'$, and hence $q_i \arun{U_i} q_i'$.
Finally, we note that $U_1 \cdots U_n = \ltr{a} \cdot W_1 \cdots W_k \cdot W_{k+1} \cdots W_m = V \cdot W = U$.

In the inductive step, $w = q_1 \cdots q_n$ and $w'' \in \delta_2(q_2 \cdots{} q_n, \ltr{a})$, with $q_1 \in F$.
By induction, $U = U_2 \cdots{} U_n$ where, for $2 \leq i \leq n$, $q_i' \in F$ with $q_i \arun{U_i}_A q_i'$.
Here, $U_1 = 1$ and $q_1' = q_1$ suffices.

The case where $w \arun{V}_{A_2} w''$ is a parallel unit run can be treated similarly.
\end{proof}

\cref{fact:split-stacked-run} tells us that for $q \in Q$ we have that $L_{A_2}(q) \subseteq L_A(q)$.
After all, if $q \arun{U}_{A_2} w$ for some $w \in F_2$, then we find $q' \in F$ such that $q \arun{U}_A q'$, and hence $U \in L_A(q)$.

For the converse inclusion, the following facts tell us how we can compose runs in $A_2$.

\begin{fact}%
\label{fact:run-pad-right}
If $w \arun{U}_{A_2} w'$ and $x \in Q_2$, then $w \cdot x \arun{U}_{A_2} w' \cdot x$.
\end{fact}
\begin{proof}[Proof of \cref{fact:run-pad-right}]
We proceed by induction on the length $\ell$ of $w \arun{U}_{A_2} w'$.
In the base, where $\ell = 0$, we have that $w = w'$ and $U = 1$.
We then know that $w \cdot x = w' \cdot x$; hence $w \cdot x \arun{U}_{A_2} w' \cdot x$.

For the inductive step, let $\ell > 1$.
We then find $w'' \in Q_2$ and $U = V \cdot W$ such that $w \arun{V}_{A_2} w''$ is a unit run, and $w'' \arun{W}_{A_2} w'$ is of length $\ell - 1$.
Hence, $w'' \cdot x \arun{W}_{A_2} w' \cdot x$ by induction.
If $w \arun{V}_{A_2} w''$ is a sequential unit run, then $V = \ltr{a}$ for some $\ltr{a} \in \Sigma$, and $w'' \in \delta_2(w, \ltr{a})$.
By construction of $\delta_2$, we have $w'' \cdot x \in \delta_2(w \cdot x, \ltr{a})$, which means that $w \cdot x \arun{V}_{A_2} w'' \cdot x$.
In total, we have $w \cdot x \arun{U}_{A_2} w' \cdot x$.

The case where $w \arun{V}_{A_2} w''$ is a parallel unit run can be treated similarly.
\end{proof}

\begin{fact}%
\label{fact:run-follow-fork}
Let $q, r \in Q$ with $q' \in \gamma(q, \mset{r})$, and let $r \arun{U}_{A_2} w$ be nontrivial.
Then $q \arun{U}_{A_2} w \cdot q'$.
\end{fact}
\begin{proof}[Proof of \cref{fact:run-follow-fork}]
Since $r \arun{U}_{A_2} w$ is nontrivial, $U = V \cdot W$ and $x \in Q_2$ such that $r \arun{V}_{A_2} x$ is a unit run, and $x \arun{W}_{A_2} w$.
If $r \arun{V}_{A_2} x$ is a sequential unit run, then $V = \ltr{a}$ for some $\ltr{a} \in \Sigma$, and $x \in \delta_2(r, \ltr{a})$.
By construction of $\delta_2$, we obtain $q'', r \in Q$ and $y \in Q_2$ such that $x = q'' \cdot y$ and $r \mathrel{\uparrow} r' \cdot y$ as well as $q'' \in \delta(r', \ltr{a})$.
By definition of $\uparrow$, also $q \mathrel{\uparrow} r' \cdot y \cdot q'$, and thus $x \cdot q' = q'' \cdot y \cdot q' \in \delta_2(q, \ltr{a})$.
We then find that $q \arun{V}_{A_2} x \cdot q'$.
Since $x \cdot q' \arun{W}_{A_2} w \cdot q'$ by \cref{fact:run-pad-right}, we conclude that $q \arun{U}_{A_2} w \cdot q'$.

The case where $r \arun{V}_{A_2}$ is a parallel unit run can be argued similarly.
\end{proof}

\begin{fact}%
\label{fact:run-pad-left}
Let $w \arun{U}_{A_2} w'$ be nontrivial, and $x \in F_2$.
Then $x \cdot w \arun{U}_{A_2} w'$.
\end{fact}
\begin{proof}[Proof of \cref{fact:run-pad-left}]
Since $w \arun{U}_{A_2} w'$ is non-trivial, we find that $U = V \cdot W$ and $w'' \in Q_2$ such that $w \arun{V}_{A_2} w''$ is a unit run, and $w'' \arun{W}_{A_2} w'$.
If $w \arun{V}_{A_2} w''$ is a sequential unit run, then $V = \ltr{a}$ for some $\ltr{a} \in \Sigma$, and $w'' \in \delta_2(w, \ltr{a})$.
A simple inductive argument on the length of $x$ then tells us that $w'' \in \delta_2(x \cdot w, \ltr{a})$ as well.
From this, it follows that $x \cdot w \arun{V}_{A_2} w''$, and thus $x \cdot w \arun{U}_{A_2} w'$.

The case where $w \arun{V}_{A_2} w''$ is a parallel unit run can be argued similarly.
\end{proof}

\noindent
Finally, we can use the above to show that $A_2$ can simulate the unary forks of $A$.

\begin{fact}%
\label{fact:simulate-on-stack}
If $q \arun{U}_A q'$, then there exists $x \in F_2$ such that $q \arun{U}_{A_2} x \cdot q'$.
\end{fact}
\begin{proof}[Proof of \cref{fact:simulate-on-stack}]
We proceed by induction on $q \arun{U}_A q'$.
If $q \arun{U}_A q'$ is trivial, then the claim is satisfied by choosing $x = 1$.
Otherwise, suppose $q \arun{U}_A q'$ is a sequential unit run, i.e., $U = \ltr{a}$ for some $\ltr{a} \in \Sigma$, and $q' \in \delta(q, \ltr{a})$.
Since $q \mathrel{\uparrow} q$, we then have that $q' \in \delta_2(q, \ltr{a})$, and hence $q \arun{U}_{A_2} q'$.

For the inductive step, there are again two cases to consider.
\begin{itemize}
    \item
    Suppose $q \arun{U}_A q'$ because $U = V \cdot W$ and there exists $q'' \in Q$ such that $q \arun{V}_A q''$ and $q'' \arun{W}_A q'$.
    By induction, we obtain $x'', x' \in F_2$ such that $q \arun{V}_{A_2} x'' \cdot q''$ and $q'' \arun{W}_{A_2} x' \cdot q'$.
    Without loss of generality, $q'' \arun{W}_A q'$ is non-trivial, and hence neither is $q'' \arun{W}_{A_2} x' \cdot q'$.
    By \cref{fact:run-pad-left}, we find that $x'' \cdot q'' \arun{W}_{A_2} x' \cdot q'$.
    In total, we find that $q \arun{U}_{A_2} x \cdot q'$.

    \item
    Suppose $q \arun{U}_A q'$ because $U = U_1 \parallel \cdots \parallel U_n$, and there exist $r_1, \dots, r_n \in Q$ and $r_1', \dots, r_n' \in F$ such that for $1 \leq i \leq n$ we have $r_i \arun{U_i}_A r_i'$, and $q' \in \gamma(q, \mset{r_1, \dots, r_n})$.
    There are two subcases.
    \begin{itemize}
        \item
        If $n = 1$, then by induction we find $x_1 \in F_2$ such that $r_1 \arun{U_1}_{A_2} x_1' \cdot r_1'$.
        By \cref{fact:run-follow-fork}, we then find $q \arun{U}_{A_2} x_1' \cdot r_1' \cdot q'$.
        Choosing $x' = x_1' \cdot r_1'$ satisfies the claim.

        \item
        If $n \geq 2$, then $q' \in \gamma_2(q, \mset{r_1, \dots, r_n})$.
        By induction, we find for $1 \leq i \leq n$ an $x_i \in F_2$ with $r_i \arun{U_i}_{A_2} x_i' \cdot r_i' \in F_2$.
        Thus $q \arun{U}_{A_2} q'$; choosing $x' = 1$ satisfies the claim.
        \qedhere
    \end{itemize}
\end{itemize}
\end{proof}

\noindent
The above allows us to prove that, for $q \in Q$, we have $L_A(q) \subseteq L_{A_2}(q)$.
To this end, suppose $U \in L_A(q)$; then there exists a $q' \in F$ with $q \arun{U}_A q'$.
By \cref{fact:simulate-on-stack} we find $x' \in F_2$ with $q \arun{U}_{A_2} x \cdot q' \in F_2$, and hence $U \in L_{A_2}(q)$.
Since $L_{A_2}(q) \subseteq L_A(q)$, it follows that $L_A(q) = L_{A_2}(q)$.

\smallskip
Note that it is $2$-forking by construction.
For parsimony, observe that if $w \in \phi \in \M(Q_2)$ and $x \in Q_2$ such that $\gamma_2(x, \phi)$, then $\phi \in \M(Q)$ by definition of $\gamma_2$, and hence $w = q$ for some $q \in Q$.
A simple inductive argument then tells us that there exists an $r \in Q$ such that $\gamma(r, \phi) \neq \emptyset$.
Since $A$ is parsimonious, we know that $1 \not\in L_A(q)$; since $L_A(q) = L_{A_2}(q)$, it follows that $1 \not\in L_{A_2}(q)$.
\end{proof}

\subsection{Ensuring flat-branching}

\restateflatbranchingcorrectness*
\begin{proof}
The proof of this statement consists of several steps.
In the sequel, we will write $\pst{Q} = \setcompr{\pst{q}}{q \in Q}$ and $\sst{Q} = \setcompr{\sst{q}}{q \in Q}$.
We start by making the following observations:
\begin{fact}%
\label{fact:flat-branching-correctness-observations}
The following hold for all $\psi \in \M(Q)$:
\begin{enumerate}[label={(\roman*)}]
    \item\label{property:fork-stack-vs-size}
    If $\phi \blacktriangleleft \psi$, then $|\phi| \leq |\psi|$.

    \item\label{property:flat-branching-fork-target}
    If $p \in Q_3$ and $\phi \in \M(Q_3)$ such that $\gamma_3(p, \phi) \neq \emptyset$, then $\phi \in \M(\sst{Q})$.

    \item\label{property:flat-branching-fork-landing}
    If $p \in Q$ and $\phi \in \M(Q_3)$, then $\gamma_3(\sst{p}, \phi) \subseteq \sst{Q} \cup \pst{Q}$.
\end{enumerate}
\end{fact}
\begin{proof}
We treat the claims in the order given.
\begin{enumerate}[label={(\roman*)}]
    \item
    This claim is proved by induction on $\blacktriangleleft$.
    In the base, $\phi \blacktriangleleft \psi$ because $\phi = \psi$, and so the claim holds immediately.
    For the inductive step, we have that $\phi \blacktriangleleft \psi$ because $\psi = \psi_1 \sqcup \psi_2$, such that $\phi \blacktriangleleft \psi_1$, and $\gamma(p, \psi_2) \cap F \neq \emptyset$ with $\phi \blacktriangleleft \psi_1 \sqcup \mset{p}$.
    By induction, we have that $|\phi| \leq |\psi_1| + 1$.
    Since $A$ is $1$-forking, we have that $|\psi_2| \geq 1$; hence, we conclude that $|\phi| \leq |\psi_1| + |\psi_2| = |\psi|$.

    \item
    If $p \in Q_3$ and $\phi \in \M(Q_3)$ such that $\gamma_3(p, \phi) \neq \emptyset$, then by definition of $\gamma_3$ we have that $\phi = \sst{\psi}$ for some $\psi \in \M(Q)$.
    Hence, $\phi \in \M(\sst{Q})$.

    \item
    Suppose that $p \in Q$ and $\phi \in \M(Q_3)$, and let $q \in \gamma_3(\sst{p}, \phi)$.
    By definition of $\gamma_3$, we have that $q \in \set{\sst{r}, \pst{r}}$ such that $r \in \gamma(p, \psi)$ for some $\psi \in \M(Q)$.
    Hence, $q \in \sst{Q} \cup \pst{Q}$.
    \qedhere
\end{enumerate}
\end{proof}

\noindent
We are now set to prove that $A_3$ indeed satisfies the right properties.

\begin{fact}%
\label{fact:flat-branch-properties}
$A_3$ is $2$-forking, parsimonious, and flat-branching.
\end{fact}
\begin{proof}
For $2$-forking, suppose $p \in Q_3$ and $\phi \in \M(Q_3)$ such that $\gamma_3(p, \phi) \neq \emptyset$.
Then by definition of $\gamma_3$ we find $\chi, \psi \in \M(Q)$ and $r \in Q$, such that $\phi = \sst{\chi}$, $\psi \blacktriangleleft \chi$, and $\gamma(r, \psi) \neq \emptyset$.
Since $A$ is $2$-forking, we can conclude by \cref{fact:flat-branching-correctness-observations}\ref{property:fork-stack-vs-size} that $2 \leq |\psi| \leq |\chi| = |\phi|$.

For parsimony, suppose $\gamma_3(p,\phi) \neq \emptyset$ and $q \in \phi$.
Then by \cref{fact:flat-branching-correctness-observations}\ref{property:flat-branching-fork-target} we know that $q \in \sst{Q}$, so $q \not\in F_3$.
Since $A_3$ is $1$-forking, $1 \not\in L_{A_3}(q)$ by \cref{fact:parsimony-vs-empty-pomset}, hence $A_3$ is parsimonious.

For flat-branching, suppose $p \in Q_3$ is a fork target.
Then by \cref{fact:flat-branching-correctness-observations}\ref{property:flat-branching-fork-target}, we know that $p \in \sst{Q}$.
By \cref{fact:flat-branching-correctness-observations}\ref{property:flat-branching-fork-landing}, we can then conclude that
\(
    \gamma_3(p, \psi) \cap F_3
        \subseteq (\pst{Q} \cup \sst{Q}) \cap F_3
        = \emptyset.
\)
\end{proof}

\noindent
We can now relate the runs of $A_3$ to those in $A$ as follows.

\begin{fact}%
\label{fact:flat-branching-correctness-backward}
If $p \arun{U}_{A_3} q$, then the following hold:
\begin{enumerate}[label={(\roman*)}]
    \item
    If $p = \sst{p'}$ and $q \in \set{\sst{q'}, \pst{q'}}$, then $p' \arun{U}_A q'$.

    \item
    If $p \in \set{\sst{p'}, \pst{p'}}$ and $q = \top$, then there exists a $q' \in F$ with $p' \arun{U}_A q'$.
\end{enumerate}
\end{fact}
\begin{proof}
We proceed by induction on $p \arun{U}_{A_3} q$.
In the base, the case where $p \arun{U}_{A_3} q$ is trivial holds vacuously.
Otherwise, if $p \arun{U}_{A_3} q$ because $U = \ltr{a}$ for some $\ltr{a} \in \Sigma$ and $q \in \delta_3(p, \ltr{a})$, then we know that $p \not\in \pst{Q}$ by definition of $\delta_3$.
Therefore, assume $p = \sst{p'}$ for some $p' \in Q$.
There are two cases.
\begin{itemize}
    \item
    If $q\in\set{\sst {q'},\pst {q'}}$, then it must be the case that $q' \in \delta(p', \ltr{a})$, so $p' \arun{\ltr{a}}_A q'$.

    \item
    If $q = \top$, then $\delta(p', \ltr{a}) \cap F \neq \emptyset$.
    Choose $q' \in \delta(p',\ltr{a}) \cap F$ to find that so $p' \arun{\ltr{a}}_A q'$.
\end{itemize}
For the inductive step, there are again two cases.
\begin{itemize}
    \item
    Suppose $p \arun{U}_{A_3} q$ because $U = V \cdot W$ and there exists an $r \in Q_3$ such that $p \arun{V}_{A_3} r$ and $r \arun{W}_{A_3} q$.
    Furthermore, we may assume w.l.o.g.\ that neither of these runs is trivial.
    Since $\top$ does not permit nontrivial runs, we have $r \in \pst{Q} \cup \sst{Q}$.
    Furthermore, $p \not\in \pst{Q}$, because if $p \in \pst{Q}$ then $r = \top$; we set $p = \sst{p'}$.
    We do a case analysis on $r$.
    \begin{itemize}
        \item
        If $r = \pst{r'}$, then necessarily $q = \top$.
        By the induction hypothesis we get that $p' \arun{V}_A r'$ and we find $q' \in F$ such that $r' \arun{W}_A q'$, hence $p' \arun{U}_A q'$.

        \item
        If $r = \sst{r'}$, then by induction we get that $p' \arun{V \cdot W}_A r'$.
        We now look at $q$:
        \begin{itemize}
            \item
            If $q \in \set{\sst {q'},\pst{q'}}$, then by induction we have $r' \arun{W}_A q'$, so $p' \arun{V \cdot W}_A q'$.
            \item
            If $q = \top$, then induction gives us $q' \in F$ such that $r' \arun{W}_A q'$, and thus $p' \arun{V \cdot W}_A q'$.
        \end{itemize}
    \end{itemize}

    \item
    Suppose $p \arun{U}_{A_3} q$ because $q_1, \dots, q_n \in Q_3$ with $q \in \gamma_3(p, \mset{q_1, \dots, q_n})$, and $U = U_1 \parallel \cdots \parallel U_n$ such that for $1 \leq i \leq n$ we have $q_i \arun{U_i}_{A_3} \top$.
    Since $A_3$ is $1$-forking and parsimonious, each $U_i$ is non-empty by \cref{fact:parsimony-vs-empty-pomset}.
    By definition of~$\gamma_3$, for each $1 \leq i \leq n$ there exists a $q'_i \in Q$ such that $q_i = \sst{q'_i}$.
    By induction, we obtain for each $1 \leq i \leq n$ a $q_i'' \in F$ such that $q_i' \arun{U_i}_A q_i''$.

    On the one hand, suppose $p = \sst{p'}$ for $p' \in Q$.
    In that case, $q \in \set{\sst{q'}, \pst{q'}}$ for some $q' \in \gamma(p', \phi)$, with $\phi \blacktriangleleft \mset{q_1, \dots, q_n}$.
    We show $p' \arun{U}_A q'$ by induction on $\blacktriangleleft$.
    In the base, $\phi = \mset{q_1, \dots, q_n}$, and so the claim follows.
    In the inductive step, $\mset{q_1, \dots, q_n} = \mset{q_1, \dots, q_k} \sqcup \mset{q_{k+1}, \dots, q_n}$ and there exists an $r \in Q$ such that $\gamma(r, \mset{q_{k+1}, \dots, q_n}) \cap F \neq \emptyset$ and $\phi \blacktriangleleft \mset{q_1, \dots, q_k} \sqcup \mset{r}$.
    In that case, $r \arun{U_{k+1} \parallel \cdots \parallel U_n}_A r'$ for some $r' \in F$; hence, by induction, $p' \arun{U_1 \parallel \cdots \parallel U_n}_A q'$.

    On the other hand, if $p = \pst{p'}$ for some $p' \in Q$, then $q' = \top$ by definition of $\gamma_3$.
    Furthermore, there exists $q' \in \gamma(p', \phi) \cap F$ for some $\phi \in \M(Q)$ with $\phi \blacktriangleleft \mset{q_1, \dots, q_n}$.
    A similar inductive argument to the previous case then shows that $p' \arun{U}_A q'$.
    \qedhere
\end{itemize}
\end{proof}

\begin{fact}%
\label{fact:flat-branching-correctness-forward}
If $p \arun{U}_A q$ is nontrivial, then $\sst{p} \arun{U}_{A_3} \sst{q}$ and $\sst{p} \arun{U}_{A_3} \pst{q}$.

Furthermore, if $q \in F$, then either $\sst{p} \arun{U}_{A_3} \top$ or $\pst{p} \arun{U}_{A_3} \top$.
\end{fact}
\begin{proof}
We proceed by induction on $p \arun{U}_A q$.
In the base, $p \arun{U}_A q$ because $U = \ltr{a}$ for some $\ltr{a} \in \Sigma$, and $q \in \delta(p, \ltr{a})$.
Thus $\sst{q}, \pst{q} \in \delta_3(\sst{p}, \ltr{a})$, and hence $\sst{p} \arun{U}_{A_3} \sst{q}$ and $\sst{p} \arun{U}_{A_3} \pst{q}$.
Furthermore, $q \in F$, then $\top \in \delta_3(\sst{p}, \ltr{a})$, and hence $\sst{p} \arun{U}_{A_3} \top$.
In the inductive step, there are two cases.
\begin{itemize}
    \item
    If $p \arun{U}_A q$ because $U = V \cdot W$ and there exists an $r \in Q$ with $p \arun{V}_A r$ and $r \arun{W}_A q$, then we can assume without loss of generality that neither of these runs is trivial.
    By induction, we then find that $\sst{p} \arun{V}_{A_3} \sst{r}$ and $\sst{p} \arun{W}_{A_3} \sst{r}$, as well as $\sst{r} \arun{W}_{A_3} \sst{q}$ and $\sst{r} \arun{W}_{A_3} \pst{q}$.
    Putting this together, we have that $\sst{p} \arun{U}_{A_3} \sst{q}$ and $\sst{p} \arun{U}_{A_3} \pst{q}$.

    Furthermore, if $q = \top$, then it suffices to prove that $\sst{r} \arun{W}_{A_3} \top$ or $\pst{r} \arun{W}_{A_3} \top$, which we obtain from $r \arun{W}_A q$ by induction.

    \item
    Suppose $p \arun{U}_A q$ because there exist $q_1, \dots, q_n \in Q$ such that $q \in \gamma(p, \mset{q_1, \dots, q_n})$, and $U = U_1 \parallel \cdots \parallel U_n$ such that for $1 \leq i \leq n$ there exists a $q_i' \in F$ with $q_i \arun{U_i}_A q_i'$.
    Since $A$ is parsimonious, we can assume without loss of generality that none of these runs is trivial.

    We then claim that, for $1 \leq i \leq n$, there exists a $\phi_i = \mset{q_{i,1}, \dots, q_{i,n_i}} \in \M(Q)$ such that $\mset{q_i} \blacktriangleleft \phi_i$, and $U_i = U_{i,1} \parallel \dots \parallel U_{i,n_i}$, such that for $1 \leq i \leq n_i$ we have that $\sst{q_{i,j}} \arun{U_{i,j}}_{A_3} \top$.
    Applying the induction hypothesis to each $q_i \arun{U_i}_A q_i' \in F$, there are two cases to consider:
    \begin{itemize}
        \item
        If $\sst{q_i} \arun{U_i}_{A_3} \top$, then we choose $n_i = 1$ and $q_{i,1} = q_i$ and $U_{i,1} = U_i$.

        \item
        If $\pst{q_i} \arun{U_i}_{A_3} \top$, then by construction of $A_3$ this must be a parallel unit run.
        Consequently, there exist $\sst{q_{i,1}}, \dots, \sst{q_{i,n_i}} \in \sst{Q}$ with $\top \in \gamma_3(\pst{q_i}, \mset{q_{i,1}, \dots, q_{i,n_i}})$, and $U_i = U_{i,1} \parallel \cdots \parallel U_{i,n_i}$ such that for $1 \leq j \leq n_i$ we have that $\sst{q_{i,j}} \arun{U_{i,j}}_{A_3} \top$.
        By definition of $\gamma_3$, we then obtain $\psi \in \M(Q)$ such that $\gamma(q_i, \psi) \cap F \neq \emptyset$ and $\psi \blacktriangleleft \mset{q_{i,1}, \dots, q_{i,n_i}}$.
        A straightforward inductive argument on the definition of $\blacktriangleleft$ shows that it is transitive; hence, since $\mset{q_i} \blacktriangleleft \psi$, we have that $\mset{q_i} \blacktriangleleft \mset{q_{i,1}, \dots, q_{i,n_i}}$.
    \end{itemize}

    \noindent
    Using the above, it follows that $\mset{q_1, \dots, q_n} \blacktriangleleft \mset{q_{1,1}, \dots, q_{n, n_n}}$.
    Hence,
    \[
        \sst{q}, \pst{q} \in \gamma_3(\sst{p}, \mset{\sst{q_{1,1}}, \dots, \sst{q_{n, n_n}}})
    \]
    Since $U = U_{1,1} \parallel \cdots \parallel U_{n_n}$, it follows that $\sst{p} \arun{U}_{A_3} \sst{q}$ and $\sst{p} \arun{U}_{A_3} \pst{q}$.

    Furthermore, if $q \in F$, then $\top \in \gamma_3(\pst{p}, \mset{\sst{q_{1,1}}, \dots, \sst{q_{n, n_n}}})$, and hence $\pst{p} \arun{U}_{A_3} \top$.
    \qedhere
\end{itemize}
\end{proof}

\noindent
We are now ready to show that our construction preserves languages.
More specifically, \cref{fact:flat-branching-correctness-backward,fact:flat-branching-correctness-forward} together imply that for $p \in Q$, we have
\[
    L_{A_1}(q) =
        L_{A_3}(\pst q) \cup
        L_{A_3}(\sst{q}) \cup
        \begin{cases}
        L_{A_3}(\top) & q \in F \\
        \emptyset & \text{otherwise}
        \end{cases}
\]

To see that our construction preserves fork-acyclicity and that $A_3$ is bounded, one can show that if $p, q \in Q$ are such that $\sst{p} \preceq_{A_3} \sst{q}$, $\sst{p} \preceq_{A_3} \pst{q}$, $\pst{p} \preceq_{A_3} \sst{q}$ or $\pst{p} \preceq_{A_3} \pst{q}$, then $p \preceq_{A_3} q$.
A fork cycle in $A_3$ thus gives rise to a fork cycle in $A$, which means that if $A$ is fork-acyclic, then so is $A_3$.
Furthermore, the support of a state $\sst{q}$ or $\pst{q}$ in $A_3$ is contained in $\setcompr{\sst{p}, \pst{p}}{p \in \pi_A(q)} \cup \set{\top}$; since the latter is finite as $A$ is bounded, it follows that $A_3$ must also be bounded.
\end{proof}

\section{Lemmas about the syntactic pomset automaton}

\subsection{Deconstruction lemmas}

\restaterundeconstructsequential*
\begin{proof}
We proceed by induction on the length $\ell$ of $e_1 \cdot e_2 \sarun{U} f$.
In the base, where $\ell = 0$, we have $f = e_1 \cdot e_2$ (hence $e_1, e_1 \in \sacc$) and $U = 1$.
We can then choose $f_1 = e_1$ and $f_2 = e_2$ as well as $U_1 = U_2 = 1$, to find that $e_1 \sarun{U_1} f_1$ and $e_2 \sarun{U_2} f_1$, of length zero.

For the inductive step, let $e_1 \cdot e_2 \sarun{U} f$ be of length $\ell+1$.
We find that $U = U_0 \cdot U'$, and a $g \in \terms$ such that $e_1 \cdot e_2 \sarun{U_0} g$ is a unit run, and $g \sarun{U'} f$ is of length $\ell$.
If $e_1 \cdot e_2 \sarun{U_0} g$ is a sequential unit run, then $U_0 = \ltr{a}$ for some $\ltr{a} \in \Sigma$, and $g \in \sderiv(e_1 \cdot e_2, \ltr{a}) = \sderiv(e_1, \ltr{a}) \fatsemi e_2 \cup e_1 \star \sderiv(e_2, \ltr{a})$.
This gives us two cases to consider.
\begin{itemize}
    \item
    If $g \in \sderiv(e_1, \ltr{a}) \fatsemi e_2$, then $g = g_1 \cdot e_2$ such that $g_1 \in \sderiv(e_1, \ltr{a})$.
	By induction we find $f_1, f_2 \in \sacc$ and $U' = U_1' \cdot U_2'$ such that $g_1 \sarun{U_1'} f_1$, and $e_2 \sarun{U_2'} f_2$, of length at most $\ell$.
    We choose $U_1 = U_0 \cdot U_1'$ and $U_2 = U_2'$ to find $U = U_0 \cdot U' = U_0 \cdot U_1' \cdot U_2' = U_1 \cdot U_2$, as well as $e_1 \sarun{U_1} f_1$ of length at most $\ell + 1$, and $e_2 \sarun{U_2} f_2$ of length at most $\ell$.

    \item
    If $g \in e_1 \star \sderiv(e_2, \ltr{a})$, then first note that $e_1 \in \sacc$, and $g \in \sderiv(e_2, \ltr{a})$.
    We choose $U_1 = 1$ and $U_2 = U$ as well as $f_1 = e_1$ and $f_2 = f'$ to find that $U = U_1 \cdot U_2$ as well as $e_1 \sarun{U_1} f_1$ of length zero.
    Lastly, $e_2 \sarun{U_0} g \sarun{U'} f' = f_2$, meaning $e_2 \sarun{U} f_2$ of length at most $\ell + 1$.
\end{itemize}
The case where $e_1 \cdot e_2 \sarun{U_0} g$ is a parallel unit run can be treated similarly.
\end{proof}

\restaterundeconstructparallel*
\begin{proof}
If $e_1 \parallel e_2 \sarun{U} f$ is trivial, then $U = 1$ and $e_1 \parallel e_2 = f \in \sacc$.
Hence, $e_1, e_2 \in \sacc$; we can choose $f_1 = e_1$, $f_2 = e_2$ and $U_1 = U_2 = 1$ to satisfy the claim.

Otherwise, there exist $U_0, U' \in \SP(\Sigma)$ and $g \in \terms$ such that $U = V \cdot W$ and $e_1 \parallel e_2 \sarun{V} g$ is a unit run, and $g \sarun{W} f$.
We can discount the possibility that $e_1 \parallel \sarun{V} g$ is a sequential unit run, because $\sderiv(e_1 \parallel e_2, \ltr{a}) = \emptyset$ for all $\ltr{a} \in \Sigma$.
Hence, $e_1 \parallel e_2 \sarun{V} g$ is a parallel unit run, meaning that $V = V_1 \parallel \cdots \parallel V_n$ and there exists a $\phi = \mset{h_1, \dots, h_n} \in \M(\terms)$ such that $g \in \pderiv(e_1 \parallel e_1, \phi)$, and for $1 \leq i \leq n$ there exists an $h_i' \in \sacc$ with $h_i \sarun{V_i} h_i'$.
By definition of $\pderiv$, it then follows that $n = 2$ and $g = 1$ as well as (without loss of generality) $e_1 = h_1$ and $e_2 = h_2$.
Since $g = 1$, it must be that $g \sarun{W} f$ is trivial, and hence $W = 1$, meaning that $U = V$.
We choose $f_1 = h_1'$, $f_2 = h_2'$, $U_1 = V_1$ and $U_2 = V_2$ to satisfy the claim.
\end{proof}

\subsection{Construction lemmas}

\restaterunconstructsequential*
\begin{proof}
The proof consists of two phases; first, we verify the following.

\begin{fact}%
\label{fact:run-construct-carry}
We have that $e_1 \cdot e_2 \sarun{U} f_1 \cdot e_2$.
\end{fact}
\begin{proof}
The proof proceeds by induction on the length $\ell$ of $e_1 \sarun{U} f_1$.
In the base, where $\ell = 0$ and $f_1 = e_1$ as well as $U = 1$, we the claim holds immediately.

In the inductive step, let $e_1 \sarun{U} f_1$ be of length $\ell + 1$.
We find $e_1' \in \terms$ and $U = U_0 \cdot U'$ such that $e_1 \sarun{U_0} e_1'$ is a unit run, and $e_1' \sarun{U'} f_1$ is of length $\ell$.
By induction, $e_1' \cdot e_2 \sarun{U'} f_1 \cdot e_2$.
If $e_1 \sarun{U_0} e_1'$ is a sequential unit run, then $U_0 = \ltr{a}$ for some $\ltr{a} \in \Sigma$, and $e_1' \in \sderiv(e_1, \ltr{a})$, meaning $e_1' \cdot e_2 \in \sderiv(e_1 \cdot e_2, \ltr{a})$, hence $e_1 \cdot e_2 \sarun{U_0} e_1' \cdot e_2$.
We conclude that $e_1 \cdot e_2 \sarun{U} f_1 \cdot e_2$.

  The case where $e_1 \sarun{U_0} e_1'$ is a parallel unit run is similar.
\end{proof}

Next, we note the following.

\begin{fact}%
\label{fact:run-construct-ledge}
There exists an $f \in \sacc$ such that $f_1 \cdot e_2 \sarun{V} f$.
\end{fact}
\begin{proof}
  There are two cases to consider, based on the length of $e_2 \sarun{V} f_1$.
  \begin{itemize}
  \item
    If $\ell = 0$, we know that $f_2 = e_2$ and $V = 1$.
    We can then choose $f = f_1 \cdot e_2$.

  \item
    In the inductive step, let $e_2 \sarun{V} f_1$ be of length $\ell + 1$.
    We find $e_2' \in \terms$ and $V = V_0 \cdot V'$ such that $e_2 \sarun{V_0} e_2'$ is a unit run, and $e_2' \sarun{V'} f_2$ is of length $\ell$.
    If $e_2 \sarun{V_0} e_2'$ is a sequential unit run, then $V_0 = \ltr{a}$ for some $\ltr{a} \in \Sigma$, and $e_2' \in \sderiv(e_2, \ltr{a})$, and thus $e_2' \in \sderiv(f_1 \cdot e_2, \ltr{a})$.
    Hence, $f_1 \cdot e_2 \sarun{V_0} e_2'$, meaning $f_1 \cdot e_2 \sarun{V} f_2$; choosing $f = f_2$ satisfies the claim.

    The case where $e_2 \sarun{V_0} e_2'$ is a parallel unit run is similar.
    \qedhere
  \end{itemize}
\end{proof}

Putting these together, we find $f \in \sacc$ such that $e_1 \cdot e_2 \sarun{U \cdot V} f$.
\end{proof}

\restaterunconstructparallel*
\begin{proof}
Since $1 \in \pderiv(e_1 \parallel e_2, \mset{e_1, e_2})$, the claim follows immediately.
\end{proof}

\subsection{Correctness of the syntactic PA}

\restatebrzozowskimorphism*
\begin{proof}
  We treat the claims case-by-case.
  \begin{itemize}
  \item
    To show $L_\Sigma(0) = \emptyset$, suppose that $U \in L_\Sigma(0)$.
    In that case, $0 \sarun{U} e$ for some $e \in \sacc$.
    Since $0 \not\in \sacc$, this means that $0 \sarun{U} e$ cannot be trivial.
    In that case, there exists an $e' \in \terms$ such that $0 \sarun{U} e'$ is a unit run.
    However, this contradicts that $\sderiv(0, \ltr{a}) = \emptyset$ for all $\ltr{a} \in \Sigma$, and $\pderiv(e, \phi) = \emptyset$ for all $\phi \in \M(\terms)$.
    Therefore, our assumption that $U \in L_\Sigma(0)$ must be false.
    We conclude that $L_\Sigma(0) = \emptyset$.

  \item
    To show $L_\Sigma(1) = \set{1}$, suppose that $U \in L_\Sigma(1)$, i.e., $1 \sarun{U} e$ for some $e \in \sacc$.
    By an argument similar to the previous case, we can argue that $1 \sarun{U} e$ is trivial, and hence $U = 1$.
    The other inclusion follows from the fact that $1 \in \sacc$ and $1 \sarun{1} 1$.

  \item
    To show $L_\Sigma(\ltr{a}) = \set{\ltr{a}}$, suppose that $U \in L_\Sigma(\ltr{a})$, i.e., $\ltr{a} \sarun{U} e$ for some $e \in \sacc$.
    Since $\ltr{a} \not\in \sacc$, we know $\ltr{a} \sarun{U} e$ must be non-trivial.
    This means that we can write $U = U_0 \cdot U'$, and there exists an $f \in \terms$ such that $\ltr{a} \sarun{U_0} f$ is a unit run, and $f \sarun{U'} e$.
    A quick glance at $\sderiv$ and $\pderiv$ then tells us that $f = 1$.
    By the previous case, we know that $f \sarun{U'} e$ must be trivial; hence $U' = 1$ and $f = e$.
    Indeed, $\ltr{a} \sarun{U_0} f$ must be a sequential unit run, for $\pderiv(\ltr{a}, \phi) = \emptyset$ for all $\phi \in \M(\terms)$.
    This tells us that $U = \ltr{b}$ and $f \in \sderiv(\ltr{a}, \ltr{b})$ for some $\ltr{b} \in \Sigma$; by definition of $\sderiv$, it follows that $\ltr{b} = \ltr{a}$.

    For the other inclusion, let $U = \ltr{a}$; then $\ltr{a} \sarun{\ltr{a}} 1$ immediately, and hence $\ltr{a} \in L_\Sigma(\ltr{a})$.

  \item
    To show $L_\Sigma(e + f) = L_\Sigma(e) \cup L_\Sigma(f)$, suppose $U \in L_\Sigma(e + f)$, i.e., $e + f \sarun{U} g$ for $g \in \sacc$.
    By \cref{lemma:run-deconstruct-plus}, we find $g' \in \sacc$ with $e \sarun{U} g'$ or $f \sarun{U} g'$, and hence $U \in L_\Sigma(e) \cup L_\Sigma(f)$.

    For the other inclusion, suppose that $U \in L_\Sigma(e)$.
    We then have that $e \sarun{U} g$ for some $g \in \sacc$.
    By \cref{lemma:run-construct-plus}, there exists a $g' \in \sacc$ such that $e + f \sarun{U} g'$, and hence $U \in L_\Sigma(e+f)$.
    The case where $U \in L_\Sigma(f)$ can be treated similarly.

  \item
    To show $L_\Sigma(e \cdot f) = L_\Sigma(e) \cdot L_\Sigma(f)$, suppose that $U \in L_\Sigma(e \cdot f)$, i.e., $e \cdot f \sarun{U} g$ for some $g \in \sacc$.
    By \cref{lemma:run-deconstruct-sequential}, we find $g_0, g_1 \in \sacc$ such that $U = U_0 \cdot U_1$ as well as $e \sarun{U_0} g_0$ and $f \sarun{U_1} g_1$.
    This means that $U_0 \in L_\Sigma(e)$ and $U_1 \in L_\Sigma(f)$, and thus $U \in L_\Sigma(e) \cdot L_\Sigma(f)$.

    If $U \in L_\Sigma(e) \cdot L_\Sigma(f)$, then $U = U_0 \cdot U_1$ such that $U_0 \in L_\Sigma(e)$ and $U_1 \in L_\Sigma(f)$.
    This means that there exist $g_0, g_1 \in \sacc$ such that $e \sarun{U_0} g_0$ and $f \sarun{U} g_1$.
    By \cref{lemma:run-construct-sequential}, there exists a $g \in \sacc$ such that $e \cdot f \sarun{U} g$, and hence $U \in L_\Sigma(e \cdot f)$.

  \item
    To show $L_\Sigma(e \parallel f) = L_\Sigma(e) \parallel L_\Sigma(f)$, suppose $U \in L_\Sigma(e \parallel f)$, i.e., $e \parallel f \sarun{U} g$ for some $g \in \sacc$.
    By \cref{lemma:run-deconstruct-parallel}, we find $g_1, g_2 \in \sacc$ and $U_1, U_2 \in \SP(\Sigma)$ such that $U = U_1 \parallel U_2$ as well as $e_1 \sarun{U_1} g_1$ and $e_2 \sarun{U_2} g_2$.
    It then follows that $U = U_1 \parallel U_2 \in L_\Sigma(e) \parallel L_\Sigma(f)$.

    If $U \in L_\Sigma(e) \parallel L_\Sigma(f)$, then $U = U_1 \parallel U_2$ such that $U_1 \in L_\Sigma(e)$ and $U_2 \in L_\Sigma(f)$.
    This means that there exist $g_1, g_2 \in \sacc$ such that $e \sarun{U_1} g_1$ and $f \sarun{U_2} g_2$.
    By \cref{lemma:run-construct-parallel}, we find that $e \parallel f \sarun{U} 1 \in \sacc$, and hence $U \in L_\Sigma(e \parallel f)$.

  \item
    To show $L_\Sigma(e^*) = {L_\Sigma(e)}^*$, suppose $U \in L_\Sigma(e^*)$, i.e., $e^* \sarun{U} f$ for $f \in \sacc$.
    By \cref{lemma:run-deconstruct-star}, we find that $U = U_1 \cdots U_n$ and $f_1, \dots, f_n \in \sacc$ such that for $1 \leq i \leq n$ we have $e \sarun{U_i} f_i$.
    Hence, we know for $1 \leq i \leq n$ that $U_i \in L_\Sigma(e)$, and therefore $U = U_1 \cdots U_n \in {L_\Sigma(e)}^*$.

    For the other direction, let $U \in {L_\Sigma(e)}^*$.
    Then we can write $U = U_1 \cdots U_n$ such that for $1 \leq i \leq n$ it holds that $U_i \in L_\Sigma(e)$.
    We find for $1 \leq i \leq n$ an $f_i \in \sacc$ such that $e \sarun{U_i} f_i$.
    By \cref{lemma:run-construct-star}, we find an $f \in \sacc$ such that $e^* \sarun{U} f$, and hence $U \in L_\Sigma(e^*)$.
    \qedhere
  \end{itemize}
\end{proof}

\restatebrzozowskisoundness*
\begin{proof}
The proof proceeds by induction on the structure of $e$.
In the base, we need to consider the cases where $e \in \set{0, 1} \cup \Sigma$, all of which go through by \cref{lemma:brzozowski-morphism}.

For the inductive step, there are four cases.
We argue the case where $e = e_1 + e_2$; the other cases are similar.
Using \cref{lemma:brzozowski-morphism} and the induction hypothesis, we can derive that
\[
    L_\Sigma(e_1 + e_2)
        = L_\Sigma(e_1) \cup L_\Sigma(e_2)
        = \sem{e_1} \cup \sem{e_2}
        = \sem{e_1 + e_2}
    \qedhere
\]
\end{proof}

\restatesyntacticpaforkacyclic*
\begin{proof}
  For fork-acyclicity, we define the following.
  For $e \in \terms$, we define $\pdepth{e}$ as follows:
  \begin{align*}
    \pdepth{0} &= 0
    & \pdepth{e + f} &= \max(\pdepth{e}, \pdepth{f})
    & \pdepth{e^*} &= \pdepth{e} \\
    \pdepth{1} &= 0
    & \pdepth{e \parallel f} &= \max(\pdepth{e}, \pdepth{f}) + 1 \\
    \pdepth{\ltr{a}} &= 0
    & \pdepth{e \cdot f} &= \max(\pdepth{e}, \pdepth{f})
  \end{align*}
  We now claim that if $e \preceq_\Sigma f$, then $\pdepth{e} \leq \pdepth{f}$.
  To see this, it suffices to prove the claim for the rules that generate $\preceq_\Sigma$; this gives us three cases to consider.
  \begin{enumerate}[(i)]
  \item
    If $e \preceq_\Sigma f$ because there exists an $\ltr{a} \in \Sigma$ with $e \in \sderiv(f, \ltr{a})$, we proceed by induction on $f$.
    In the base, $f \in \Sigma$ and $e = 1$; but then $\pdepth{e} = 0 \leq \pdepth{f}$ immediately.

    \medskip\noindent
    For the inductive step, there are four cases to consider.
    \begin{itemize}
    \item
      If $f = f_1 + f_2$, then assume without loss of generality that $e \in \sderiv(f_1, \ltr{a})$.
      By induction, $\pdepth{e} \preceq_\Sigma \pdepth{f_1}$; since $\pdepth{f_1} \leq \pdepth{f}$, the claim follows.

    \item
      If $f = f_1 \cdot f_1$, then there are two subcases to consider.
      \begin{itemize}
      \item
        If $e \in \sderiv(f_1, \ltr{a}) \fatsemi f_2$, then $e = f_1' \cdot f_2$ with $f_1' \in \sderiv(f_1, \ltr{a})$.
        By induction, $\pdepth{f_1'} \leq \pdepth{f_1}$.
        We then know that
        \[
          \pdepth{e} = \max(\pdepth{f_1'}, \pdepth{f_2}) \leq \max(\pdepth{f_1}, \pdepth{f_2}) = \pdepth{f}
        \]

      \item
        If $e \in f_1 \star \sderiv(f_2, \ltr{a})$, then $e \in \sderiv(f_2, \ltr{a})$.
        By induction, $\pdepth{e} \leq \pdepth{f_2}$.
        Since $\pdepth{f_2} \leq \pdepth{f}$, the claim follows.
      \end{itemize}

    \item
      We can disregard the case where $f = f_1 \parallel f_2$, for $\sderiv(f, \ltr{a}) = \emptyset$.

    \item
      If $f = f_1^*$, then $e = f_1' \cdot f_1^*$ with $f_1' \in \sderiv(f_1, \ltr{a})$.
      By induction, $\pdepth{f_1'} \leq \pdepth{f_1}$.
      We then know that $\pdepth{e} = \max(\pdepth{f_1'}, \pdepth{f_1}) \leq \pdepth{f_1} = \pdepth{f}$.
    \end{itemize}

  \item
    If $e \preceq_\Sigma f$ because there exists a $\phi \in \M(\terms)$ with $e \in \pderiv(f, \phi)$, we proceed by induction on $f$.
    In the base, where $f \in \set{0, 1} \cup \Sigma$, the claim holds vacuously, because $\pderiv(f, \phi) = \emptyset$.

    \medskip\noindent
    For the inductive step, all cases except the one for parallel composition are similar to the argument above.
    Now, if $f = f_1 \parallel f_2$, then $e = 1$, and hence $\pdepth{e} = 0 \leq \pdepth{f}$.

  \item\label{property:preceq-vs-depth-strict}
    If $e \preceq_\Sigma f$ because there exists a $\phi \in \M(\terms)$ with $e \in \phi$ and $\pderiv(f, \phi) \neq \emptyset$, we proceed by induction on $f$, showing $\pdepth{e} < \pdepth{f}$.
    In the base, the claim holds vacuously.

    \medskip\noindent
    For the inductive step, there are four cases to consider.
    \begin{itemize}
    \item
      If $f = f_1 + f_2$, then assume without loss of generality that $\pderiv(f_1, \phi) \neq \emptyset$.
      By induction, we have $\pdepth{e} < \pdepth{f_1}$.
      Since $\pdepth{f_1} \leq \pdepth{f}$, we are done.

    \item
      If $f = f_1 \cdot f_2$, then there are two subcases to consider.
      \begin{itemize}
      \item
        If $\pderiv(f_1, \phi) \fatsemi f_2 \neq \emptyset$, then $\pderiv(f_1, \phi) \neq \emptyset$.
        By induction, we have that $\pdepth{e} < \pdepth{f_1}$.
        Since $\pdepth{f_1} \leq \pdepth{f}$, we are done.

      \item
        If $f_1 \star \pderiv(f_2, \phi) \neq \emptyset$, then $\pderiv(f_2, \phi) \neq \emptyset$.
        By induction, we have that $\pdepth{e} < \pdepth{f_1}$.
        Since $\pdepth{f_1} \leq \pdepth{f}$, we are done.
      \end{itemize}

    \item
      If $f = f_1 \parallel f_2$, then without loss of generality $\phi = \mset{f_1, f_2}$ and $e = f_1$.
      By definition of $\pdepth$, we then find that $\pdepth{e} = \pdepth{f_1} < \pdepth{f}$.

    \item
      If $f = f_1^*$, then $\pderiv(f_1, \phi) \neq \emptyset$.
      By induction, $\pdepth{e} < \pdepth{f_1}$.
      Since $\pdepth{f_1} = \pdepth{f}$, we are done.
    \end{itemize}
  \end{enumerate}
  Now, if $e, f \in \terms$ such that there exists a $\phi \in \terms$ with $\pderiv(e, \phi) \neq \emptyset$ and $f \in \phi$, then first of all $f \preceq_\Sigma e$ by definition.
  Suppose towards a contradiction that $e \preceq_\Sigma f$; then $\pdepth{e} \leq \pdepth{f}$ by the above.
  However, we also know that $\pdepth{f} < \pdepth{e}$ by~\ref{property:preceq-vs-depth-strict} above --- a contradiction.
  We can therefore conclude that $e \not\preceq_\Sigma f$, and hence $f \prec_\Sigma e$.
\end{proof}

\restatesyntacticpabounded*
\begin{proof}
  We should show that for $e \in \terms$, it holds that $\ssupp(e)$ is finite.
  Since $\ssupp(e)$ is the smallest closed set that contains $e$, it suffices to find a finite closed set $S(e)$ such that $e \in S(e)$; since $\ssupp(e) \subseteq S(e)$ by definition, the claim then follows.
  In particular, to show that $S(e)$ is support-closed, it suffices to verify that for $f \in S(e)$ the following hold:
  \begin{itemize}
  \item
    For all $\ltr{a} \in \Sigma$, we have $\sderiv(e, \ltr{a}) \subseteq S(e)$.

  \item
    For all $\phi \in \M(\terms)$, we have $\pderiv(e, \phi) \subseteq S(e)$.

  \item
    If $e \in \phi \in \M(\terms)$ and $\pderiv(e, \phi) \neq \emptyset$, then $e \in S(e)$.
  \end{itemize}

  \noindent
  We proceed by induction on $e$.
  In the base, there are two cases to consider.
  \begin{itemize}
  \item
    If $e \in \set{0, 1}$, then we choose $S(e) = \set{e}$.
    Note that $\sderiv(e, \ltr{a}) = \emptyset$ for all $\ltr{a} \in \Sigma$, and $\pderiv(e, \phi) = \emptyset$ for all $\phi \in \M(\terms)$ --- hence, the three conditions above hold vacuously.

  \item
    If $e = \ltr{a}$ for some $\ltr{a} \in \Sigma$, then we choose $S(e) = \set{1, \ltr{a}}$.
    To see that $S(e)$ is support-closed, we inspect the derivatives of $\ltr{a}$; the derivatives for $1$ satisfy the right conditions by the previous case.
    First, for all $\ltr{b} \in \Sigma$, we have that $\sderiv(\ltr{a}, \ltr{b}) \subseteq \set{1}$, and hence $\sderiv(\ltr{b}, \ltr{a}) \subseteq S(e)$.
    Second, for all $\phi \in \M(\terms)$ we have that $\pderiv(\ltr{a}, \phi) = \emptyset$, and hence $\pderiv(\ltr{a}, \phi) \subseteq S(e)$.
    The case where $f \in \phi \in \M(\terms)$ with $\pderiv(\ltr{a}, \phi) \neq \emptyset$ cannot occur.
  \end{itemize}

  \noindent
  For the inductive step, there are four cases to consider.
  \begin{itemize}
  \item
    If $e = e_1 + e_2$, then we choose
    \[
      S(e) = \ssupp(e_1) \cup \ssupp(e_2) \cup \set{e}
    \]
    First, we note that $S(e)$ is finite by induction.
    To see that $S(e)$ is closed, it suffices to consider the derivatives of $\set{e}$, since $\ssupp(e_1)$ and $\ssupp(e_2)$ are closed by definition.
    \begin{itemize}
    \item
      For all $\ltr{a} \in \Sigma$, we have that $\sderiv(e_1 + e_2, \ltr{a}) = \sderiv(e_1, \ltr{a}) \cup \sderiv(e_2, \ltr{a})$.
      Now, since $\sderiv(e_1, \ltr{a}) \subseteq \ssupp(e_1)$ and $\sderiv(e_2, \ltr{a}) \subseteq \ssupp(e_2)$, we find that $\sderiv(e_1 + e_2, \ltr{a}) \subseteq S(e)$ as well.

    \item
      For all $\phi \in \M(\terms)$, we have that $\pderiv(e_1 + e_2, \phi) = \pderiv(e_1, \phi) \cup \pderiv(e_2, \phi)$.
      By an argument similar to the above, we find that $\pderiv(e_1 + e_2, \phi) \subseteq S(e)$.

    \item
      If $f \in \phi \in \M(\terms)$ such that $\pderiv(e_1 + e_2, \phi) \neq \emptyset$, then either $\pderiv(e_1, \phi) \neq \emptyset$ or $\pderiv(e_2, \phi) \neq \emptyset$.
      Hence, either $\pderiv(e_1, \phi) \neq \emptyset$ or $\pderiv(e_2, \phi) \neq \emptyset$, and hence $f \in \ssupp(e_1) \cup \ssupp(e_2) \subseteq S(e)$.
    \end{itemize}

  \item
    If $e = e_1 \cdot e_2$, then we choose
    \[
      S(e) = \ssupp(e_1) \fatsemi e_2 \cup e_1 \star \ssupp(e_2) \cup \ssupp(e_1)
    \]
    First, note that $S(e)$ is finite by induction.
    Also, $e \in S(e)$, because $e_1 \in \ssupp(e_1)$.
    To see that $S(e)$ is support-closed, it suffices to consider the elements of the first set above, since $e_1 \star \ssupp(e_2)$ and $\ssupp(e_1)$ are already support-closed.
    Let $e' = e_1' \cdot e_2$ for some $e_1' \in \ssupp(e_1)$.
    \begin{itemize}
    \item
      For all $\ltr{a} \in \Sigma$, we have that
      \[
        \sderiv(e_1' \cdot e_2, \ltr{a}) = \sderiv(e_1', \ltr{a}) \fatsemi e_1 \cup e_1' \star \sderiv(e_2, \ltr{a})
      \]
      Since $\sderiv(e_1', \ltr{a}) \cdot e_1 \subseteq \ssupp(e_1)$ and $\sderiv(e_2, \ltr{a}) \subseteq \ssupp(e_2)$, the claim then holds.

    \item
      For all $\phi \in \M(\terms)$, we can show that $\pderiv(e', \phi)$ again occurs in $S(e)$, by a similar argument.

    \item
      If $f \in \phi \in \M(\terms)$ such that $\pderiv(e', \phi) \neq \emptyset$, then $\pderiv(e_1', \phi) \fatsemi e_2 \neq \emptyset$, or $e_1' \star \pderiv(e_2, \phi) \neq \emptyset$.
      In the former case, $f \in \ssupp(e_1') \subseteq \ssupp(e_1) \subseteq S(e)$, while in the latter case $f \in \ssupp(e_2) \subseteq S(e)$.
    \end{itemize}

  \item
    If $e = e_1 \parallel e_2$, then we choose
    \[
      S(e) = \set{e_1 \parallel e_2, 1} \cup \ssupp(e_1) \cup \ssupp(e_2)
    \]
    We again have that $S(e)$ is finite by induction.
    To see that $S(e)$ is support-closed, it suffices to consider the derivatives of $e_1 \parallel e_2$.
    \begin{itemize}
    \item
      For all $\ltr{a} \in \Sigma$, we have that $\sderiv(e_1 \parallel e_2, \ltr{a}) = \emptyset \subseteq S(e)$.

    \item
      For $\phi \in \M(\terms)$, we have that $\pderiv(e_1 \parallel e_2, \phi) \subseteq \set{1} \subseteq S(e)$ by definition.

    \item
      For $f \in \phi \in \M(\terms)$ such that $\pderiv(e_1, \parallel e_2, \phi) \neq \emptyset$, we have that $\phi = \mset{e_1, e_2}$.
      In that case, $f \in \ssupp(e_1)$ or $f \in \ssupp(e_2)$.
    \end{itemize}

  \item
    If $e = e_1^*$, then we choose
    \[
      S(e) = \ssupp(e_1) \fatsemi e_1^* \cup \ssupp(e_1) \cup \set{e_1^*}
    \]
    First, we note that $S(e)$ is again finite by induction.
    To see that $S(e)$ is support-closed, it suffices to consider $\ssupp(e_1) \fatsemi e_1^* \cup \set{e_1^*}$.
    To this end, let $e' = e_1' \cdot e_1^*$ with $e_1' \in \ssupp(e_1)$.
    \begin{itemize}
    \item
      For all $\ltr{a} \in \Sigma$, we have that
      \begin{align*}
        \sderiv(e_1' \cdot e_1^*, \ltr{a})
        &= \sderiv(e_1', \ltr{a}) \fatsemi e_1^* \cup e_1 \star \sderiv(e_1^*, \ltr{a}) \\
        &\subseteq \sderiv(e_1', \ltr{a}) \fatsemi e_1^* \cup \sderiv(e_1^*, \ltr{a}) \\
        &= \sderiv(e_1', \ltr{a}) \fatsemi e_1^* \cup \sderiv(e_1, \ltr{a}) \fatsemi e_1^* \\
        &\subseteq \ssupp(e_1) \fatsemi e_1^* \subseteq S(e)
      \end{align*}
      Furthermore, $\sderiv(e_1^*, \ltr{a}) \subseteq S(e)$ by a similar argument.

    \item
      If $\phi \in \M(\terms)$, then $\pderiv(e_1' \cdot e_1^*, \phi) \subseteq S(e)$ and $\pderiv(e_1^*, \phi) \subseteq S(e)$ by a similar argument.

    \item
      If $f \in \phi \in \M(\terms)$ and $\pderiv(e_1' \cdot e_1^*, \phi) \neq \emptyset$, then $\pderiv(e_1', \phi) \neq \emptyset$, hence $f \in \ssupp(e_1') \subseteq \ssupp(e_1) \subseteq S(e)$.
      When $\pderiv(e_1^*, \phi) \neq \emptyset$, we have that $f \in S(e)$ by a similar argument.
      \qedhere
    \end{itemize}
  \end{itemize}
\end{proof}

\section{Lemmas about expressions to automata}

\restateautomatatoexpressionssmallstep*
\begin{proof}
For the implication from left to right, suppose that $U \in \sem{e_{pr}}$; we have two cases:
\begin{itemize}
    \item
    If $U = \ltr{a}$ such that $r \in \delta(p, \ltr{a})$, then $p \arun{U}_A r$ is a sequential unit run immediately.

    \item
    If $s_1, \dots, s_n \in Q$ with $U \in \sem{e_{s_1} \parallel \cdots \parallel e_{s_n}}$ and $r \in \gamma(p, \mset{s_1, \dots, s_n})$, then $U = U_1 \parallel \cdots \parallel U_n$ such that for all $1 \leq i \leq n$ we have $U_i \in \sem{e_{s_i}}$.
    By the main induction hypothesis we have for each $1 \leq i \leq n$ an $s_i' \in F$ with $s_i \arun{U_i}_A s_i'$.
    Hence $p \arun{U}_A r$ is a parallel unit run.
\end{itemize}

\noindent
For the converse implication, there are two cases to consider.
\begin{itemize}
    \item
    If $p \arun{U}_A r$ is a sequential unit run, then $U = \ltr{a}$ for some $\ltr{a} \in \Sigma$ and $r \in \delta(p, \ltr{a})$.
    In that case $U \in \sem{e_{pr}}$ by construction.

    \item
    If $p \arun{U}_A r$ is a parallel unit run, then $U = U_1 \parallel \cdots \parallel U_n$ and there exist $s_1, \dots, s_n \in Q$ and $s_1', \dots, s_n' \in F$ such that for $1 \leq i \leq n$ it holds that $s_i \arun{U_i}_A s_i'$, and furthermore $r \in \gamma(p, \mset{s_1, \dots, s_n})$.
    By the main induction hypothesis, we have that $U_i \in \sem{e_{s_i}}$ for $1 \leq i \leq n$, and thus $U \in \sem{U_1 \parallel \cdots \parallel U_n}$.
    We can then conclude that $U \in \sem{e_{pr}}$ again.
    \qedhere
\end{itemize}
\end{proof}

\end{document}